\documentclass[english]{article}
\usepackage[T1]{fontenc}
\usepackage[latin9]{inputenc}
\usepackage{babel}
\usepackage{array}
\usepackage{float}
\usepackage{multirow}
\usepackage{amsmath}
\usepackage{amsthm}
\usepackage{amssymb}
\usepackage{graphicx}
\usepackage[hidelinks]{hyperref}
\usepackage{placeins}
\usepackage[dvipsnames]{xcolor}
\makeatletter

\providecommand{\tabularnewline}{\\}

\theoremstyle{remark}
\newtheorem{rem}{\protect\remarkname}
\theoremstyle{plain}
\newtheorem{lem}{\protect\lemmaname}
\theoremstyle{plain}
\newtheorem{prop}{\protect\propositionname}

\usepackage{babel}
\usepackage{babel}

\providecommand{\lemmaname}{Lemma}
\providecommand{\propositionname}{Proposition}

\@ifundefined{showcaptionsetup}{}{%
	\PassOptionsToPackage{caption=false}{subfig}}
\usepackage{subfig}
\makeatother

\providecommand{\lemmaname}{Lemma}
\providecommand{\propositionname}{Proposition}
\providecommand{\remarkname}{Remark}

\begin{document}
	\title{ How to Compute Invariant Manifolds and their Reduced Dynamics in High-Dimensional Finite-Element Models}
	\author{Shobhit Jain\footnote{Corresponding author: shjain@ethz.ch}, George Haller}
	\date{\vspace{-5ex}
	}
	\maketitle
	\begin{center}
		Institute for Mechanical Systems, ETH Zürich \\
		Leonhardstrasse 21, 8092 Zürich, Switzerland 
		\par\end{center}
	\begin{abstract}
		Invariant manifolds are important constructs for the quantitative
		and qualitative understanding of nonlinear phenomena in dynamical systems. In nonlinear damped mechanical systems, for instance,
		spectral submanifolds have emerged as useful tools for the computation
		of forced response curves, backbone curves, detached resonance curves
		(\emph{isolas}) via exact reduced-order models. For conservative nonlinear
		mechanical systems, Lyapunov subcenter manifolds and their reduced dynamics provide a way
		to identify nonlinear amplitude-frequency relationships in the form
		of conservative backbone curves. Despite these powerful predictions
		offered by invariant manifolds, their use has largely been limited
		to low-dimensional academic examples. This is because several challenges
		render their computation unfeasible for realistic
		engineering structures described by finite-element models. In this
		work, we address these computational challenges and develop methods
		for computing invariant manifolds and their reduced dynamics in very high-dimensional
		nonlinear systems arising from spatial discretization of the governing
		partial differential equations. {We illustrate our computational algorithms on finite-element models of mechanical structures that range from a simple beam containing tens of degrees of freedom to an aircraft wing containing more than a hundred-thousand degrees of freedom.}
	\end{abstract}

	
	\section{Introduction}
	
	Invariant manifolds are low-dimensional
	surfaces in the phase space of a dynamical system that constitute organizing
	centers of nonlinear dynamics. These surfaces are composed of full
	system trajectories that stay on them for all times, partitioning
	the phase space locally into regions of different behavior. For instance, invariant
	manifolds attached to fixed points can be viewed as the nonlinear
	analogues of (flat) modal subspaces of the linearized
	system. Classic examples are the stable, unstable and center manifolds
	tangent to the stable, unstable and center subspaces of fixed points (see, e.g.,
	Guckenheimer \& Holmes~\cite{Guckenheimer1983}). The most important class of invariant manifolds are those that attract other trajectories and hence their own low-dimensional internal dynamics acts a mathematically exact reduced-order model for the full, high-dimensional system. The focus of this article is to compute such invariant manifolds and their reduced dynamics accurately and efficiently in very high-dimensional nonlinear systems. 
	
	The theory of invariant
	manifolds has matured over more than a century of research and has
	been applied to numerous fields for the qualitative understanding of 
	nonlinear behavior of systems (see Fenichel~\cite{Fenichel1971,Fenichel1974,Fenichel1977}, Hirsch et al.~\cite{Hirsch1977}, Wiggins~\cite{Wiggins1994}, Eldering~\cite{Eldering}, Nipp \& Stoffer~\cite{Nipp2013}). The computation of invariant manifolds,
	on the other hand, is a relatively new and rapidly evolving discipline
	due to advances in scientific computing. In this paper, we address
	some challenges that have been hindering the computation of invariant manifolds in
	high-dimensional mechanical systems arising from spatially discretized
	partial differential equations (PDEs).
	
	The methods for computing invariant manifolds can be divided
	into two categories: local and global. Local methods
	approximate an invariant manifold in the neighborhood of simpler invariant sets,
	such as fixed points, periodic orbits or invariant tori. Such local
	approximations are performed using Taylor series approximations around
	the fixed point or Taylor-Fourier series around the periodic orbit/invariant
	torus (see Simo~\cite{Simo1990}). Global methods, on the other hand,
	seek invariant manifolds globally in the phase space. Global techniques
	generally employ numerical \emph{continuation} for growing invariant
	manifolds from their local approximation that may be obtained from the linearized dynamics (see Krauskopf et al.~\cite{Krauskopf2005} for a survey). 
	
	\subsection*{Global techniques }
	
	A key aspect of most global techniques is to discretize the manifold into a mesh e.g., via a collocation or a spectral approach (see Dancowicz \& Schilder~\cite{Dankowicz}, Krauskopf et al.~\cite{Krauskopf2007}) and solve invariance equations for the unknowns
	at the mesh points. For growing an $M-$dimensional manifold via $q$ collocation/spectral points (along each of the $ M $ dimensions) in an $N-$dimensional dynamical system, one needs to solve a system of $\mathcal{O}\left(Nq^M\right)$ \emph{nonlinear} algebraic equations at each continuation step. This is achieved via
	an iterative solver such as the Newton's method. As $N$
	invariably becomes large in the case of discretized PDEs governing
	mechanics applications, numerical continuation of invariant manifolds
	via collocation and spectral approaches becomes computationally intractable.
	Indeed, while global approaches are often discussed for general systems,
	the most common applications of these approaches tend to involve
	low-dimensional problems, such as the computation of the Lorenz manifold (Krauskopf \& Osinga~\cite{Krauskopf2003}). 
	
	Global approaches also include the continuation of trajectory segments
	along the invariant manifold, expressed as a family of trajectories.
	This is achieved by formulating a two-point boundary value problem
	(BVP) satisfied by the trajectory on the manifold and numerically
	following a branch of solutions (see Krauskopf et al.~\cite{Krauskopf2005},
	Guckenheimer et al.~\cite{Guckenheimer2015}). While collocation and
	spectral methods are valid means to achieve this end as well, the
	(multiple) shooting method method (see Keller~\cite{Keller1968},
	Stoer \& Bulirsch~\cite{Stoer2002}) has a distinguishing appeal from
	a computational perspective for high-dimensional problems. In the
	(multiple) shooting method, an initial guess for one point on the
	solution trajectory is iteratively corrected such until the two-point
	BVP is solved up to a required precision. In each iteration, one performs
	numerical time integration of the full nonlinear system between the
	two points of the BVP, which is a computationally expensive undertaking
	for large systems. However, time integration involves the solution
	$\mathcal{O}\left(N\right)$ nonlinear algebraic equations at each
	time step, in contrast to collocation and spectral methods which require
	$\mathcal{O}\left(Nq\right)$ nonlinear algebraic equations to be
	solved simultaneously for $q$ collocation/spectral points. Coupled
	with advances in domain decomposition methods for time integration
	(see Carraro et al.~\cite{Carraro2015}), the multiple shooting method
	provides a feasible alternative to collocation and spectral methods.
	Still, covering a multi-dimensional invariant manifold using trajectory
	segments in a high-dimensional system is an elusive task even for multiple shooting methods (see, e.g., Jain et al.~\cite{Jain2019}). 
	
	%Auto~\cite{Auto} employs orthogonal
	%collocation to approximate solutions and is able to continue solution
	%families in two or three parameters. It is also capable of automatically
	%detecting and continuing solutions along bifurcating branches. As
	%it is written in a low-level programming language, Auto is computationally
	%advanced. However, it is difficult to use due to its rudimentary user
	%interface and the FORTRAN command structure. 
	
	A number of numerical continuation packages have enabled the computation
	of global invariant manifolds via collocation, spectral or multiple
	shooting methods. AUTO~\cite{Auto}, a FORTRAN-based package, constitutes
	the earliest organized effort towards continuation and bifurcation
	analysis of parameter-dependent ODEs. AUTO~\cite{Auto} employs orthogonal
	collocation to approximate solutions and is able to continue solution
	families in two or three parameters. The MATLAB-based package
	Matcont~\cite{Matcont} addresses some of the limitations of AUTO, albeit
	at a loss of computational performance. Additionally, Matcont can also
	perform normal form analysis. $\textsc{coco}$~\cite{Dankowicz}
	is an extensively documented and object-oriented MATLAB package which
	enables continuation via multidimensional atlas algorithms (Dankowicz et al.~\cite{Dankowicz2020}) and implements
	collocation as well as spectral methods. Another recent MATLAB package,
	NLvib~\cite{Krack2019}, implements the pseudo arc-length technique
	for the continuation of single-parameter family of periodic orbits
	in nonlinear mechanical systems via the shooting method or the spectral
	method (commonly referred to as harmonic balance in mechanics). Similar continuation
	packages are also available in the context of delay differential equations
	(see DDE-BIFTOOL~\cite{DDE_BIFTool_2001}, written in MATLAB and Knut~\cite{Knut_2013}, written in C++). The main focus of all these and
	other similarly useful packages is to implement automated
	continuation procedures, including demonstrations on low-dimensional
	examples, but not on the computational complexity of the operations
	as $N$ increases. As discussed above, the collocation/spectral/shooting
	techniques that are invariably employed in such packages limit their
	ability to compute invariant manifolds in high-dimensional mechanics
	problems, where system dimensionality $N$ varies from several thousands
	to millions. 
	
	\subsection*{Local techniques}
	
	In contrast to global techniques, local techniques for invariant manifold computations produce approximations valid in a neighborhood of a fixed point, periodic orbit or invariant
	torus. As a result, local
	techniques are generally unable to compute homoclinic
	and heteroclinic connections. Nonetheless, in engineering applications, local approximations of
	invariant manifolds often suffice for assessing the influence of nonlinearities
	on a well-understood linearized response. 
	
	Center manifold computations and their associated local bifurcation
	analysis via \emph{normal forms} (see Guckenheimer \& Holmes~\cite{Guckenheimer1983})
	are classic applications of local approximations where the manifold
	is expressed locally as a graph over the center subspace. For an $M-$dimensional
	manifold, this local graph is sought via an $M-$variate Taylor series
	where the coefficient of each monomial term is unknown. These unknown
	coefficients are determined by solving the invariance equations in
	a recursive manner at each polynomial order of approximation, i.e.,
	the solution at a lower order can be computed without the knowledge
	of the higher-order terms. The computational procedure simply involves
	the solution of a system of  $\mathcal{O}(N)$ \emph{linear} equations for each
	monomial in the Taylor expansion (see Simo~\cite{Simo1990} for flows,
	Fuming \& Küpper,~\cite{Fuming1994} for maps). Thus, in the computational
	context of high-dimensional problems, local techniques that employ
	Taylor series approximations exhibit far greater feasibility in comparison
	to global techniques that involve the continuation of collocation,
	spectral or shooting based solutions. 
	
	More recently, the parametrization method has emerged as a rigorous
	framework for the local analysis and computation of invariant manifolds
	of discrete and continuous time dynamical systems. This method was
	first developed in papers by Cabré, Fontich \& de la Llave~\cite{Cabre2003,Cabre2003a,Cabre2005}
	for invariant manifolds tangent to eigenspaces of fixed points of
	nonlinear mappings on Banach spaces, and then extended to
	whiskered tori by Haro \& de la Llave~\cite{Haro2006,Haro2006a,Haro2007}.
	We refer to the monograph by Haro et al.~\cite{Haro2016} for an overview
	of the results. An important feature of the parametrization method
	is that it does not require the manifold parametrization to be the
	graph of a function and hence allows for folds
	in the manifold. Furthermore, the method returns the dynamics on the
	invariant manifold along with its embedding. The formal computation
	can again be carried out via Taylor series expansions when the invariant
	manifold is attached to a fixed point and via Fourier-Taylor expansions
	when it is attached to an invariant torus perturbing
	from a fixed point (see, e.g., Mireles James~\cite{MirelesJames2015},
	Castelli et al.~\cite{Castelli2015}, Ponsioen et al.~\cite{Ponsieon2018,Ponsioen2020}).
	
	The main focus of the parametrization method has been on the computer-assisted
	proofs of existences and uniqueness of invariant manifolds, for which
	the dynamical system is conveniently diagonalized at the linear level.
	Furthermore, as discussed by Haro et al.~\cite{Haro2016}, this diagonalization
	allows a choice between different styles of parametrization for
	the reduced dynamics on the manifold, such as a normal form style,
	a graph style or a mixed style. Recent applications of the parametrization
	method include the computation of spectral submanifolds or SSMs (Haller \&
	Ponsioen~\cite{Haller2016}) and Lyapunov subcenter manifolds or LSMs (Kelley~\cite{Kelley1969}). For these manifolds the normal form parametrization style can
	be used to directly extract \emph{forced response curves} (FRC) and
	\emph{backbone} curves in nonlinear mechanical systems, as we will discuss in this paper (see Ponsioen
	et al.~\cite{Ponsieon2018,Ponsioen2020}, Breunung \& Haller~\cite{Breunung2018},
	Veraszto et al.~\cite{Verasto2020}).

	\subsection*{Our contributions}
	While helpful for proofs and expositions, the routinely performed diagonalization and the associated linear
	coordinate change in invariant manifold computations, render
	the parametrization method unfeasible to high-dimensional mechanics
	problems for two reasons. First, diagonalization
	involves the computation of all $N$ eigenvalues of an $N-$dimensional
	dynamical system and second, the nonlinear coefficients in physical
	coordinates exhibit an inherent sparsity in mechanics applications
	that is annihilated by the linear coordinate change associated to
	diagonalization. Both these factors lead to unmanageable computation
	times and memory requirements when $N$ becomes large, as we discuss
	in Section~\ref{subsec:Pitfalls-of-transformation} of this manuscript. 
	
	To address these issues, we develop here a new computational methodology for local approximations
	to invariant manifolds via the parametrization method. The key aspects
	making this methodology scalable to high-dimensional mechanics
	problems are the use of physical coordinates and
	just the minimum number of eigenvectors. In the
	autonomous setting, we seek to compute invariant manifolds attached
	to fixed points where we develop expressions for the Taylor
	series coefficients that determine the parametrization of the invariant
	manifold as well as its reduced dynamics in different styles of parametrization (see Section~\ref{sec:computation_autonomous}).
	We develop similar expressions in the non-autonomous periodic or quasiperiodic
	setting, where we seek to compute invariant manifolds or \emph{whiskers}
	attached to an invariant torus perturbed from a hyperbolic fixed point
	under the addition of small-amplitude non-autonomous terms. In this
	case, we seek to compute the coefficients in Fourier-Taylor series that parametrize
	the invariant manifold as well as its reduced dynamics in different parametrization styles (see Section
	\ref{sec:computations_non-autonomous}). Finally, we apply this methodology
	to high-dimensional examples arising from a finite-element discretization
	of structural mechanics problems, whose forced
	response curves we recover from a normal form style parametrization of SSMs (see Section~\ref{sec:Numerical-examples}).

	Related computational ideas have already been used in other contexts. For instance, Beyn \& Kleß~\cite{Beyn1998}
	performed similar Taylor series-based computations of invariant manifolds
	attached to fixed points in physical coordinates using master modes.
	Their work predates the parametrization method and does not involve
	the choice of reduced dynamics or normal forms. Recently, Carini et
	al.~\cite{Carini2015} focused on computing center manifolds of fixed
	points and their normal forms using only master modes in physical
	coordinates. While this an application of the parametrization method in the normal form style, they attribute their results to earlier related work by Coullet \& Spiegel~\cite{Coullet1983} and use these center manifolds for analyzing stability of flows around bifurcating parameter
	values. 
	
	{ More recently, Vizzaccaro et al.~\cite{Vizzaccaro2020} and Opreni et al.~\cite{Opreni2021} have computed normal
		forms on second-order, proportionally-damped
		mechanical systems with up to cubic nonlinearities and derived explicit
		expressions up to cubic order accuracy (see also Touz\'e et al.~\cite{Touze2021} for a review). This is a direct application of the parametrization method via a normal-form style parametrization to formally compute
		assumed invariant manifolds whose existence/uniqueness is a priori
		unclear (cf. Haller \& Ponsioen~\cite{Haller2016}). These results in~\cite{Vizzaccaro2020,Opreni2021} provide low-order approximations to SSMs~\cite{Haller2016}, whose computation up to arbitrarily high orders of accuracy has already been automated in prior work~\cite{Ponsieon2018,Ponsioen2020} for mechanical systems with diagonalized linear part. A major computational advance in the approach of Vizzaccaro et al.~\cite{Vizzaccaro2020} is the nonintrusive use of finite element software to compute normal form coefficients up to cubic order. All these prior results, however, are fundamentally developed for unforced (non-autonomous) systems. 
		
		The computation procedure we develop is generally applicable
		to first-order systems with smooth nonlinearities, periodic or quasiperiodic forcing, and enables automated computation of various types of invariant manifolds such as stable-, unstable-, and center manifolds, LSMs, and SSMs, up to arbitrarily high orders of accuracy in physical coordinates.} Finally, a numerical implementation of these computational techniques is available
	in the form of an open-source MATLAB package, SSMTool 2.0~\cite{SSMTool2.0}, which
	is integrated with a generic finite-element solver (Jain et al.~\cite{FECode})
	for mechanics problems. We describe some key symbols and the notation
	used in the remainder of this paper in Table~\ref{tab:Notation} before
	proceeding towards the technical setup in the next section. 
	
	\FloatBarrier
	\section{General setup}	
	\begin{table}[H]
		\centering{}%
		
		\begin{tabular}{r>{\raggedright}p{10cm}}
			\hline 
			Symbol & Meaning\tabularnewline
			\hline 
			%		\hline 
			$n$ & Dimensionality of full second-order mechanical system~\eqref{eq:second-order}\tabularnewline
			%		\hline 
			$N$ & Dimensionality of the full first-order system~\eqref{eq:first_order} in the first-order
			($N=2n$ for mechanical systems)\tabularnewline
			%		\hline 
			$E$ & Master (spectral) subspace of a fixed point of system~\eqref{eq:first_order}  \tabularnewline
			%		\hline 
			$\mathcal{W}(E)$ & Invariant manifold tangent to $E$ at its fixed point\tabularnewline
			
			$M$ & $\dim(E)=\dim(\mathcal{W}(E))$: Dimensionality of the invariant manifold 
			constructed around $E$\tabularnewline
			
			$\mathbf{W}:\mathbb{C}^M\to \mathbb{R}^N$ & Parametrization for the invariant manifold $\mathcal{W}(E)$ \tabularnewline
			
			$\mathbf{R}:\mathbb{C}^M\to \mathbb{C}^M$ & Parametrization for the reduced dynamics on $\mathcal{W}(E)$ \tabularnewline
			
			%		\hline 
			$\mathbf{p}\in\mathbb{C}^{M}$ & Parametrization coordinates describing reduced dynamics: $\dot{\mathbf{p}}=\mathbf{R}(\mathbf{p},t)$\tabularnewline
			%		\hline 
			$(\bullet)^{\top}$ & Transpose of a matrix or vector\tabularnewline
			%		\hline 
			$\bar{(\bullet)}$ & Complex conjugation operation\tabularnewline
			%		\hline 
			$(\bullet)^{\star}$ & Complex conjugate transpose for a matrix or vector; adjoint for an
			operator.\tabularnewline
			%		\hline 
			$K$ & Number of rationally incommensurate forcing frequencies in system~\eqref{eq:second-order} or \eqref{eq:first_order} \tabularnewline
			%		\hline 
			$\boldsymbol{\Omega}\in\mathbb{R}_{+}^{K}$ & Quasiperiodic forcing frequencies\tabularnewline
			%		\hline 
			$\gamma_{\epsilon}$ & small amplitude invariant torus of system~\eqref{eq:first_order}
			\tabularnewline
			$\mathcal{W}(E, \gamma_{\epsilon})$ & Invariant manifold (whisker) of torus $\gamma_{\epsilon}$ perturbed from spectral subspace $E$.  \tabularnewline
			
			$\mathbf{W}_{\epsilon}:\mathbb{C}^{M}\times\mathbb{T}^{K}\to\mathbb{R}^{N}$ & Parametrization for the whisker $\mathcal{W}(E, \gamma_{\epsilon})$ \tabularnewline
			
			$\mathbf{R}_{\epsilon}:\mathbb{C}^{M}\times\mathbb{T}^{K}\to\mathbb{C}^{M}$ & Parametrization for the reduced dynamics on $\mathcal{W}(E, \gamma_{\epsilon})$ \tabularnewline
			
			$\mathbf{e}_{j}$ & $[0,\dots,0,\overset{\overset{j-\mathrm{th\,position}}{\uparrow}}{1},0,\dots,0]^{\top}$
			: vector aligned along the $j-$th coordinate axis in Euclidean space\tabularnewline
			%		\hline 
			$\mathrm{vec}(\bullet)$ & Vectorization operation\tabularnewline
			%		\hline 
			$\Delta_{i,M}$ & An ordered set, $\{\boldsymbol{\ell}_{1},\dots,\boldsymbol{\ell}_{M^{i}}\in\{1,\dots,M\}^{i}\}$, which contains all possible $i-$tuples drawn from the set $\{1,\dots,M\}$
			for any $i\in\mathbb{N}$ \tabularnewline
			%		\hline 
			$\mathrm{i}$ & imaginary unit\tabularnewline
			%		\hline 
			$\otimes$ & Kronecker product\tabularnewline
			\hline 
		\end{tabular}\caption{\label{tab:Notation}Notation}
	\end{table}

	We are mainly interested here in dynamical systems arising from mechanics
	problems. Such problems are governed by PDEs
	that are spatially discretized typically via the finite element method.
	The discretization results in a system of second-order ordinary differential
	equations for the generalized displacement $\mathbf{x}(t)\in\mathbb{R}^{n}$,
	which can be written as 
	\begin{equation}
	\mathbf{M}\ddot{\mathbf{x}}+\mathbf{C}\dot{\mathbf{x}}+\mathbf{K}\mathbf{x}+\mathbf{f}(\mathbf{x},\dot{\mathbf{x}})=\epsilon\mathbf{f}^{ext}(\mathbf{x},\dot{\mathbf{x}},\mathbf{\Omega}t).\label{eq:second-order}
	\end{equation}
	Here $\mathbf{M},\mathbf{C},\mathbf{K}\in\mathbb{R}^{n\times n}$
	are the mass, stiffness and damping matrices; $\mathbf{f}(\mathbf{x},\dot{\mathbf{x}})\in\mathbb{R}^{n}$
	is the purely nonlinear internal force; $\mathbf{f}^{ext}(\mathbf{x},\dot{\mathbf{x}},\mathbf{\Omega}t)\in\mathbb{R}^{n}$
	denotes the (possibly linear) external forcing with frequency vector
	$\mathbf{\Omega}\in\mathbb{R}^{K}$ for some $K\ge0$. The function
	$\mathbf{f}^{ext}$ is autonomous for $K=0$, periodic in $t$ for
	$K=1$, and quasi-periodic for $K>1$ with $ K $ rationally incommensurate frequencies. The second-order system~\eqref{eq:second-order}
	may be expressed in a first-order form as 
	\begin{align}
	\mathbf{B}\dot{\mathbf{z}} & =\mathbf{Az}+\mathbf{F}(\mathbf{z})+\epsilon\mathbf{F}^{ext}(\mathbf{z},\boldsymbol{\phi}),\label{eq:first_order}\\
	\dot{\boldsymbol{\phi}} & =\boldsymbol{\Omega},
	\end{align}
	where $\mathbf{z}=\left[\begin{array}{c}
	\mathbf{x}\\
	\dot{\mathbf{x}}
	\end{array}\right]$, $\mathbf{A},\mathbf{B}\in\mathbb{R}^{2n\times2n},\text{\ensuremath{\mathbf{F}:\mathbb{R}^{2n}\to\mathbb{R}^{2n}}},\mathbf{F}^{ext}:\mathbb{R}^{2n}\times\mathbb{T}^{K}\to\mathbb{R}^{2n}$
	denote the first-order quantities derived from system~\eqref{eq:second-order}.
	Such a first-order conversion is not unique: two equivalent choices
	are given by (see Tisseur \& Meerbergen~\cite{Tisseur2001}) {\footnotesize{}
		\begin{equation}
		(L1):\quad\mathbf{A}=\left[\begin{array}{cc}
		\mathbf{0} & \mathbf{\mathbf{N}}\\
		\mathbf{-\mathbf{K}} & -\mathbf{C}
		\end{array}\right],\quad\mathbf{B}=\left[\begin{array}{cc}
		\mathbf{N} & \mathbf{0}\\
		\mathbf{0} & \mathbf{M}
		\end{array}\right],\quad\mathbf{F}(\mathbf{z})=\left[\begin{array}{c}
		\mathbf{0}\\
		-\mathbf{f}(\mathbf{x},\dot{\mathbf{x}})
		\end{array}\right],\quad\mathbf{F}^{ext}(\mathbf{z},\boldsymbol{\phi})=\left[\begin{array}{c}
		\mathbf{0}\\
		\mathbf{f}^{ext}(\mathbf{x},\dot{\mathbf{x}},\boldsymbol{\phi})
		\end{array}\right],\label{eq:L1}
		\end{equation}
	}{\footnotesize\par}
	
	{\footnotesize{}
		\begin{equation}
		(L2):\quad\mathbf{A}=\left[\begin{array}{cc}
		-\mathbf{K} & \mathbf{0}\\
		\mathbf{0} & \mathbf{N}
		\end{array}\right],\quad\mathbf{B}=\left[\begin{array}{cc}
		\mathbf{C} & \mathbf{M}\\
		\mathbf{N} & \mathbf{0}
		\end{array}\right],\quad\mathbf{F}(\mathbf{z})=\left[\begin{array}{c}
		\mathbf{-\mathbf{f}(\mathbf{x},\dot{\mathbf{x}})}\\
		\mathbf{0}
		\end{array}\right],\quad\mathbf{F}^{ext}(\mathbf{z},\boldsymbol{\phi})=\left[\begin{array}{c}
		\mathbf{f}^{ext}(\mathbf{x},\dot{\mathbf{x}},\boldsymbol{\phi})\\
		\mathbf{0}
		\end{array}\right],\label{eq:L2}
		\end{equation}
	}where $\mathbf{N}\in\mathbb{R}^{n\times n}$ may be chosen as any
	non-singular matrix. If the matrices $\mathbf{M},\mathbf{C},\mathbf{K}$
	are symmetric, then the choice of $\mathbf{N=-K}$ for $(L1)$ and
	$\mathbf{N}=\mathbf{M}$ for $(L2)$ results in the first-order matrices
	$\mathbf{A},\mathbf{B}$ being symmetric. The computation methodology
	we will discuss is for any first-order system of the form~\eqref{eq:first_order}
	for $\mathbf{z}\in\mathbb{R}^{N}$. In particular, we have $N=2n$
	for second-order mechanical systems of the form~\eqref{eq:second-order}.
	
	We first focus on the autonomous ($\epsilon=0$) limit of the system
	\eqref{eq:first_order}, given by
	\begin{equation}
	\mathbf{B}\dot{\mathbf{z}}=\mathbf{Az}+\mathbf{F}(\mathbf{z}),\label{eq:first_order_aut}
	\end{equation}
	whose linearization at the fixed point $\mathbf{z}=0$ is 
	\begin{equation}
	\mathbf{B}\dot{\mathbf{z}}=\mathbf{Az}.\label{eq:first_order_linear}
	\end{equation}
	The linear system~\eqref{eq:first_order_linear} has invariant manifolds
	defined by eigenspaces of the generalized eigenvalue problem 
	\begin{equation}
	\left(\mathbf{A}-\lambda_{j}\mathbf{B}\right)\mathbf{v}_{j}=\mathbf{0},\quad j=1,\dots,N,\label{eq:lefteig}
	\end{equation}
	where for each distinct eigenvalue $\lambda_{j}$, there exists an
	\emph{eigenspace} $E_{j}\subset\mathbb{R}^{N}$ spanned by the real
	and imaginary parts of the corresponding generalized eigenvector $\mathbf{v}_{j}\in\mathbb{C}^{N}$.
	These eigenspaces are invariant for the linearized system~\eqref{eq:first_order_linear}
	and, by linearity, a subspace spanned by any combination of eigenspaces
	is also invariant for the system~\eqref{eq:first_order_linear}. A
	general invariant subspace of this type is known as a \emph{spectral
		subspace}~\cite{Haller2016} and is obtained by the direct-summation
	of eigenspaces as
	\[
	E_{j_{1},\dots,j_{q}}=E_{j_{1}}\oplus\dots\oplus E_{j_{q}},
	\]
	where $\oplus$ denotes the direct sum of vector spaces and $E_{j_{1},\dots,j_{q}}$
	is the spectral subspace obtained from the eigenspaces $E_{j_{1}},\dots,E_{j_{q}}$
	for some $q\in\mathbb{N}$. Classic examples of spectral subspaces
	are the stable, unstable and center subspaces, which are denoted by
	$E^{s}$, $E^{u}$ and $E^{c}$ and are obtained from eigenspaces
	associated to eigenvalues with negative, positive and zero real parts,
	respectively. By the center manifold theorem, these classic invariant
	subspaces of the linear system~\eqref{eq:first_order_linear} persist
	as invariant manifolds under the addition of nonlinear terms in system
	\eqref{eq:first_order_aut}. Specifically, there exist stable, unstable
	and center invariant manifolds $W^{s},W^{u}$ and $W^{c}$ tangent
	to $E^{s},E^{u}$ and $E^{c}$ at the origin $\mathbf{0}\in\mathbb{R}^{N}$
	respectively. All these manifolds are invariant and $W^{s}$ and $W^{u}$
	are also unique (see, e.g., Guckenheimer \& Holmes~\cite{Guckenheimer1983}).
	
	In analogy with the stable manifold $W^{s}$, which is the nonlinear
	continuation of the stable subspace $E^{s}$, a spectral submanifold
	(SSM)~\cite{Haller2016} is an invariant submanifold of the stable
	manifold $W^{s}$ that serves as the smoothest nonlinear continuation
	of a given stable spectral subspace of $E^{s}$. The existence and
	uniqueness results for such spectral submanifolds under appropriate
	conditions are derived by Haller \& Ponsioen~\cite{Haller2016} using
	the parametrization method of Cabré et al.~\cite{Cabre2003,Cabre2003a,Cabre2005}.
	The parametrization method also serves as a tool to compute these
	manifolds.
	
	We are interested in locally approximating the invariant manifolds
	of the fixed point $\mathbf{0}\in\mathbb{R}^{N}$ of system~\eqref{eq:first_order_aut}
	using the parametrization method. Let $\mathcal{W}(E)$
	be an invariant manifold of system~\eqref{eq:first_order_aut} which
	is tangent to a master spectral subspace $E\subset\mathbb{R}^{N}$
	at the origin such that
	\begin{equation}
	\mathrm{dim}(E)=\mathrm{dim}\left(\mathcal{W}(E)\right)=M<N.\label{eq:dimM}
	\end{equation}

	Let $\mathbf{V}_{E}=[\mathbf{v}_{1},\dots,\mathbf{v}_{M}]\in\mathbb{C}^{N\times M}$
	be a matrix whose columns contain the (right) eigenvectors corresponding
	to the master modal subspace $E$. Furthermore, we define
	a dual matrix $\mathbf{U}_{E}=[\mathbf{u}_{1},\dots,\mathbf{u}_{M}]\in\mathbb{C}^{N\times M}$
	which contains the corresponding left-eigenvectors that span the adjoint
	subspace $E^{\star}$ as
	\begin{equation}
	\mathbf{u}_{j}^{\star}\left(\mathbf{A}-\lambda_{j}\mathbf{B}\right)=\mathbf{0},\quad j=1,\dots,M,\label{eq:righteig}
	\end{equation}
	where we choose these eigenvectors to satisfy the normalization condition
	\begin{equation}
	\mathbf{u}_{i}^{\star}\mathbf{B}\mathbf{v}_{j}=\delta_{ij}~(\textrm{Kronecker delta)}.\label{eq:normalization}
	\end{equation}
	Using the eigenvalue problems~\eqref{eq:lefteig}-\eqref{eq:righteig},
	we obtain the following relations for the matrices $\mathbf{V}_{E}$
	and $\mathbf{U}_{E}$ 
	\begin{equation}
	\mathbf{A}\mathbf{V}_{E}=\mathbf{B}\mathbf{V}_{E}\mathbf{\Lambda}_{E},\label{eq:matrighteig}
	\end{equation}
	\begin{equation}
	\mathbf{U}_{E}^{\star}\mathbf{A}=\mathbf{\Lambda}_{E}\mathbf{U}_{E}^{\star}\mathbf{B},\label{eq:matlefteig}
	\end{equation}
	where $\mathbf{\Lambda}_{E}:=\mathrm{diag}(\lambda_{1},\dots,\lambda_{M})$.
	We note that if the matrices $\mathbf{M},\mathbf{C},\mathbf{K}$ are
	symmetric, then the matrices $\mathbf{A},\mathbf{B}$ will be symmetric
	as well. In that case, the left and the right eigenvectors $\mathbf{u}_{j},\mathbf{v}_{j}$
	are identical and we may conveniently choose $\mathbf{U}_{E}=\bar{\mathbf{V}}_{E}$,
	with the overbar denoting complex conjugation. 
	
	The common approach to local invariant manifold computation involves
	diagonalizing the system~\eqref{eq:first_order_aut} as 
	\begin{align}
	\dot{\mathbf{q}} & =\mathbf{\Lambda q}+\mathbf{T}(\mathbf{q}),\label{eq:diag_autonomous}
	\end{align}
	where
	\begin{equation}
	\mathbf{\Lambda}:=\mathrm{diag}(\lambda_{1},\dots,\lambda_{N}),\quad\mathbf{T}(\mathbf{q}):=\mathbf{U}^{\star}\mathbf{F}(\mathbf{V}\mathbf{q}),\quad\mathbf{V}=[\mathbf{v}_{1},\dots,\mathbf{v}_{N}],\quad\mathbf{U}=[\mathbf{u}_{1},\dots,\mathbf{u}_{N}],\label{eq:diag_terms}
	\end{equation}
	and $\mathbf{q}\in\mathbb{C}^{N}$ are modal coordinates with $\mathbf{z}=\mathbf{V}\mathbf{q}$.
	When $\mathbf{B}=\mathbf{I}_{N}$, then using the normalization condition
	$\eqref{eq:normalization}$, we obtain $\mathbf{U}^{\star}=\mathbf{V}^{-1}$,
	which results in the familiar diagonalized form~\eqref{eq:diag_autonomous}
	with $\mathbf{T}(\mathbf{q})=\mathbf{V}^{-1}\mathbf{F}(\mathbf{V}\mathbf{q})$.
	
	While the form~\eqref{eq:diag_autonomous} is very helpful for the
	purposes of proving the existence and uniqueness properties of invariant
	manifolds, it presents a computationally intractable form for the
	actual computation of invariant manifolds in high-dimensional finite element problems, as we will see next.
	\FloatBarrier
	\section{Pitfalls of the diagonalized form~\eqref{eq:diag_autonomous}}
	
	\label{subsec:Pitfalls-of-transformation}
	In this work, we use the Kronecker notation for expressing smooth nonlinear functions as a multi-variate Taylor series in terms of their arguments. The Kronecker product (also known as the outer/dyadic product) is commonly denoted by the symbol $ \otimes $. For a column vector $ \mathbf{z} \in \mathbb{R}^N$, the Kronecker product operation $ \mathbf{z} \otimes \mathbf{z} $ returns the matrix $ \mathbf{z}\mathbf{z}^{\top} \in \mathbb{R}^{N\times N}$. In index notation, we write
	\begin{equation}\label{key}
	(\mathbf{z}\otimes \mathbf{z})_{ij} = z_i z_j, \quad \forall i,j \in {1,\dots,N}.  
	\end{equation}
	The Kronecker notation is more generally defined for obtaining the product of higher-order tensors, where a first-order tensor can be viewed as a vector, a second-order tensor, as a matrix and an order-$ k $ tensor as a $k-$dimensional array. Specifically, the Kronecker product of two tensors of orders $ p $ and $ q $ yields a tensor of order $ p+q $. We refer to Van Loan~\cite{vanLoan2000} for a concise review of the Kronecker product and its properties. 
	
	Now, the system nonlinearity
	$\mathbf{F}$ (see eq.~\eqref{eq:first_order_aut}) can be expanded
	in terms of the physical coordinates $\mathbf{z}\in\mathbb{R}^{N}$
	as
	\begin{equation}
	\mathbf{F}(\mathbf{z})=\sum_{k\in\mathbb{N}}\mathbf{F}_{k}\mathbf{z}^{\otimes k},\label{eq:expF}
	\end{equation}
	where $\mathbf{z}^{\otimes k}$
	denotes the term $\mathbf{z}\otimes\dots\otimes\mathbf{z}$~($k$-times), containing $ N^k $ monomial terms at degree $ k $ in the variables $ \mathbf{z} $. The array~$\mathbf{F}_{k}\in\mathbb{R}^{N\times N^{k}}$ contains the coefficients of the nonlinearity $\mathbf{F}$ associated to each of these monomials. Similarly, the nonlinearity $\mathbf{T}$ (see eq.~\eqref{eq:diag_autonomous})
	in modal coordinates $\mathbf{q}\in\mathbb{C}^{N}$ can be expanded
	as 
	\begin{equation}
	\mathbf{T}(\mathbf{q})=\sum_{k\in\mathbb{N}}\mathbf{T}_{k}\mathbf{q}^{\otimes k}.\label{eq:expT}
	\end{equation}
	
	\subsection{Eigenvalue and eigenvector computation}
	\FloatBarrier
	For local approximations of invariant manifolds around a fixed point
	of~\eqref{eq:first_order_aut}, it is commonly assumed that the complete
	generalized spectrum of the matrix $\mathbf{B}^{-1}\mathbf{A}$ (or
	generalized eigenvalues of the pair $\mathbf{A},\mathbf{\mathbf{B}}$)
	is known and that a basis in which the linear system~\eqref{eq:first_order_linear}
	takes its Jordan canonical form is readily available (see, e.g., Simo
	\cite{Simo1990}, Homburg et al.~\cite{Homburg1995}, Tian \& Yu~\cite{Tian2013},
	Haro et al.~\cite{Haro2016}, Ponsioen et al.~\cite{Ponsieon2018,Ponsioen2020}).
	For small to moderately-sized systems, obtaining a complete set of
	(generalized) eigenvalues/eigenvectors can indeed be accomplished
	using numerical eigensolvers, but this quickly transforms into an
	intractable task as the system size increases.

	\begin{figure}[H]
		\centering{}\includegraphics[width=0.4\linewidth]{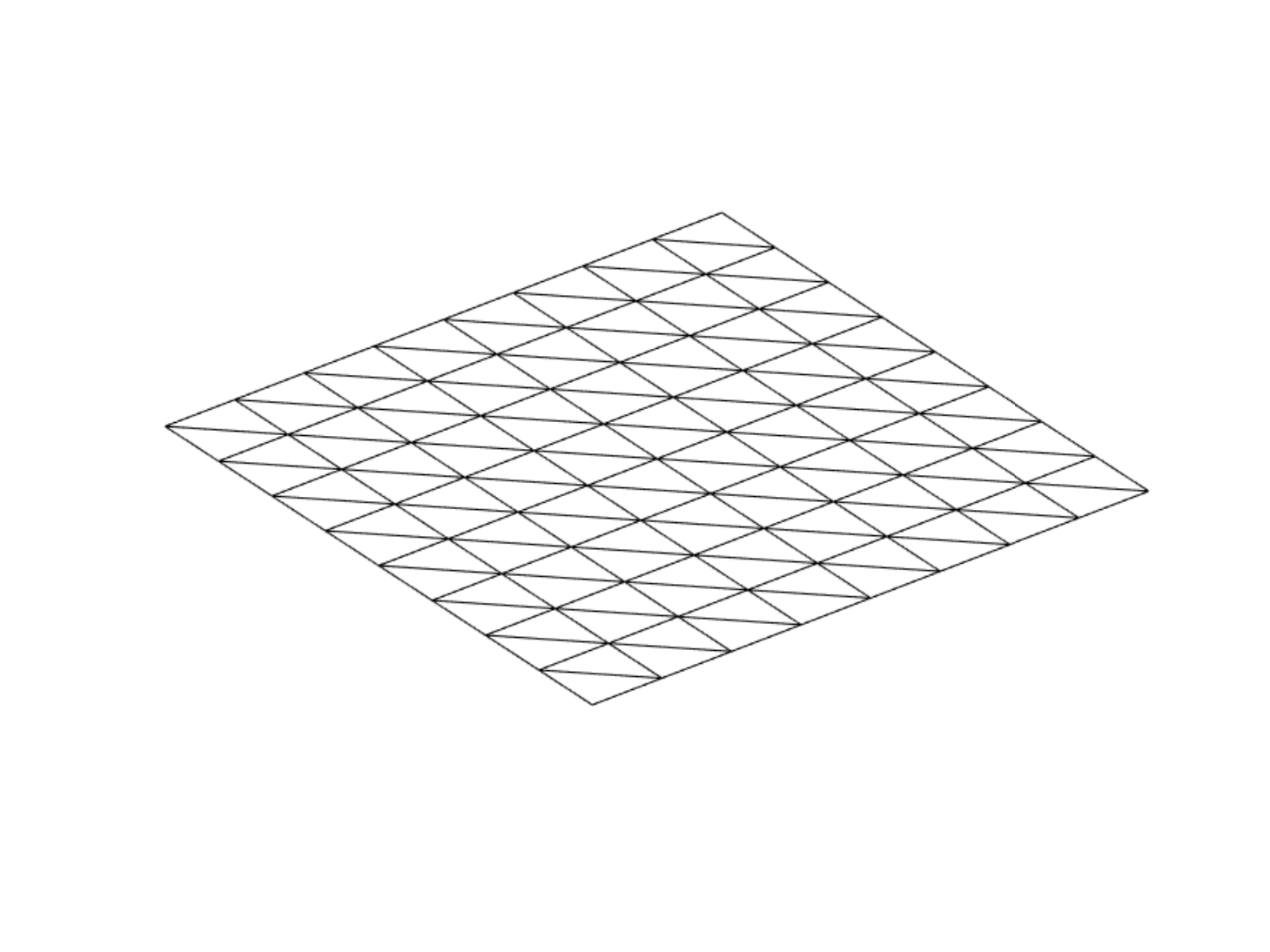} \caption{\label{fig:mesh} \small \textbf{Mesh for case study}: We use a shell-based
			finite element mesh of a square plate with geometric nonlinearities,
			arising from von Kármán strains for performing numerical experiments
			(see~Figures~\ref{fig:ev} and~\ref{fig:mem}). The material is
			linear elastic with Young's modulus $70$ GPa, density $2700$ kg/$\mathrm{mm}^{3}$
			and Poisson's ratio $0.33.$ The plate has a thickness of $8$ mm
			and length is proportional to the square root of the total number
			of elements. This fixes the size of the elements in the mesh and avoids numerical errors that may otherwise arise in larger meshes.}
	\end{figure}

	While techniques in numerical linear algebra can help us determine
	a small subset of eigenvalues and eigenvectors for very high-dimensional
	systems using a variety of iterative methods (see, e.g., Golub \&
	Van Loan~\cite{GVL}), obtaining a complete set of eigenvectors of
	such systems remains unfeasible despite the availability of modern
	computing tools. To emphasize this, we illustrate in Figure~\ref{fig:ev}
	the time and memory required for the eigenvalue computation for the
	finite element mesh for a square plate (see Figure~\ref{fig:mesh}).
	The purpose of this comparison is to report trends in computational
	complexity rather than precise numbers. To this end, we have used
	MATLAB across all comparisons, which may not be the fastest computing platform generally but is known to assimilate the state-of-the-art algorithms for numerical linear algebra computations (Golub \& Van
	Der Vorst~\cite{Golub2000}). 
	
	\begin{figure}[H]
		\begin{centering}
			\subfloat[]{\centering{}\includegraphics[width=0.48\linewidth]{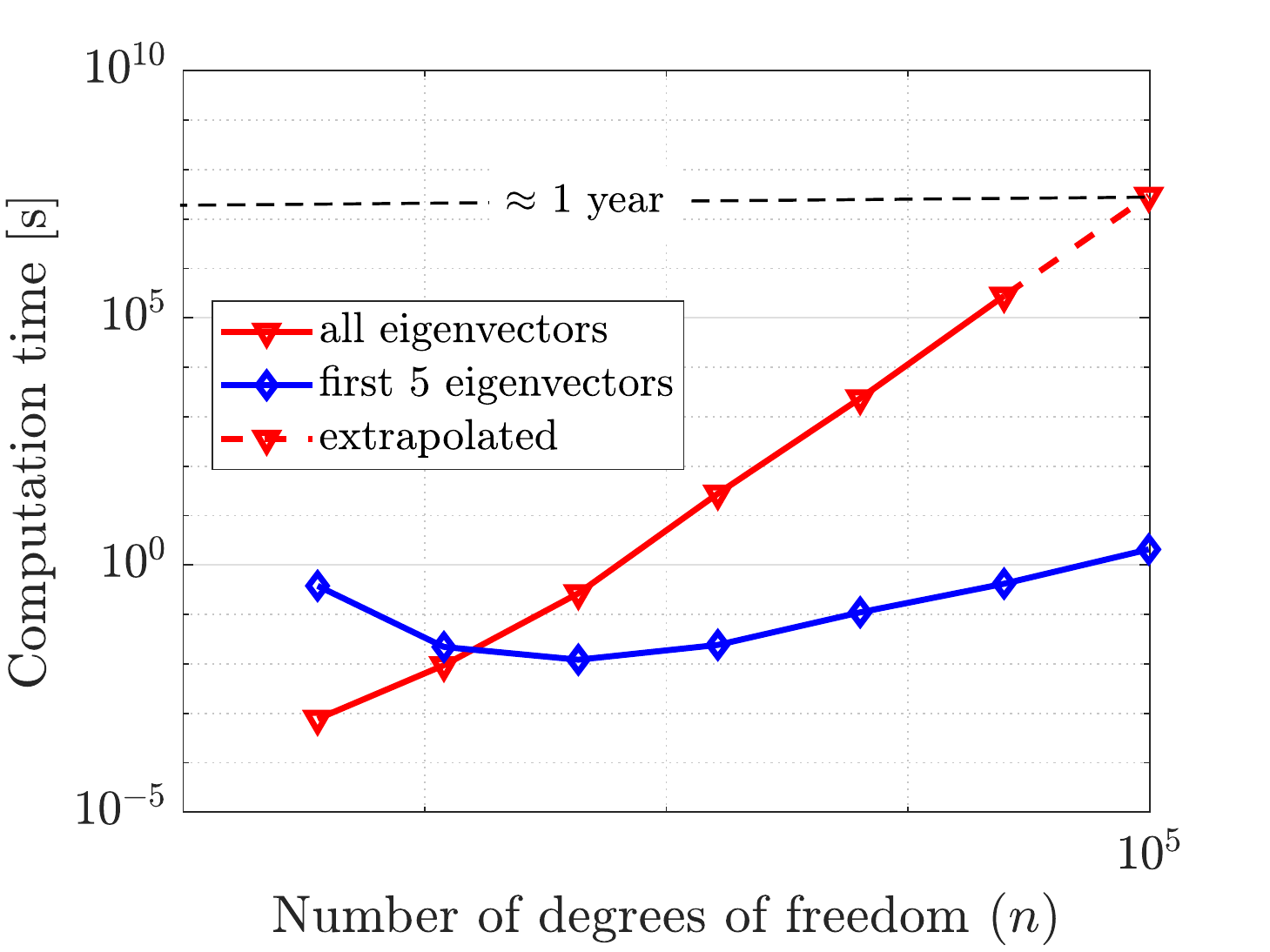}
				
			}\hfill{}\subfloat[]{\centering{}\includegraphics[width=0.48\linewidth]{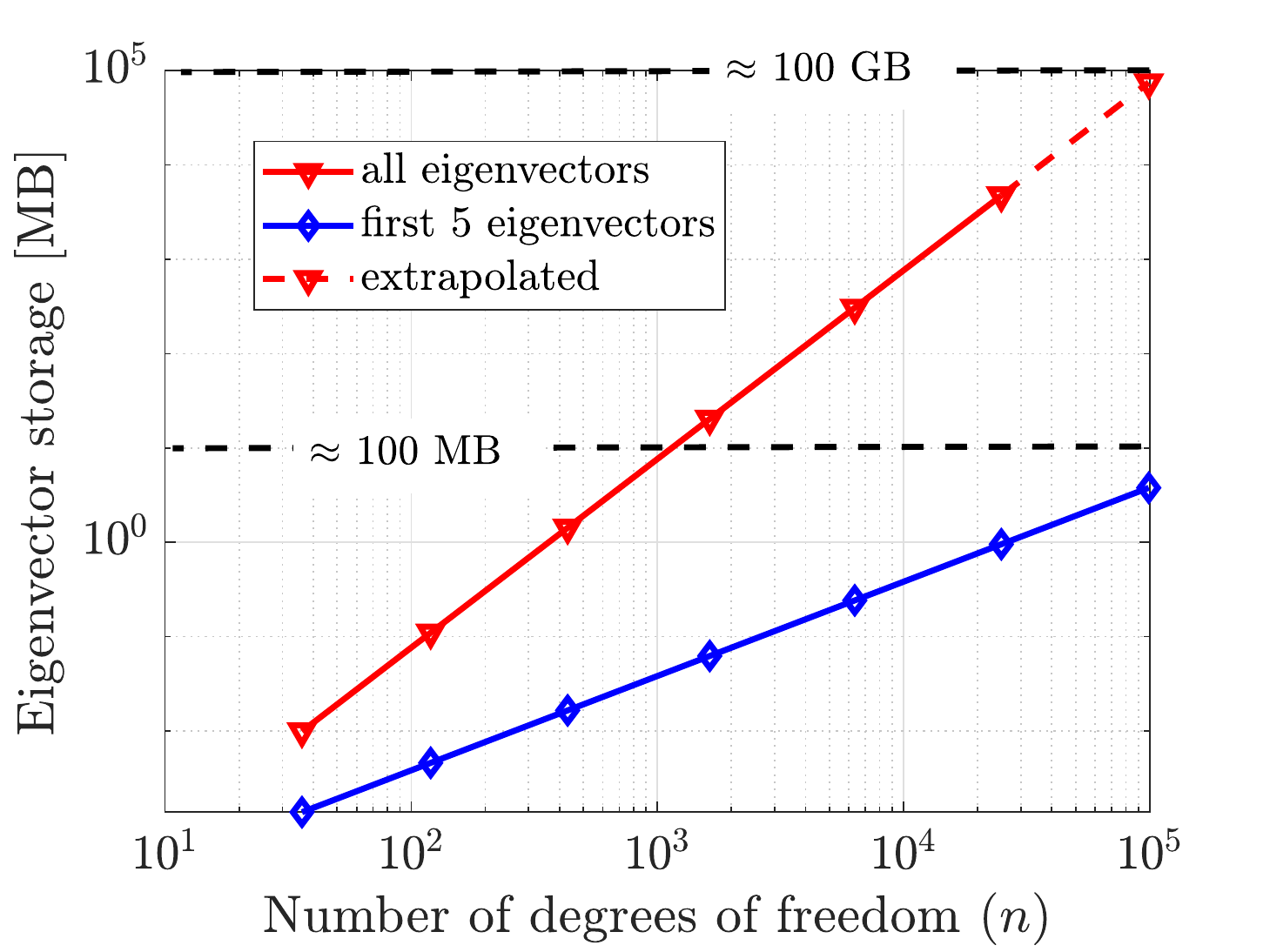}
			
		}
		\par\end{centering}
	\caption{\label{fig:ev} \small\textbf{Cost of computing eigenvalues and eigenvectors
		}of an $n-$degree-of-freedom plate example (see Figure~\ref{fig:mesh})
		for $n=$ 36, 120, 432, 1,632, 6,336, 24,960, 99,072. (a) Computation time and
		(b) memory required for obtaining all $n$ eigenvalues and eigenvectors
		using the $\texttt{eig}$ command of MATLAB compared to those for
		obtaining a subset of $5$ smallest-magnitude eigenvalues and their
		associated eigenvectors using the $\texttt{eigs}$ commands of MATLAB.
		These computations were performed on the ETH Z\"{u}rich Euler cluster
		with $10^{5}$ MB RAM. The computation of all eigenvalues in the case
		of $n=$ 99,072 degrees of freedom was manually terminated as the estimated
		computation time via extrapolation was found to be around 341 days.}
\end{figure}

Figure~\ref{fig:ev}a shows that as the number of degrees of freedom,
$n$, increases from a few tens to approximately a hundred thousand,
the computational time required for computing a full set of eigenvalues
of the system grows polynomially up to almost a year. For computing
a subset of eigenvalues in discretized PDEs, sparse iterative eigensolvers
are used, such as the routines (e.g., Stewart~\cite{Stewart2002},
Lehoucq et al.~\cite{Lehoucq1998}) implemented by the MATLAB's $\texttt{eigs}$
command (cf. direct eigensolvers implemented by the $\texttt{eig}$
command). These sparse solvers are considered inefficient for nearly
full or less sparse matrices. There are, therefore, two competing factors here,
sparsity of the matrices and the size of the matrices. The small matrices
in the beginning have very low sparsity and sparse eigensolver \texttt{eigs}
of MATLAB is inefficient for computing eigenvalues here. Indeed, we
see that computing the full set of eigenvalues for a small matrix (using the $\texttt{eig}$ command) ends up being less expensive compared to computing a subset of eigenvalues. As sparsity is governed by the number of DOFs that are shared by neighboring
elements relative to the total number of degrees of freedom, it increases
with mesh refinement. Thus, sparse eigensolvers become more efficient
as we refine the mesh initially, but after the refinement reaches an
optimum value, the computation time is governed solely by the size
of the matrix. 

Furthermore, all these eigenvectors must be held in the computer's
active memory (RAM) in typical invariant manifold computations, which
contributes towards very high memory requirements, as shown in Figure~\ref{fig:ev}b.
At the same time, these figures also show that a small subset of eigenvectors
can be quickly computed and easily stored even for very high dimensional
systems. 
\FloatBarrier
\subsection{Unfeasible memory requirements due to coordinate-change}

\label{subsec:memory}

\noindent Aside from the cost of eigenvalue computation, invariant
manifold computations typically involve local approximations via Taylor
series. These are obtained by transforming the system into modal coordinates
(see eq~\eqref{eq:diag_autonomous}), expressing the manifold locally
as a graph over the master subspace, substituting the polynomial ansatz
into eq.~\eqref{eq:diag_autonomous} and solving the invariance equations
recursively at each order. While such a modal transformation results
in decoupling of the governing equations at the linear level, it generally
annihilates the inherent sparsity in the nonlinear terms, as shown
in Figure~\ref{fig:mem}a. That sparsity generally arises because
only neighboring elements of the numerical mesh share coupled degrees
of freedom. Due to the loss of this sparsity upon transformation to
the diagonal form~\eqref{eq:diag_autonomous}, the number of polynomial
coefficients required to describe the nonlinearities increases by
orders of magnitude, resulting in unfeasible memory requirements.

\begin{figure}[H]
	\begin{centering}
		\subfloat[]{\centering{}\centering \includegraphics[width=0.55\linewidth]{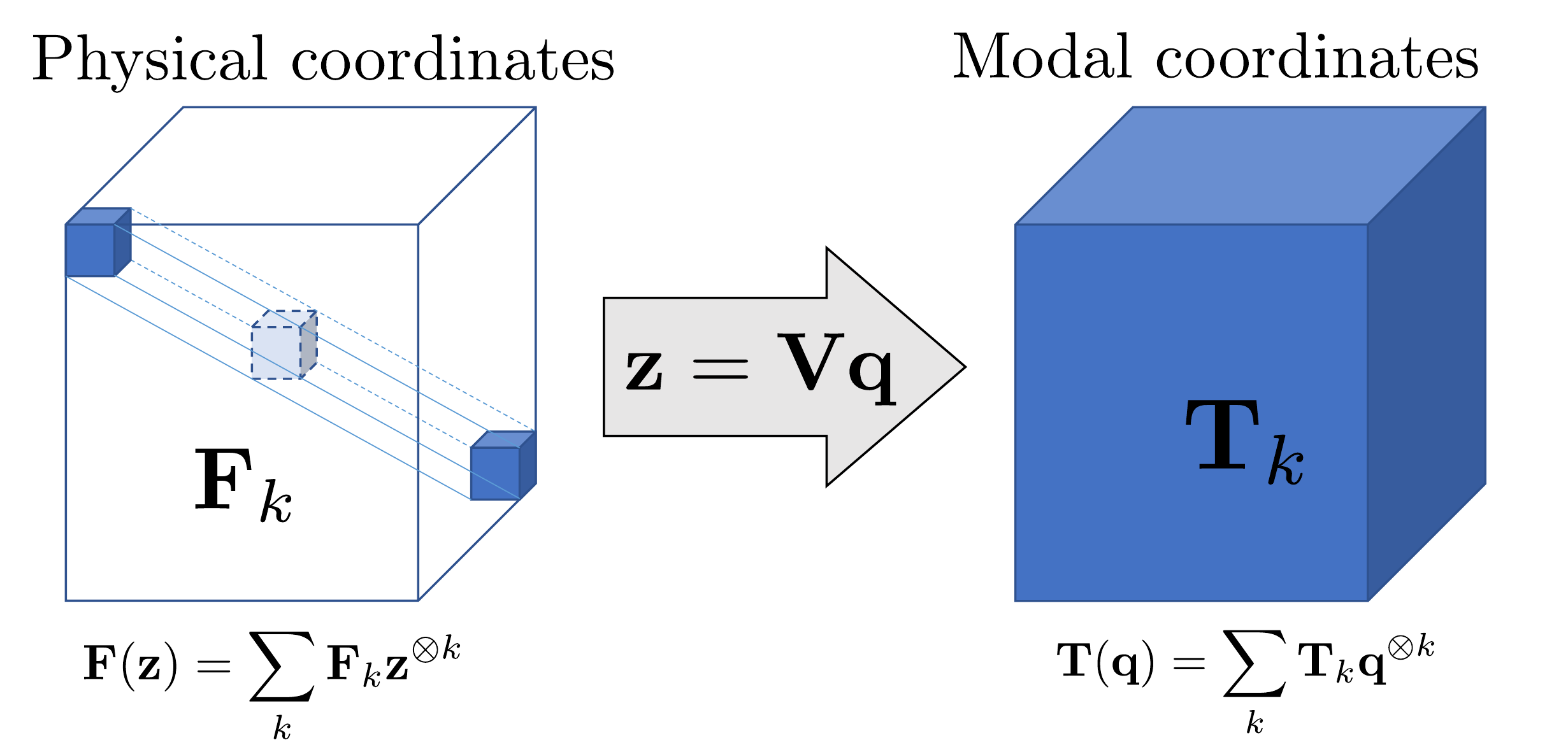}
			
		}\hfill{}\subfloat[]{\centering{}\includegraphics[width=0.4\linewidth]{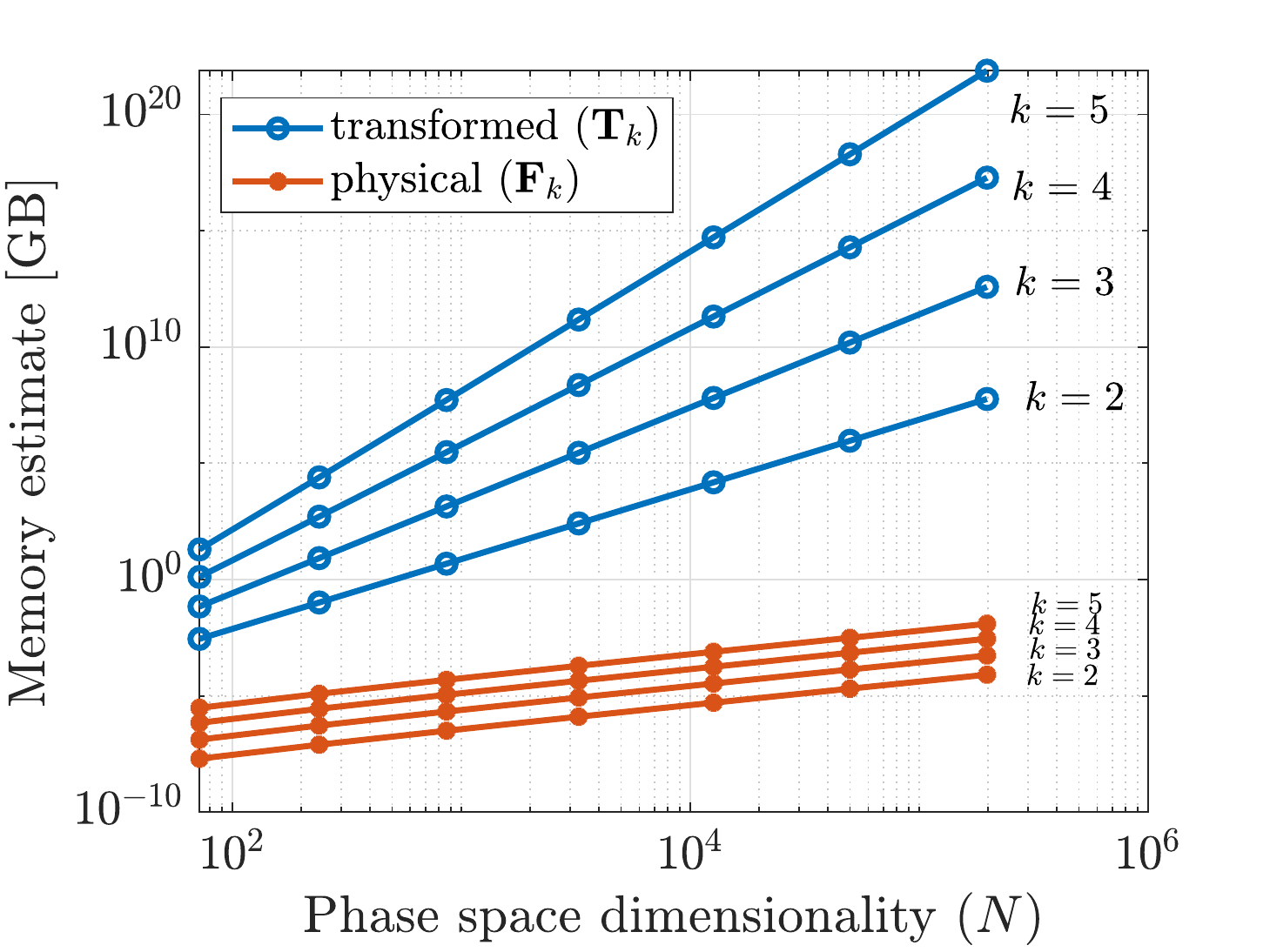}
		
	}
	\par\end{centering}
\caption{\label{fig:mem} \small (a)  \textbf{Destruction of sparsity}: Transforming
	the system~\eqref{eq:first_order_aut} into~\eqref{eq:diag_autonomous}
	via the linear transformation $\mathbf{z}=\mathbf{Vq}$ results in
	destruction of the inherent sparsity of the governing equations in
	physical coordinates $\mathbf{z}$, which leads to unfeasible memory
	(RAM) requirements for storing nonlinearity coefficients. Here, the
	multi-dimensional array $\mathbf{F}_{k}$ and $\mathbf{T}_{k}$ represent
	the polynomial coefficients at degree $k$ for the nonlinearities
	$\mathbf{F}$ (in physical coordinates) and $\mathbf{T}$ (in modal
	coordinates); see eqs.~\eqref{eq:first_order_aut},~\eqref{eq:diag_autonomous},
	\eqref{eq:expF} and~\eqref{eq:expT}. (b) Comparison of memory requirements
	for storing the nonlinearities $\mathbf{F}$ and $\mathbf{T}$ at
	degrees $k=2,3,4,5$ in the $n-$degree-of-freedom square plate example
	(see Figure~\ref{fig:mesh}) with $n=$ 36, 120, 432, 1,632, 6,336, 24,960, 99,072
	and phase space dimension $N=2n$.}
\end{figure}

Indeed, in Figure~\ref{fig:mem}b, we compare the memory estimates
for storing these coefficients in physical vs. modal coordinates as
a function of the system's phase space dimension $N$. We see
that even for the moderately sized meshes of the square plate example
(see Figure~\ref{fig:mesh}) considered here, the storage requirements
for the transformed coefficients reaches astronomically high values
in the order of several terabytes/petabytes. At the same time, however,
note that the RAM requirements for handling the same coefficients
in physical coordinates are much less than a gigabyte, which is easily
manageable for modern computers.

\section{Computing invariant manifolds of fixed points in physical coordinates}

\label{sec:computation_autonomous}

Unlike commonly employed computational approaches~\cite{Simo1990,Homburg1995,Tian2013,Haro2016,Ponsieon2018,Ponsioen2019},
we now describe the computation of general invariant manifolds in
physical coordinates using the eigenvectors and eigenvalues associated
to the master subspace $E$ only. This is motivated by the
computational advantages we expect based on Figures~\ref{fig:ev}
and~\ref{fig:mem}.

We seek to compute an invariant manifold $\mathcal{W}\left(E\right)$
tangent to a spectral subspace $E$ at the origin of system
\eqref{eq:first_order_aut}. Let $\mathbf{W}:\mathbb{C}^{M}\to\mathbb{R}^{N}$
be a mapping that parametrizes the $M-$dimensional manifold~$\mathcal{W}\left(E\right)$
and let $\mathbf{p}\in\mathbb{C}^{M}$ be its parametrization coordinates.
Then, $\mathbf{W}(\mathbf{p})$ provides us the coordinates of the
manifold in the phase space of system~\eqref{eq:first_order_aut}, as
shown in Figure~\ref{fig:param}. For any trajectory $\mathbf{z}(t)$
on the invariant manifold $\mathcal{W}\left(E\right)$,
we have a \emph{reduced dynamics} trajectory $\mathbf{p}(t)$ in the
parametrization space such that
\begin{equation}
\mathbf{\mathbf{z}}(t)=\mathbf{W}(\mathbf{p}(t)).\label{eq:invariance}
\end{equation}
Let $\mathbf{R}:\mathbb{C}^{M}\to\mathbb{C}^{M}$ be a parametrization
for the reduced dynamics. Then, any reduced dynamics trajectory~$\mathbf{p}(t)$ satisfies
\begin{equation}
\dot{\mathbf{p}}=\mathbf{R}(\mathbf{p}).\label{eq:red_dyn}
\end{equation}
\begin{figure}[H]
	\centering \includegraphics[width=0.8\linewidth]{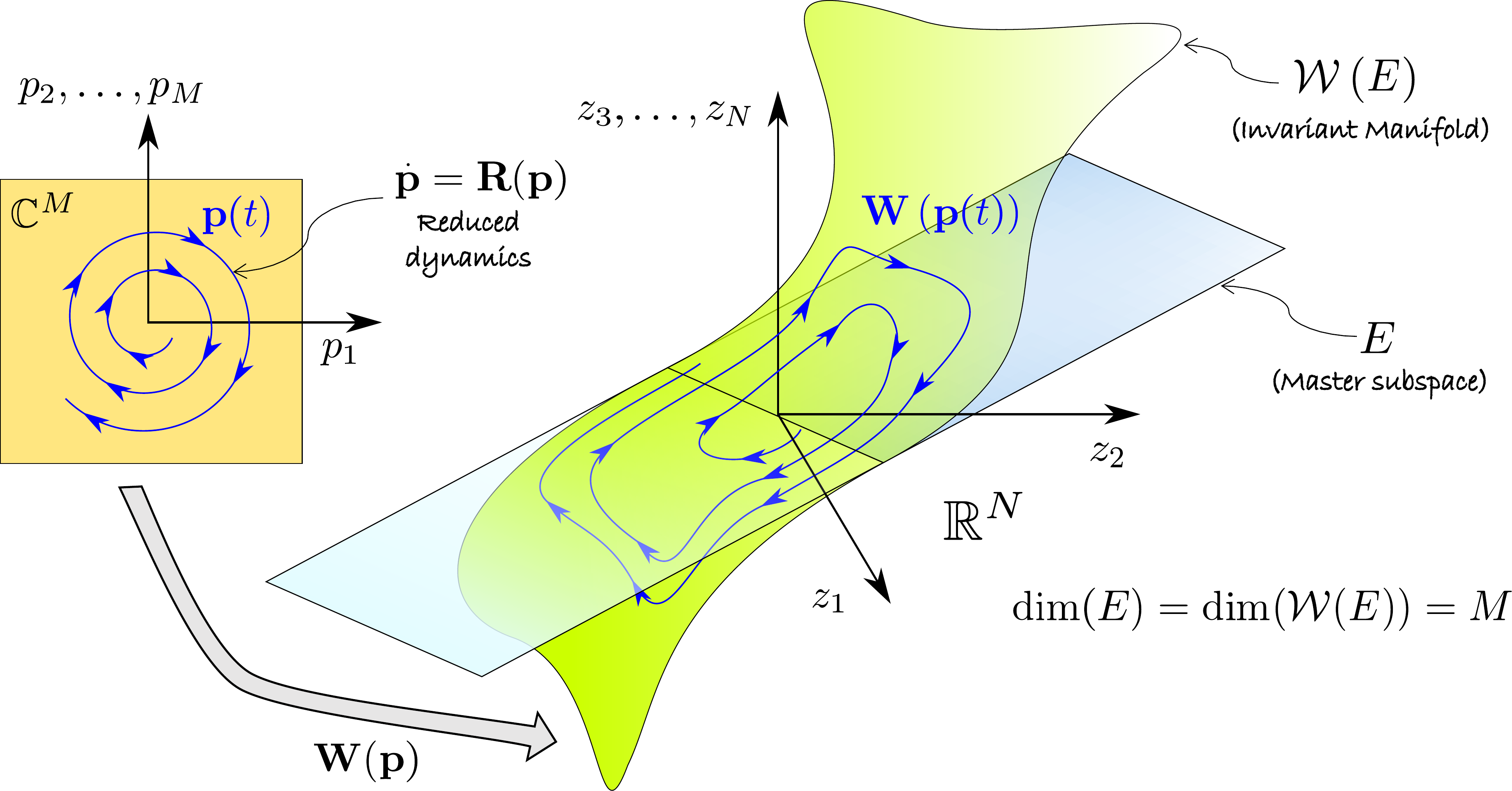}
	\caption{\label{fig:param}\small Using the parametrization method, we obtain the
		parametrization $\mathbf{W}:\mathbb{C}^{M}\to\mathbb{R}^{N}$ of an
		$M-$dimensional invariant manifold for the system~\eqref{eq:first_order_aut}
		constructed around a spectral subspace $E$ with $\dim(E)=\dim(\mathcal{W}(E))=M$.
		This manifold is tangent to $E$ at the origin. Furthermore,
		we have the freedom to choose the parametrization $\mathbf{R}:\mathbb{C}^{M}\to\mathbb{C}^{M}$
		of the reduced dynamics on the manifold such that the function $\mathbf{W}$
		also maps the reduced system trajectories $\mathbf{p}(t)$ onto the
		full system trajectories on the the invariant manifold, i.e., $\mathbf{z}(t)=\mathbf{W}\left(\mathbf{p}(t)\right)$.}
\end{figure}
Differentiating eq.~\eqref{eq:invariance} with respect to $t$ and
using eqs.~\eqref{eq:first_order} and~\eqref{eq:red_dyn}, we obtain
the \emph{invariance equation} of $\mathcal{W}(E)$ as
\begin{equation}
\mathbf{B}\,(D\mathbf{W})\,\mathbf{R}=\mathbf{A}\mathbf{W}+\mathbf{F}\circ\mathbf{W}.\label{eq:inv_eqn}
\end{equation}
To solve this invariance equation, we need to determine the parametrizations
$\mathbf{W}$ and $\mathbf{R}$. We choose to parameterize the manifold
and its reduced dynamics in the form of multivariate polynomial expansions
as 
\begin{equation}
\mathbf{W}(\mathbf{p})=\sum_{i\in\mathbb{N}}\mathbf{W}_{i}\mathbf{p}^{\otimes i},\label{eq:expW}
\end{equation}
\begin{equation}
\mathbf{R}(\mathbf{p})=\sum_{i\in\mathbb{N}}\mathbf{R}_{i}\mathbf{p}^{\otimes i},\label{eq:expR}
\end{equation}
where $\mathbf{W}_{i}\in\mathbb{C}^{N\times M^{i}}$, $\mathbf{R}_{i}\in\mathbb{C}^{M\times M^{i}}$
are matrix representation of multi-dimensional arrays containing the
unknown polynomial coefficients at degree $i$ for the parametrizations
$\mathbf{W}$ and $\mathbf{R}$. Furthermore, we have the expansion
\eqref{eq:expF} for the nonlinearity $\mathbf{F}$ in physical coordinates,
where $\mathbf{F}_{i}\in\mathbb{R}^{N\times N^{i}}$ are sparse arrays,
which are straight-forward to store despite their large size (see
Section~\ref{subsec:memory}). 

Using the expansions~\eqref{eq:expF}, \eqref{eq:expW} and \eqref{eq:expR}, we collect
the coefficients of the multivariate polynomials in the invariance
equation~\eqref{eq:inv_eqn} at degree $i\ge1$, similarly to Ponsioen
et al.~\cite{Ponsieon2018}, as 
\begin{equation}
\left(\mathbf{B}D\mathbf{W}\mathbf{R}\right)_{i}=\mathbf{A}\mathbf{W}_{i}+\left(\mathbf{F}\circ\mathbf{W}\right)_{i},\label{eq:inv_eqn_i}
\end{equation}
where
\begin{equation}
\left(\mathbf{B}D\mathbf{W}\mathbf{R}\right)_{i}=\mathbf{B}\ensuremath{\mathbf{W}_{1}}\ensuremath{\mathbf{R}_{i}}+\mathbf{B}\text{\ensuremath{\sum_{j=2}^{i}}\ensuremath{\mathbf{W}_{j}\boldsymbol{\mathcal{R}}_{i,j}}},\label{eq:BDW_i}
\end{equation}
with 
\begin{equation}
\boldsymbol{\mathcal{R}}_{i,j}:=\ensuremath{\sum_{k=1}^{j}}\ensuremath{\ensuremath{\underbrace{\mathbf{I}_{M}\otimes\dots\otimes\mathbf{I}_{M}\otimes\overset{\overset{k-\mathrm{th\,position}}{\uparrow}}{\mathbf{R}_{i-j+1}}\otimes\mathbf{I}_{M}\otimes\dots\otimes\mathbf{I}_{M}}_{j-\mathrm{terms}}}},\label{eq:R_ij}
\end{equation}
and
\begin{equation}
\left(\mathbf{F}\circ\mathbf{W}\right)_{i}=\sum_{j=2}^{i}\mathbf{F}_{j}\left(\sum_{\mathbf{q}\in\mathbb{N}^{j},|\mathbf{q}|=i}\mathbf{W}_{q_{1}}\otimes\dots\otimes\mathbf{W}_{q_{j}}\right).\label{eq:FoW_i}
\end{equation}
At leading-order, i.e., for $i=1$, equation~\eqref{eq:inv_eqn_i}
simply yields
\begin{equation}
\mathbf{A}\mathbf{W}_{1}=\mathbf{B}\mathbf{W}_{1}\mathbf{R}_{1}.\label{eq:eigenvalue}
\end{equation}
Comparing equation~\eqref{eq:eigenvalue} with the eigenvalue problem
\eqref{eq:matrighteig}, we choose a solution for $\mathbf{W}_{1},\mathbf{R}_{1}$
in terms of the master modes and their eigenvalues as 
\begin{equation}
\mathbf{W}_{1}=\mathbf{V}_{E},\quad\mathbf{R}_{1}=\boldsymbol{\Lambda}_{E}.\label{eq:W1}
\end{equation}

\begin{rem}
	The solution choice~\eqref{eq:W1} for $\mathbf{W}_{1},\mathbf{R}_{1}$
	is not unique. Indeed, we may choose $\mathbf{W}_{1}\in\mathbb{C}^{N\times M}$
	to be any matrix whose columns span the master subspace $E$,
	generally resulting in a non-diagonal $\mathbf{R}_{1}$. Since our
	system is defined in the space of reals ($\mathbb{R}$), a real choice
	of $\mathbf{W}_{1},\mathbf{R}_{1}$ allows us to choose the parametrization
	coordinates $\mathbf{p}$ in $\mathbb{R}^{M}$ instead of $\mathbb{C}^{M}$.
	This will result in $\mathbf{W}_{i}\in\mathbb{R}^{N\times M^{i}}$,
	$\mathbf{R}_{i}\in\mathbb{R}^{M\times M^{i}}$ for each $i$, which
	reduces the computational memory requirements by half relative to
	the complex setting. 
\end{rem}
At any order $i\ge2$ in eq.~\eqref{eq:inv_eqn_i}, we collect the terms
containing the coefficients $\mathbf{W}_{i}$ on the left-hand side
and the lower degree terms on the right-hand side as
\begin{equation}
\mathbf{B}\text{\ensuremath{\mathbf{W}_{i}\boldsymbol{\mathcal{R}}_{i,i}}}-\mathbf{A}\mathbf{W}_{i}=\mathbf{C}_{i}-\mathbf{B}\mathbf{W}_{1}\mathbf{R}_{i}\label{eq:ith_order_inv_vec}
\end{equation}
where $\boldsymbol{\mathcal{R}}_{i,i}$ is defined according to eq.
\eqref{eq:R_ij} and
\[
\mathbf{C}_{i}:=\left(\mathbf{F}\circ\mathbf{W}\right)_{i}-\mathbf{B}\text{\ensuremath{\sum_{j=2}^{i-1}}\ensuremath{\mathbf{W}_{j}\boldsymbol{\mathcal{R}}_{i,j}}}.
\]
We solve~\eqref{eq:ith_order_inv_vec} recursively for $i\ge2$ by
vectorizing it as (see, e.g., Van Loan~\cite{vanLoan2000})
\begin{equation}
\boldsymbol{\mathcal{L}}_{i}\mathbf{w}_{i}=\mathbf{h}_{i}(\mathbf{R}_{i}),\label{eq:vectorized_inv_eqn}
\end{equation}
where 
\begin{align}
\mathbf{w}_{i} & :=\mathrm{vec}\left(\mathbf{W}_{i}\right)\in\mathbb{C}^{NM^{i}},\label{eq:wi_def}\\
\boldsymbol{\mathcal{L}}_{i} & :=\left(\boldsymbol{\mathcal{R}}_{i,i}^{\top}\otimes\mathbf{B}\right)-\left(\mathbf{I}_{M^{i}}\otimes\mathbf{A}\right)\in\mathbb{C}^{NM^{i}\times NM^{i}},\label{eq:L_i_def}\\
\mathbf{h}_{i}(\mathbf{R}_{i}) & :=\mathrm{vec}\left(\mathbf{C}_{i}\right)-\mathbf{D}_{i}\mathrm{vec}\left(\mathbf{R}_{i}\right)\label{eq:hi_def}\\
\boldsymbol{\mathcal{R}}_{i,i} & =\text{\ensuremath{\sum_{j=1}^{i}}\ensuremath{\left(\mathbf{I}_{M}\right)^{\otimes j-1}\otimes\mathbf{R}_{1}\otimes\left(\mathbf{I}_{M}\right)^{\otimes i-j}}}\in\mathbb{C}^{M^{i}\times M^{i}},\textrm{(from definition\,\eqref{eq:R_ij})}\label{eq:R_i_def}\\
\mathbf{D}_{i} & :=\left(\mathbf{I}_{M^{i}}\otimes\mathbf{B}\mathbf{W}_{1}\right)\in\mathbb{C}^{NM\times M^{i}}.\label{eq:D_i_def}
\end{align}

In eq.~\eqref{eq:vectorized_inv_eqn}, the matrix $\boldsymbol{\mathcal{L}}_{i}$
is often called the order$-i$ \emph{cohomological} operator induced
by the linear flow~\eqref{eq:first_order_linear} and the master subspace
$E$ on the linear space whose coefficients are homogeneous,
$M-$variate polynomials of degree $i$ (see Haro et al.~\cite{Haro2016},
Murdock~\cite{Murdock2003}). Hence, at any order of expansion $i$,
the entries of $\boldsymbol{\mathcal{L}}_{i}$ are completely determined
using only the linear part of the full and reduced systems, i.e.,
via the matrices $\mathbf{A},\mathbf{B}$ and $\mathbf{R}_{1}$ (which
is equal to $\boldsymbol{\Lambda}_{E}$ due to the choice~\eqref{eq:W1}). 
\begin{rem}
	\label{rem:parallelization} For a diagonal choice of $\mathbf{R}_{1}$
	(see, e.g., choice~\ref{eq:W1}), the matrix $\boldsymbol{\mathcal{L}}_{i}$
	has a block-diagonal structure, i.e., system~\eqref{eq:vectorized_inv_eqn}
	can be split into $M^{i}$ decoupled linear systems containing $N$
	equations each. Hence, the coefficients parametrizing the manifold
	and its reduced dynamics can be determined independently for each
	monomial in $\mathbf{p}^{\otimes i}$. This splitting of the large
	system~\eqref{eq:vectorized_inv_eqn} into smaller decoupled systems
	not only eases computations but also makes these computations appealing
	for a parallel computing, which has the potential to speed these computations
	up by a factor of $M^{i}$ at each order $i$. 
\end{rem}
Note that the system~\eqref{eq:ith_order_inv_vec} is under-determined
in terms of the unknowns $\mathbf{W}_{i},\mathbf{R}_{i}$. As discussed
by Haro et al.~\cite{Haro2016}, this underdeterminacy turns out to
be an advantage, as it provides us the freedom to choose a particular
style of parametrization depending on the context. When the matrix
$\boldsymbol{\mathcal{L}}_{i}$ is non-singular for every $i\ge2$,
the cohomological equations~\eqref{eq:vectorized_inv_eqn} have a
unique solution $\mathbf{W}_{i}$ for any choice of the reduced dynamics
$\mathbf{R}_{i}$. The trivial choice
\begin{equation}
\mathbf{R}_{i}=\mathbf{0}\quad\forall i\ge2,\label{eq:red_dyn_trivial}
\end{equation}
leads to linear reduced dynamics. However, as we will see next, $\boldsymbol{\mathcal{L}}_{i}$
may be singular in the presence of \emph{resonances} and yet system
\eqref{eq:vectorized_inv_eqn} may be solvable under an appropriate
choice of parametrization. 

\subsection{Choice of parametrization}

\subsubsection*{Eigenstructure of $\boldsymbol{\mathcal{L}}_{i}$}

In order to choose the reduced dynamics $\mathbf{R}$ appropriately,
we seek to explore the eigenstructure of $\boldsymbol{\mathcal{L}}_{i}$
in relation to that of the matrix $\boldsymbol{\mathcal{R}}_{i}$
and the generalized matrix pair $\left(\mathbf{B},\mathbf{A}\right)$.
We first derive a general result which helps us compute the eigenstructure
of $\boldsymbol{\mathcal{R}}_{i}$. For notational purposes, we introduce
an ordered set $\Delta_{i}$ which contains all $i$-tuples $\boldsymbol{\ell}_{j}$
(indexed lexicographically) taking values in the range $1,\dots,M$,
defined as
\begin{equation}
\Delta_{i,M}:=\{\boldsymbol{\ell}_{1},\dots,\boldsymbol{\ell}_{M^{i}}\in\{1,\dots,M\}^{i}\subset\mathbb{N}^{i}\}.\label{eq:index_set}
\end{equation}
As an example, consider the case of $i=3$ and $M=2$. Then, we may
order the $3-$tuples in $\Delta_{3,2}$ lexicographically as $\{(1,1,1),(1,1,2),(1,2,1),(1,2,2),(2,1,1),(2,1,2),(2,2,1),(2,2,2)\}$,
which contains $2^{3}$ elements. Essentially, each $ \boldsymbol{\ell} \in \Delta_{i,M} $ corresponds to a monomial of degree $ i $ in the reduced variables $ \mathbf{p}\in \mathbb{C}^M $, i.e., $ p_{\ell_1}p_{\ell_2}\dots p_{\ell_i} $   

At a high order $i$ for the manifold expansion, the matrix $\boldsymbol{\mathcal{R}}_{i,i}$
in eq.~\eqref{eq:ith_order_inv_vec} may be high-dimensional, even
though its components only involve the low-dimensional matrices $\mathbf{I}_{M}$
and $\mathbf{R}_{1}$ (see eq.~\eqref{eq:R_i_def}). Proposition~\ref{thm:1}
in Appendix~\ref{sec:Proof-of-Proposition1} allows us to compute
all the eigenvalues and eigenvectors of $\boldsymbol{\mathcal{R}}_{i,i}$
simply in terms of those of $\mathbf{R}_{1}$. Indeed, let the eigenvalues
of $\mathbf{R}_{1}^{\top}$ be given by $\lambda_{1},\dots,\lambda_{M}$.
Note that for a diagonal choice of ~$\mathbf{R}_{1}$ (see choice~\eqref{eq:W1}), the left and right eigenvectors are simply given by the unit vectors aligned with the coordinate axes
in $\mathbb{C}^{M}$, i.e. $\mathbf{e}_{1},\dots,\mathbf{e}_{M}$.
Then, from Proposition~\ref{thm:1}, the eigenvalues and eigenvectors
of $\boldsymbol{\mathcal{R}}_{i,i}^{\top}$ are given as 
\begin{equation}
\lambda_{\mathbf{\boldsymbol{\ell}}}:=\lambda_{\ell_{1}}+\dots+\lambda_{\ell_{i}},\quad\mathbf{e}_{\boldsymbol{\ell}}=\mathbf{e}_{\ell_{1}}\otimes\dots\otimes\mathbf{e}_{\ell_{i}},\quad\boldsymbol{\ell}\in\Delta_{i,M}.\label{eq:eig_R1}
\end{equation}
Furthermore, Proposition~\ref{thm:2} in Appendix~\ref{sec:Proof-of-Proposition1}
characterizes the eigenstructure of $\boldsymbol{\mathcal{L}}_{i}$
in relation to that of the matrices $(\mathbf{B},\mathbf{A})$ and $\boldsymbol{\mathcal{R}}_{i,i}$.
From Proposition~\ref{thm:2}, we deduce that $\boldsymbol{\mathcal{L}}_{i}$
is singular whenever the resonance $\lambda_{\mathbf{\boldsymbol{\ell}}}=\lambda_{j}$
occurs for some $\boldsymbol{\ell}\in\Delta_{i,M}$, $j\in\{1,\dots,N\}$.
In this case, the solvability of eq.~\eqref{eq:vectorized_inv_eqn}
depends on the nature of such resonances. Hence, these resonances
are distinguished into inner and outer resonances as 
\begin{align}
\textrm{Inner resonances:}\quad\lambda_{\mathbf{\boldsymbol{\ell}}} & =\lambda_{j},\quad\boldsymbol{\ell}\in\Delta_{i,M},j\in\{1,\dots,M\},\label{eq:internal_res}\\
\textrm{Outer resonances:}\quad\lambda_{\mathbf{\boldsymbol{\ell}}} & =\lambda_{j},\quad\boldsymbol{\ell}\in\Delta_{i,M},j\in\{M+1,\dots,N\}.\label{eq:cross_res}
\end{align}
Both inner and outer resonances result in the cohomological operator
$\boldsymbol{\mathcal{L}}_{i}$ becoming singular. The main difference
between these resonances is that the cohomological equation~\eqref{eq:ith_order_inv_vec}
can be solved in the presence of inner resonances by adjusting the
parametrization choice of $\mathbf{R}_{i}$ so that the right-hand
side of~\eqref{eq:ith_order_inv_vec} belongs to $\mathrm{im(}\boldsymbol{\mathcal{L}}_{i})$,
which we will discuss shortly. In the presence of outer resonances,
however, the right-hand side of equation~\eqref{eq:ith_order_inv_vec}
cannot be adjusted to lie in the range of the operator~$\boldsymbol{\mathcal{L}}_{i}$
and, hence, system~\eqref{eq:ith_order_inv_vec} has no solution.
Indeed, the manifold does not exist in presence of certain outer resonances
(see Cabré et al.~\cite{Cabre2003}, Haller \& Ponsioen~\cite{Haller2016}).
Haro et al.~\cite{Haro2016} refer to inner and outer resonances as
internal and cross resonances, we use this terminology of Ponsioen
et al.~\cite{Ponsieon2018} as internal resonances carry a different
meaning in the context of mechanics.

Next, we discuss two common choices for reduced dynamics parametrization,
i.e., the \emph{normal form} and the \emph{graph} style parametrizations,
which are useful for solving the cohomological equation~\eqref{eq:vectorized_inv_eqn}
in the presence of inner resonances. 

\subsubsection*{Normal form parametrization}

Normal forms provide us tools for the qualitative and quantitative
understanding of local bifurcations in dynamical systems. Normal form
computations for any dynamical system involve successive near-identity
coordinate transformations to simplify the transformed dynamics. The
simplest form of dynamics that one can hope for is linear. In the
presence of resonances, however, a transformation that linearizes
the dynamics does not exist and the normal form procedure results
in nonlinear dynamics that is only ``as simple as possible''. This
is achieved by systematically removing the nonessential terms
from the Taylor series up to any given order (see, e.g., Guckenheimer
\& Holmes~\cite{Guckenheimer1983}). 

Using the parametrization method, we can simultaneously compute the
normal form parametrization for the reduced dynamics $\mathbf{R}$
along with the parametrization $\mathbf{W}$ for the manifold. As
discussed earlier, in the absence of any inner resonances, the
trivial choice (see eq.~\eqref{eq:red_dyn_trivial}) leads to the
simplest (i.e., linear) reduced dynamics, which is automatically obtained
from the normal form procedure. However, when the inner resonance
relations~\eqref{eq:internal_res} hold, then the dynamics cannot
be linearized. In such cases, we can compute the essential nonlinear
terms at degree $i$ following the normal form style of parametrization.
This involves first projecting the invariance equation~\eqref{eq:vectorized_inv_eqn}
onto $\mathrm{ker}(\boldsymbol{\mathcal{L}}_{i}^{\star})$ to eliminate
the unknowns $\mathbf{W}_{i}$ and then solving for the essential
nontrivial terms in $\mathbf{R}_{i}$ by computing a partial inverse
(see Murdock~\cite{Murdock2003}). In prior approaches, this is achieved
by transforming the governing equations into diagonal coordinates~\cite{Guckenheimer1983,Haro2016,Ponsieon2018,Ponsioen2020} which
causes the matrix $\boldsymbol{\mathcal{L}}_{i}$ to be diagonal and
hence simplifies the detection of its kernel. Here, we develop explicit
expressions for the computation of normal form directly in physical
coordinates using only the knowledge of the left eigenvectors $\mathbf{u}_{j}$
associated to the master subspace $E$, as summarized below. 

We focus on the case of a system with~$r_{i}\in\mathbb{N}_{0}$ inner
resonances and no outer resonances at order~$i$. Taking $\mathbf{D}=\mathcal{\boldsymbol{R}}_{i,i}^{\top}$
and $\mathbf{C}=\mathbf{I}_{M^{i}}$ in Proposition~\ref{thm:2},
we can directly estimate $\mathrm{ker}(\boldsymbol{\mathcal{L}}_{i}^{\star})$
using only the eigenvalues $\lambda_{j}$ and the corresponding eigenvectors
of the master subspace $E$. Specifically, the generalized
eigenvalues for the matrix pair $(\mathbf{I}_{M^{i}},\mathcal{\boldsymbol{R}}_{i,i}^{\top})$
are given by $\mu_{\boldsymbol{\ell}}=\frac{1}{\lambda_{\mathbf{\boldsymbol{\ell}}}}$
with left-eigenvectors~$\mathbf{e}_{\boldsymbol{\ell}}$ according
to eq.~\eqref{eq:eig_R1} and with the left kernel $\mathcal{N}_{i}$
of $\boldsymbol{\mathcal{L}}_{i}$ given as

\begin{equation}
\mathcal{N}_{i}:=\mathrm{ker}(\boldsymbol{\mathcal{L}}_{i}^{\star})=\textrm{span}\left(\left(\mathbf{e}_{\boldsymbol{\ell}}\otimes\mathbf{u}_{j}\right)\in\mathbb{C}^{NM^{i}}|\text{\ensuremath{\lambda_{\mathbf{\boldsymbol{\ell}}}}}=\lambda_{j},\boldsymbol{\ell}\in\Delta_{i,M},j\in\{1,\dots,M\}\right).\label{eq:kernel_i}
\end{equation}
Now, let $\mathbf{N}_{i}\in\mathbb{C}^{NM^{i}\times r_{i}}$ be a
basis for $\mathcal{N}_{i}$, which can be obtained by simply stacking
the column vectors $\left(\mathbf{e}_{\boldsymbol{\ell}}\otimes\mathbf{u}_{j}\right)$
from Definition~\eqref{eq:kernel_i} that are associated to the modes
with inner resonances (see eq.~\eqref{eq:internal_res}). Then, the
reduced dynamics coefficients in the normal form parametrization are
chosen by projecting the invariance equation~\eqref{eq:vectorized_inv_eqn}
onto $\mathbf{N}_{i}$ as 
\begin{equation}
\mathbf{N}_{i}^{\star}\boldsymbol{\mathcal{L}}_{i}\mathbf{w}_{i}=\mathbf{N}_{i}^{\star}\mathbf{h}_{i}(\mathbf{R}_{i}).\label{eq:projected_inv_eqn}
\end{equation}
The left-hand-side of eq.~\eqref{eq:projected_inv_eqn} is identically
zero since columns of $\mathbf{N}_{i}$ belong to $\mathrm{ker}(\boldsymbol{\mathcal{L}}_{i}^{*})$
, i.e., \begin{equation}
\mathbf{N}_{i}^{\star}\boldsymbol{\mathcal{L}}_{i}=\mathbf{0}.
\end{equation}
Hence, we are able to eliminate the unknowns $\mathbf{W}_{i}$ from
eq.~\eqref{eq:projected_inv_eqn} to obtain
\begin{equation}
\mathbf{N}_{i}^{\star}\mathbf{D}_{i}\mathrm{vec}\left(\mathbf{R}_{i}\right)=\mathbf{N}_{i}^{\star}\mathrm{vec}\left(\mathbf{C}_{i}\right).\label{eq:reduced_dynamics_normal}
\end{equation}
To solve eq.~\eqref{eq:reduced_dynamics_normal}, we may further simplify
it using the normalization~\eqref{eq:normalization} which results
in 
\begin{align}
\mathbf{N}_{i}^{\star}\mathbf{D}_{i} & =\text{\ensuremath{\mathbf{E}_{i}}}^{\top},\label{eq:relation}
\end{align}
where $\mathbf{E}_{i}\in\mathbb{R}^{M^{i+1}\times r_{i}}$ is a matrix
whose columns are of the form $\left(\mathbf{e}_{\boldsymbol{\ell}}\otimes\mathbf{e}_{j}\right)$,
such that $\left(\mathbf{\boldsymbol{\ell}},j\right)$ are pairs with
inner resonances, i.e., $\text{\ensuremath{\lambda_{\mathbf{\boldsymbol{\ell}}}}}=\lambda_{j}$
and $\mathbf{e}_{j}\in\mathbb{R}^{M}$ is the unit vector aligned
along the $j^{\mathrm{th}}$ coordinate axis. Here, the columns $\left(\mathbf{e}_{\boldsymbol{\ell}}\otimes\mathbf{e}_{j}\right)$
of $\mathbf{E}_{i}$ must be arranged analogous to the columns $\left(\mathbf{e}_{\boldsymbol{\ell}}\otimes\mathbf{u}_{j}\right)$
of $\mathbf{N}_{i}$. Using the relation~\eqref{eq:relation}, and
noting that $\mathbf{N}_{i}$ is a Boolean matrix with $\mathbf{E}_{i}\text{\ensuremath{\mathbf{E}_{i}}}^{\top}=\mathbf{I}$,
we obtain the canonical solution to~\eqref{eq:reduced_dynamics_normal}
for the coefficients $\mathbf{R}_{i}$ as 
\begin{equation}
\textrm{Normal-form style:}\quad\mathrm{vec}\left(\mathbf{R}_{i}\right)=\mathbf{E}_{i}\mathbf{N}_{i}^{\star}\mathrm{vec}\left(\mathbf{C}_{i}\right).\label{eq:normal_form}
\end{equation}
Note that for each inner resonant pair $\left(\mathbf{\boldsymbol{\ell}},j\right)$ in the definition of $ \mathcal{N}_i $~\eqref{eq:kernel_i}, the solution choice~\eqref{eq:normal_form} produces non-trivial coefficients in the $ j ^{\mathrm{th}}$ equation of reduced dynamics~\eqref{eq:red_dyn} precisely for the monomial $ p_{\ell_1}\dots p_{\ell_i} $ corresponding to an inner resonance. As a result, eq.~\eqref{eq:normal_form} directly provides the normal form coefficients of the reduced dynamics on the manifold in physical
coordinates, using only the knowledge of the master modes spanning
the adjoint modal subspace $E^{\star}$. 

\begin{rem}
	\label{rem:near_resonances}(Near-resonances) As the resonance relations
	\eqref{eq:internal_res}-\eqref{eq:cross_res} are meant for the real
	as well as imaginary parts of the eigenvalues simultaneously, these
	are seldom satisfied exactly. However, in lightly damped systems (where
	the real parts of the eigenvalues are small), near-resonances might
	exist between the imaginary parts of the eigenvalues. In such cases,
	it is desirable to include the corresponding near-resonant modes in
	the normal form parametrization of the reduced dynamics otherwise
	it leads to \emph{small divisors} (ill-conditioning) in solving system~\eqref{eq:vectorized_inv_eqn}
	and the domain of the validity of the Taylor approximation to the
	manifold shrinks (see, e.g., Guckenheimer \& Holmes~\cite{Guckenheimer1983}). 
\end{rem}

\subsubsection*{Graph style parametrization}

As the name suggests, a graph style of parametrization for the reduced
dynamics is the result of expressing the manifold as a graph over
the master subspace $E$ (see Haro et al~\cite{Haro2016}),
as done in the graph transform method. The graph style of parametrization
may be appealing in the context of center manifold computation, where
an infinite number of inner resonances may arise. For instance, in
a system with a two-dimensional center subspace with eigenvalues $\lambda_{1,2}=\pm\mathrm{i}\omega$,
its center manifold exhibits the inner resonances 
\begin{align}
\lambda_{1} & =\left(\ell+1\right)\lambda_{1}+\ell\lambda_{2},\\
\lambda_{2} & =\ell\lambda_{1}+\left(\ell+1\right)\lambda_{2},\quad\forall\ell\in\mathbb{N}.
\end{align}
In our setting, a graph style of parametrization is achieved by projecting
the invariance equations~\eqref{eq:vectorized_inv_eqn} onto the subspace
$\mathcal{G}_{i}$ defined as
\begin{equation}
\mathcal{G}_{i}:=\textrm{span}\left(\left(\mathbf{e}_{\boldsymbol{\ell}}\otimes\mathbf{u}_{j}\right)\in\mathbb{C}^{NM^{i}},\boldsymbol{\ell}\in\Delta_{i,M},j\in\{1,\dots,M\}\right).\label{eq:subspace_graph}
\end{equation}
Note that in the case of inner resonances, $\mathcal{N}_{i}\subset\mathcal{G}_{i}$
(cf. Definition~\eqref{eq:kernel_i}) and, hence, $\mathcal{G}_{i}$
includes all possible resonant subspaces at order $i$.
Then, similarly to the normal form style, we define a basis $\mathbf{G}_{i}\in\mathbb{C}^{NM^{i}\times M^{i}}$
for $\mathcal{G}_{i}$, obtained by stacking the column vectors $\left(\mathbf{e}_{\boldsymbol{\ell}}\otimes\mathbf{u}_{j}\right)$
in the definition~\eqref{eq:subspace_graph}. We obtain a graph style
of parametrization by projecting~\eqref{eq:vectorized_inv_eqn}
on to $\mathcal{G}_{i}$ and equating the right-hand side to zero
as

\begin{equation}
\mathbf{G}_{i}^{\star}\mathbf{h}_{i}(\mathbf{R}_{i})=\mathbf{0}.\label{eq:project_inv_eqn_graph}
\end{equation}
In contrast to the normal form style~\eqref{eq:normal_form}, where only the coefficients of resonant monomials are nontrivial, we generally obtain a larger set of monomials with nontrivial coefficients in the graph style. Hence, the normal form style retains only the minimal number of nonlinear terms in the reduced that are essential for solving the invariance equation~\eqref{eq:inv_eqn} at each order $ i $, whereas the graph style generally leads to more complex expressions of the reduced dynamics. 

For the choice of $\mathbf{W}_{1}=\mathbf{V}_{E}$ (see eq.
\eqref{eq:W1}) and the normalization condition~\eqref{eq:normalization},
eq.~\eqref{eq:project_inv_eqn_graph} can be simplified to obtain
the reduced dynamics coefficients in the graph style as 
\begin{equation}
\textrm{Graph style:}\quad\mathbf{R}_{i}=\mathbf{U}_{E}^{\star}\mathbf{C}_{i}.\label{eq:reduced_dyn_graph}
\end{equation}
Note that eq.~\eqref{eq:reduced_dyn_graph} directly provides the
reduced dynamics coefficients in the graph style without the evaluation
of $\mathbf{G}_{i}$. Hence, an advantage of using the graph style
parametrization relative to the normal form style is that specific
inner resonances need not be identified a priori.

More generally, a combination of graph and normal form styles of parametrization
may also be used depending on the problem. This is referred to as
a \emph{mixed} style of parametrization as discussed by Haro et al.~\cite{Haro2016}. A mixed style may be particularly appealing in the
context of parameter-dependent manifold computation, where the parameters
are dummy dynamic variables and the associated modes always have trivial
dynamics. Thus, it is desirable to choose a graph style for the parametric
modes and a normal form style for remaining master modes which may
feature inner resonances (see Haro et al.~\cite{Haro2016}, Murdock~\cite{Murdock2003}). 

\subsubsection*{Computing the parametrization coefficients $\mathbf{w}_{i}$}

Once the reduced dynamics coefficients $\mathbf{R}_{i}$ (specific
to the choice of parametrization style) are determined (see eqs.~\eqref{eq:normal_form},
\eqref{eq:reduced_dyn_graph}), we can compute the manifold parametrization
coefficients $\mathbf{W}_{i}$ by solving eq.~\eqref{eq:ith_order_inv_vec}. 

When the coefficient matrix $\boldsymbol{\mathcal{L}}_{i}$ is (nearly)
singular in the presence of (near) resonances, numerical blow-up errors (ill-conditioning)
may occur due to the small divisors that arise in solving system~\eqref{eq:ith_order_inv_vec}
using conventional solvers (see also Remark~\ref{rem:near_resonances}).
As an alternative, we adopt a norm-minimizing solution to~\eqref{eq:ith_order_inv_vec}
given by
\begin{equation}
\mathbf{w}_{i}=\min_{\mathbf{x}\in\mathbb{C}^{NM^{i}},\boldsymbol{\mathcal{L}}_{i}\mathbf{x}=\mathbf{h}_{i}(\mathbf{R}_{i})}\|\mathbf{x}\|^{2},\label{eq:lsqminnorm}
\end{equation}
which can be obtained using existing routines, such as the $\texttt{lsqminnorm}$
in MATLAB. Other commonly used techniques in the literature include
the simultaneous solution of equations~\eqref{eq:vectorized_inv_eqn}
and~\eqref{eq:reduced_dynamics_normal} or~\eqref{eq:subspace_graph}.
This involves the inversion of a bordered matrix that extends $\boldsymbol{\mathcal{L}}_{i}$
and ends up being non-singular (see, e.g., Beyn \& Kleß~\cite{Beyn1998},
Kuznetsov~\cite{Kuznetsov2004}).

To summarize, we have developed an automated procedure
for computing invariant manifolds attached to fixed points of system
\eqref{eq:first_order_aut} and for choosing different styles of parametrizations
for their reduced dynamics (eqs.~\eqref{eq:normal_form},~\eqref{eq:reduced_dyn_graph})
by solving invariance equations~\eqref{eq:vectorized_inv_eqn} in
the physical coordinates and only using the eigenvectors associated
to its master master spectral subspace. The open-source MATLAB package~\cite{SSMTool2.0} automates this computational procedure. Next, we
illustrate applications of this computation procedure developed so
far.

\subsection{Applications}

\subsubsection*{Parameter-dependent center manifolds and their reduced dynamics}

We illustrate the automated procedure developed above to compute the
center manifold in the Lorenz system and its normal form style of reduced dynamics without performing any diagonalization and using only the modes associated to the center subspace. In the following example, we
compute the $\rho$-dependent center manifold and the normal form
of the reduced dynamics to analyze the local bifurcation around $\rho=1$
(see section 3.2 in Guckenheimer \& Holmes~\cite{Guckenheimer1983}).
Consider the Lorenz system 
\begin{equation}
\begin{aligned}\dot{x} & =\sigma(y-x),\\
\dot{y} & =\rho x-y-xz,\\
\dot{z} & =-\beta z+xy,
\end{aligned}
\label{eq:lorenz}
\end{equation}
where $(x,y,z)\in\mathbb{R}^{3},\sigma,\rho,\beta>0$. The basic steps
are as follows. 
\begin{enumerate}
	\item Setup: With a new variable $\mu:=\rho-1$ in the Lorenz system~\eqref{eq:lorenz},
	we obtain an extended system of the form~\eqref{eq:first_order_aut}
	with
	\begin{equation}
	\mathbf{z}=\left[\begin{array}{c}
	x\\
	y\\
	z\\
	\mu
	\end{array}\right],\quad\mathbf{A}=\left[\begin{array}{cccc}
	-\sigma & \sigma & 0 & 0\\
	1 & -1 & 0 & 0\\
	0 & 0 & -\beta & 0\\
	0 & 0 & 0 & 0
	\end{array}\right],\quad\mathbf{B}=\mathbf{I}_{4},\quad\mathbf{F}(\mathbf{z})=\left[\begin{array}{c}
	0\\
	x\mu-xz\\
	xy\\
	0
	\end{array}\right].\label{eq:lorenz_extended}
	\end{equation}
	For $\sigma=\beta=1$, the eigenvalues of $\mathbf{A}$ are given
	by $\lambda_{1}=\lambda_{2}=0$, $\lambda_{3}=-2$ , and $\lambda_{4}=1$.
	The nonlinearity $\mathbf{F}$ is quadratic and can be expressed according
	to the Kronecker expansion~\eqref{eq:expF} as
	\[
	\mathbf{F}(\mathbf{z})=\mathbf{F}_{2}\mathbf{z}^{\otimes2},
	\]
	where the term $\mathbf{z}^{\otimes2}=(\mathbf{z}\otimes\mathbf{z})=(x^{2},xy,xz,x\mu,yx,y^{2},yz,y\mu,zx,zy,z^{2},z\mu,\mu x,\mu y,\mu z,\mu^{2})^{T}$
	contains the monomials and $\mathbf{F}_{2}\in\mathbb{R}^{4\times4^{2}}$
	is a sparse matrix representation of the coefficients of the quadratic
	nonlinearity. The non-zero entries of $\mathbf{F}_{2}$ corresponding
	to the monomials $x\eta$, $xz$ and $xy$ in the definition of $\mathbf{F}$
	(see eq.~\eqref{eq:lorenz_extended}) are given as 
	\[
	\left(\mathbf{F}_{2}\right)_{24}=1,\quad\left(\mathbf{F}_{2}\right)_{23}=-1,\quad\left(\mathbf{F}_{2}\right)_{32}=1.
	\]
	\item Choose master subspace: We construct a center manifold over the center-subspace
	$E$ spanned by the eigenvectors corresponding to the two
	zero eigenvalues. We obtain the eigenvectors associated with this
	$E$ satisfying the normalization condition~\eqref{eq:normalization},
	as 
	\begin{align}
	\boldsymbol{\Lambda}_{E} & =\left[\begin{array}{cc}
	0 & 0\\
	0 & 0
	\end{array}\right],\quad\mathbf{V}_{E}=\mathbf{U}_{E}=\left[\begin{array}{cc}
	\begin{array}{c}
	\frac{1}{\sqrt{2}}\\
	\frac{1}{\sqrt{2}}\\
	0\\
	0
	\end{array} & \begin{array}{c}
	0\\
	0\\
	0\\
	1
	\end{array}\end{array}\right].\label{eq:eigs_lorenz}
	\end{align}
	\item Assemble invariance equations: At leading-order, the parametrization
	coefficients for the center manifold and its reduced dynamics can
	be simply chosen as (see eqs.~\eqref{eq:eigs_lorenz} and~\eqref{eq:W1})
	\[
	\mathbf{W}_{1}=\mathbf{V}_{E},\quad\mathbf{R}_{1}=	\boldsymbol{\Lambda}_{E}= \mathbf{0}\in\mathbb{R}^{2\times2}.
	\]
	To obtain the parametrization coefficients at order $2$, we need
	to solve the vectorized invariance equation~\eqref{eq:vectorized_inv_eqn}
	for $i=2$, i.e., 
	\begin{equation}
	\boldsymbol{\mathcal{L}}_{2}\mathrm{vec}\left(\mathbf{W}_{2}\right)=\mathrm{vec}\left(\mathbf{C}_{2}\right)-\left(\mathbf{I}_{4}\otimes\mathbf{W}_{1}\right)\mathrm{vec}\left(\mathbf{R}_{2}\right),\label{eq:inv_lorenz_2}
	\end{equation}
	where 
	\begin{align*}
	\boldsymbol{\mathcal{L}}_{2} & =\left(\boldsymbol{\mathcal{R}}_{2,2}^{\top}\otimes\mathbf{I}_{4}\right)-\left(\mathbf{I}_{4}\otimes\mathbf{A}\right)\in\mathbb{R}^{64\times64},\\
	\mathbf{C}_{2} & =\mathbf{F}_{2}\mathbf{W}_{1}^{\otimes2}.
	\end{align*}
	\item Resonance detection: We deduce the eigenvalues of $\boldsymbol{\mathcal{R}}_{2,2}^{\top}$
	from the formula~\eqref{eq:eig_R1} as 
	\[
	\lambda_{(i,j)}=\lambda_{i}+\lambda_{j}=0,\quad\forall i,j\in\{1,2\}.
	\]
	Thus, as per definition~\eqref{eq:internal_res}, we observe inner-resonances
	between the two (zero) eigenvalues of the center-subspace, i.e.,
	\[
	\lambda_{(i,j)}=\lambda_{k}=0,\quad\forall i,j,k\in\{1,2\},
	\]
	and no outer-resonances, i.e.,
	\[
	\lambda_{(i,j)}\neq\lambda_{k}=0,\quad\forall i,j\in\{1,2\},\quad k\in\{3,4\}.
	\]
	\item Choice of parametrization: Hence, the singular system~\eqref{eq:inv_lorenz_2}
	is solvable for a nontrivial choice of reduced dynamics the $\mathbf{R}_{2}$.
	Using eq.~\eqref{eq:normal_form}, we choose the normal form parametrization
	of the reduced dynamics which results in the following non-zero entry
	in the coefficient array $\mathbf{R}_{2}$:
	\[
	\left(\mathbf{R}_{2}\right)_{12}=\frac{1}{2}.
	\]
	\item Recursion: This procedure can be recursively applied to obtain higher-order
	terms on the center manifold dynamics. The reduced dynamics on the
	two-dimensional center manifold up to cubic terms is given as
	\begin{equation}
	\dot{\mathbf{p}}=\left[\begin{array}{c}
	\dot{p}_{1}\\
	\dot{p}_{2}
	\end{array}\right]=\left[\begin{array}{c}
	\frac{1}{2}p_{1}p_{2}+\frac{1}{4}p_{1}^{3}-\frac{1}{8}p_{1}p_{2}^{2}\\
	0
	\end{array}\right].\label{eq:red_lorenz}
	\end{equation}
	Here, the variable $p_{2}$ is the modal coordinate along the center
	direction associate to the parameter $\mu$. The normal form parametrization
	automatically results in trivial dynamics along this direction. Indeed,
	the near-identity transformation associated to the normal form leaves
	the coordinate $\mu$-mode unchanged, which prompts us to replace
	$p_{2}$ by $\mu$ in eq.~\eqref{eq:red_lorenz}. Hence, we obtain
	the parameter-dependent dynamics on the center manifold as
	\[
	\dot{p}=R_{\mu}(p)=\frac{1}{2}p\mu-\frac{1}{8}p\mu^{2}+\frac{1}{4}p^{3},
	\]
	which recovers the pitchfork bifurcation (see section 3.4 in Guckenheimer
	\& Holmes~\cite{Guckenheimer1983}) with respect to the parameter
	$\mu$.
\end{enumerate}
For more involved applications to center manifold computation, we
refer to the work of Carini et al.~\cite{Carini2015}, who analyze
the stability of bifurcating flows using a similar methodology for
computing parameter-dependent center manifold and normal forms.

\subsubsection*{Lyapunov subcenter manifolds and conservative backbone curves}

The Lyapunov subcenter manifolds (LSM) form centerpieces of periodic
response in conservative, unforced, mechanical systems (see Kerschen
et al.~\cite{Kerschen2006}, de la Llave \& Kogelbauer~\cite{dela2019}).
We discuss how the above methodology can be applied in such systems
to compute LSMs and directly extract conservative backbone curves, i.e.,
the functional relationship between amplitudes and frequency of the
periodic orbits on the LSM. 

We consider the following form of a conservative mechanical system
\begin{equation}
\mathbf{M}\ddot{\mathbf{x}}+\mathbf{K}\mathbf{x}+\mathbf{f}(\mathbf{x},\dot{\mathbf{x}})=\mathbf{0},\quad\mathbf{x}\in\mathbb{R}^{n},\label{eq:conservative}
\end{equation}
where $\mathbf{M},\mathbf{K}\in\mathbb{R}^{n\times n}$ are positive
definite mass and stiffness matrices and $\mathbf{f}=\mathcal{O}(\|\mathbf{x}\|^{2},\|\dot{\mathbf{x}}\|^{2},\|\mathbf{x}\|\|\dot{\mathbf{x}}\|)$
is a conservative nonlinearity. The quadratic eigenvalue problem 
\begin{equation}
\mathbf{K}\mathbf{\boldsymbol{\varphi}}_{j}=\omega_{j}^{2}\mathbf{M}\mathbf{\boldsymbol{\varphi}}_{j},\quad j=1,\dots,n\label{eq:undamped_evp}
\end{equation}
provides us the vibration modes $\mathbf{\boldsymbol{\varphi}}_{j}\in\mathbb{R}^n$
and the corresponding natural frequencies $\omega_{j}$ of system~\eqref{eq:conservative}. In the first-order form~\eqref{eq:L2} with
$\mathbf{C}=\mathbf{0},\mathbf{N}=\mathbf{M}$, the eigenvalues and
eigenvectors can be expressed using eq.~\eqref{eq:undamped_evp} as
\begin{align}
\lambda_{2j-1} & =\mathrm{i}\omega_{j},\quad\lambda_{2j}=\bar{\lambda}_{2j-1},\\
\mathbf{v}_{2j-1} & =\left[\begin{array}{c}
\mathbf{\boldsymbol{\varphi}}_{j}\\
\lambda_{2j-1}\mathbf{\boldsymbol{\varphi}}_{j}
\end{array}\right],\quad\mathbf{v}_{2j}=\bar{\mathbf{v}}_{2j-1},\quad j=1,\dots,n.
\end{align}

Any distinct pair of eigenvalues $\pm\mathrm{i}\omega_{m}$, where
$m=1,\dots,n$, spans a two-dimensional linear modal subspace. An LSM
is a unique, analytic, two-dimensional, nonlinear extension to such
a linear modal subspace and is guaranteed to exist if the master eigenfrequency
$\omega_{m}$ is not in resonance with any of the remaining eigenfrequencies
of the system (Kelley~\cite{Kelley1969}), i.e, under the non-resonance
conditions 
\begin{equation}
\omega_{m}\neq k\omega_{i},\quad\forall k\in\mathbb{Z},\quad i=\{1,\dots,n\}\backslash m.\label{eq:nonres_LSM}
\end{equation}
The LSM over the $m^{\mathrm{th}}$ mode can be computed by solving
the invariance equation~\eqref{eq:inv_eqn} in the physical coordinates
using only the master mode $\mathbf{\boldsymbol{\varphi}}_{m}$ that
spans the two-dimensional modal subspace $E=\mathrm{span}\left(\mathbf{v}_{2m-1},\mathbf{v}_{2m}\right)$.
The leading order coefficients in the parametrizations for the LSM
and its reduced dynamics are given by eq.~\eqref{eq:W1} as 
\begin{align}
\mathbf{W}_{1} & =[\mathbf{v}_{2m-1},\mathbf{v}_{2m}],\\
\mathbf{R}_{1} & =\mathbf{\Lambda}_{E}=\mathrm{diag}(\mathrm{i}\omega_{m},-\mathrm{i}\omega_{m}).
\end{align}
Note that for any $\ell\in\mathbb{N}$, the master subspace $E$
satisfies the inner resonance relations
\begin{align}
\lambda_{2m-1} & =\left(\ell+1\right)\lambda_{2m-1}+\ell\lambda_{2m},\\
\lambda_{2m} & =\ell\lambda_{2m-1}+\left(\ell+1\right)\lambda_{2m},
\end{align}
which result in the following reduced dynamics in the normal-form
parametrization style (see eq.~\eqref{eq:reduced_dynamics_normal})
\begin{equation}
\dot{\mathbf{p}}=\mathbf{R}(\mathbf{p})=\left[\begin{array}{c}
i\omega_{m}p\\
-i\omega_{m}\bar{p}
\end{array}\right]+\sum_{\ell\in\mathbb{N}}\left[\begin{array}{c}
\gamma_{\ell}p^{\ell+1}\bar{p}^{\ell}\\
\bar{\gamma}_{\ell}p^{\ell}\bar{p}^{\ell+1}
\end{array}\right],\label{eq:red_2D_cons}
\end{equation}
where the $\gamma_{\ell}$ are the nontrivial coefficients associated to the monomials in the normal form~\eqref{eq:reduced_dynamics_normal}. Then, the following statement directly provides us the conservative
backbone associated to the $m^{\mathrm{th}}$ mode. 
\begin{lem}
	\label{lemma:cons_bb}Under the non-resonance condition~\eqref{eq:nonres_LSM},
	the backbone curve, i.e., the functional relationship between the
	polar response amplitude $\rho$ and the oscillation frequency $\omega$
	of the periodic orbits of the LSM associated to the mode $\mathbf{\boldsymbol{\varphi}}_{m}$
	of the conservative mechanical system~\eqref{eq:conservative} is
	given as 
	\begin{equation}
	\omega(\rho)=\omega_{m}+\sum_{\ell\in\mathbb{N}}\mathrm{Im}(\gamma_{\ell})\rho^{2\ell}.\label{eq:cons_backbone}
	\end{equation}
\end{lem}
\begin{proof}
	See Appendix~\ref{sec:Proof-of-Lemma}.
\end{proof}

\section{Invariant manifold{s} and their reduced dynamics under non-autonomous forcing}

\label{sec:computations_non-autonomous}

In the non-autonomous setting of system~\eqref{eq:first_order}, i.e.,
for $\epsilon>0$, the fixed point is typically replaced by an invariant
torus created by the quasiperiodic term $\epsilon\mathbf{F}^{ext}(\mathbf{z},\boldsymbol{\Omega}t)$,
or by a periodic orbit when $\Omega$ is one-dimensional. Indeed,
for small enough $\epsilon>0$, the existence of a small-amplitude
invariant torus $\gamma_{\epsilon}$ in the extended phase space of
system~\eqref{eq:first_order} is guaranteed if the origin is a hyperbolic
fixed point in its $\epsilon=0$ limit (see Guckenheimer \& Holmes~\cite{Guckenheimer1983}). In this setting, we have an invariant manifold
$\mathcal{W}(E,\gamma_{\epsilon})$, which can be viewed
as a fiber bundle that perturbs smoothly from the spectral subbundle
$\gamma_{\epsilon}\times E$ under the addition
of the nonlinear terms, as long as appropriate resonance conditions
hold (see Theorem~4 of Haller \& Ponsioen~\cite{Haller2016}, Theorem~4.1 of Haro \& de la Llave~\cite{Haro2006,Haro2006a,Haro2007}). 

\begin{figure}[H]
	\centering \includegraphics[width=0.8\linewidth]{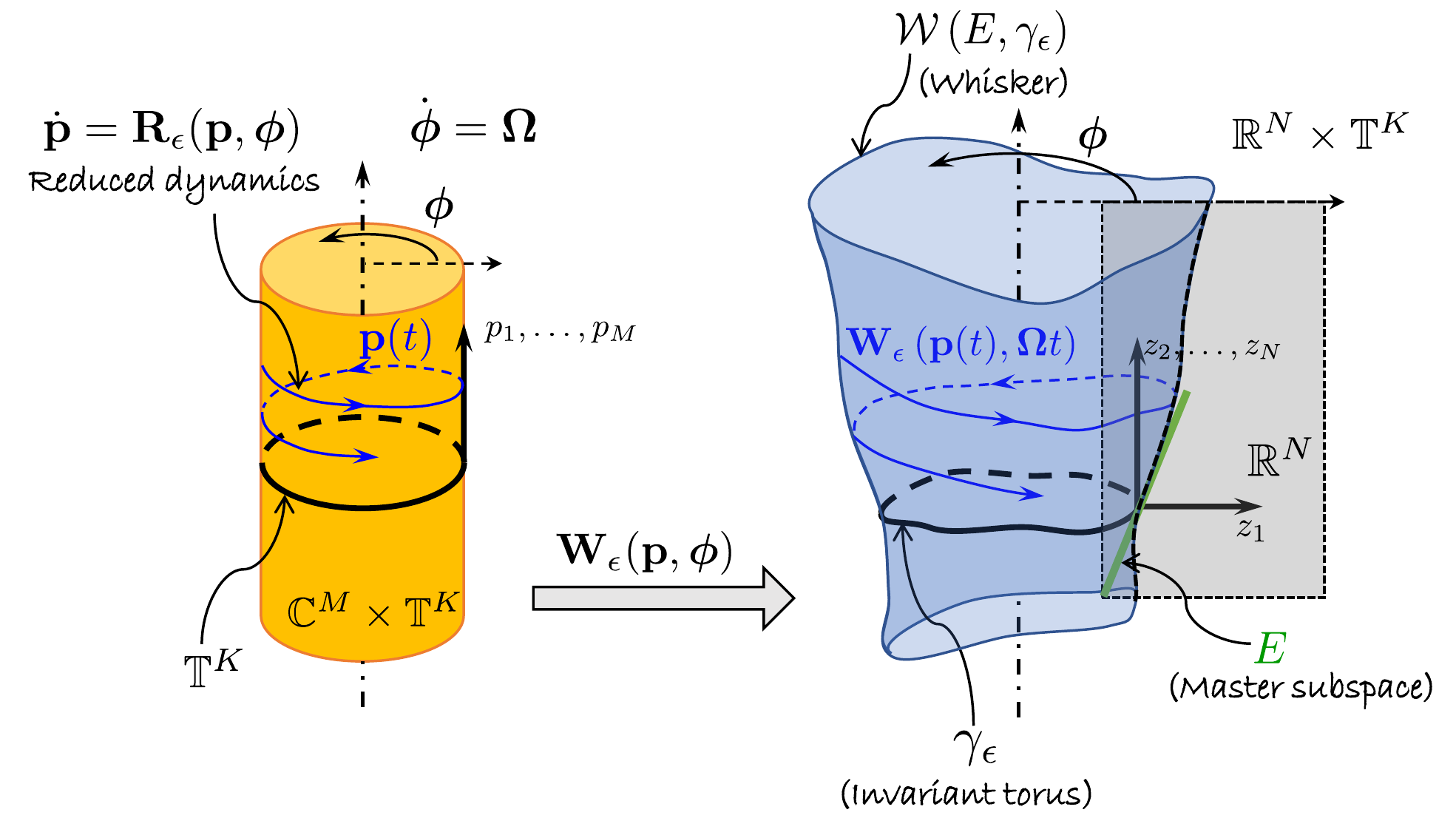}
	\caption{\label{fig:param-whisker} \small Applying the parametrization method to
		system~\eqref{eq:first_order} with $\epsilon>0$, we obtain a parametrization
		$\mathbf{W}_{\epsilon}:\mathbb{C}^{M}\times\mathbb{T}^{K}\to\mathbb{R}^{N}$
		of an $(M+K)-$dimensional perturbed manifold (whisker) attached
		to a small amplitude whiskered torus $\gamma_{\epsilon}$ parameterized
		by the angular variables $\boldsymbol{\phi}\in\mathbb{T}^{K}$ with
		$\dot{\boldsymbol{\phi}}=\mathbf{\Omega}$. This whisker is perturbed
		from the master spectral subspace $E$ of the linear system
		\eqref{eq:first_order_linear} under the addition of nonlinear and
		$\mathcal{O}(\epsilon)$ terms of system~\eqref{eq:first_order} (cf.
		Figure~\ref{fig:param}). Furthermore, we have the freedom to choose
		a parametrization $\mathbf{R}_{\epsilon}:\mathbb{C}^{M}\times\mathbb{T}^{K}\to\mathbb{C}^{M}$
		of the reduced dynamics on the manifold such that the function $\mathbf{W}_{\epsilon}$
		also maps the reduced system trajectories $\mathbf{p}(t)$ onto the
		full system trajectories on the invariant manifold, i.e., $\mathbf{z}(t)=\mathbf{W}_{\epsilon}\left(\mathbf{p}(t),\mathcal{\mathbf{\Omega}}t\right)$.}
\end{figure}

In contrast to the invariant manifold $\mathcal{W}(E)$
from the autonomous setting, the perturbed manifold or \emph{whisker},
$\mathcal{W}(E,\gamma_{\epsilon})$, is attached to $\gamma_{\epsilon}$
instead of the origin and $\dim\left(\mathcal{W}(E,\gamma_{\epsilon})\right)=\dim\left(E\right)+\dim\left(\gamma_{\epsilon}\right)=M+K$,
as shown in Figure~\ref{fig:param-whisker}. From a computational
viewpoint, now the manifold $\mathcal{W}(E,\gamma_{\epsilon})$
and its reduced dynamics need to be additionally parameterized by
the angular variables $\boldsymbol{\phi}\in\mathbb{T}^{K}$ that correspond
to the multi-frequency vector $\boldsymbol{\Omega}\in\mathbb{R}^{K}$
as 
\begin{align}
\mathbf{W}_{\epsilon}(\mathbf{p},\boldsymbol{\phi}) & =\mathbf{W}(\mathbf{p})+\epsilon\mathbf{X}(\mathbf{p},\boldsymbol{\phi})+\mathcal{O}\left(\epsilon^{2}\right),\label{eq:Weps}\\
\mathbf{R}_{\epsilon}(\mathbf{p},\boldsymbol{\phi}) & =\mathbf{R}(\mathbf{p})+\epsilon\mathbf{S}(\mathbf{p},\boldsymbol{\phi})+\mathcal{O}\left(\epsilon^{2}\right).\label{eq:Reps}
\end{align}
Here, $\mathbf{W}_{\epsilon}:\mathbb{C}^{M}\times\mathbb{T}^{K}\to\mathbb{R}^{N}$,
$\mathbf{R}_{\epsilon}:\mathbb{C}^{M}\times\mathbb{T}^{K}\to\mathbb{C}^{M}$
are parametrizations for the invariant manifold $\mathcal{W}(E,\gamma_{\epsilon})$
and its reduced dynamics; $\mathbf{W}(\mathbf{p}),\mathbf{R}(\mathbf{p})$
recover the manifold $\mathcal{W}(E)$ and its reduced dynamics
in the unforced limit of $\epsilon=0$; and $\mathbf{X}(\mathbf{p},\boldsymbol{\phi}),\mathbf{S}(\mathbf{p},\boldsymbol{\phi})$
denote the $\mathcal{O}(\epsilon)$ terms, which depend on the angular
variables $\boldsymbol{\phi}$ due to the presence of forcing $\mathbf{F}^{ext}(\mathbf{z},\boldsymbol{\Omega}t)$.
Invoking the invariance of $\mathcal{W}(E,\gamma_{\epsilon})$,
we substitute the expansions~\eqref{eq:Weps}-\eqref{eq:Reps} into
the governing equations~\eqref{eq:first_order} and collect the $\mathcal{O}(\epsilon)$
terms to obtain (cf. Ponsioen et al.~\cite{Ponsioen2020})
\begin{equation}
\mathbf{B}\left[D\mathbf{W}(\mathbf{p})\mathbf{S}(\mathbf{p},\boldsymbol{\phi})+\partial_{\mathbf{p}}\mathbf{X}(\mathbf{p},\boldsymbol{\phi})\mathbf{R}(\mathbf{p})+\partial_{\mathbf{\boldsymbol{\phi}}}\mathbf{X}(\mathbf{p},\boldsymbol{\phi})\cdot\boldsymbol{\Omega}\right]=\left[\mathbf{A}+D\mathbf{F}(\mathbf{W}(\mathbf{p}))\right]\mathbf{X}(\mathbf{p},\boldsymbol{\phi})+\mathbf{F}^{ext}(\boldsymbol{\phi},\mathbf{W}(\mathbf{p})).\label{eq:inv_non-autonomous}
\end{equation}
The terms $\mathbf{X}(\mathbf{p},\boldsymbol{\phi}),\mathbf{S}(\mathbf{p},\boldsymbol{\phi})$
can be further expanded into Taylor series in $\mathbf{p}$ with coefficients
that depend on the angular variables $\boldsymbol{\phi}$ as 
\begin{equation}
\mathbf{X}(\mathbf{p},\boldsymbol{\phi})=\mathbf{X}_{0}(\boldsymbol{\phi})+\sum_{j=1}^{\Gamma_{S}}\mathbf{X}_{j}(\boldsymbol{\phi})\mathbf{p}^{\otimes j},\label{eq:S1_i}
\end{equation}

\begin{equation}
\mathbf{S}(\mathbf{p},\boldsymbol{\phi})=\mathbf{S}_{0}(\boldsymbol{\phi})+\sum_{j=1}^{\Gamma_{R}}\mathbf{S}_{j}(\boldsymbol{\phi})\mathbf{p}^{\otimes j}.\label{eq:R1_i}
\end{equation}
Collecting the $\mathcal{O}(1)$ terms in $\mathbf{p}$ from the invariance
equation~\eqref{eq:inv_non-autonomous}, we obtain
\begin{equation}
\mathbf{B}\left[\mathbf{W}_{1}\mathbf{S}_{0}(\boldsymbol{\phi})+\partial_{\mathbf{\boldsymbol{\phi}}}\mathbf{X}_{0}(\boldsymbol{\phi})\cdot\boldsymbol{\Omega}\right]=\mathbf{A}\mathbf{X}_{0}(\boldsymbol{\phi})+\mathbf{F}^{ext}(\boldsymbol{\phi}),\label{eq:leading_order_forced}
\end{equation}
which is a system of linear differential equations for the unknown,
time-dependent coefficients $\mathbf{X}_{0}(\boldsymbol{\phi})$.
Similarly to the autonomous setting, the choice of reduced dynamics
$\mathbf{S}_{0}(\boldsymbol{\phi})$ again provides us the freedom
to remove (near-) resonant terms via a normal-form style of parametrization. 

In this work, we restrict our attention to the computation of the
leading-order non-autonomous contributions, i.e., $\mathbf{X}_{0}(\boldsymbol{\phi}),\mathbf{S}_{0}(\boldsymbol{\phi})$.
To this end, we perform a Fourier expansion of the different terms
in eq.~\eqref{eq:leading_order_forced} as

\begin{equation}
\mathbf{F}^{ext}(\mathbf{z},\boldsymbol{\phi})=\sum_{\boldsymbol{\kappa}\in\mathbb{Z}^{K}}\mathbf{F}_{0,\boldsymbol{\kappa}}e^{\mathrm{i}\left\langle \boldsymbol{\kappa},\boldsymbol{\phi}\right\rangle }+\mathcal{O}(|\mathbf{z}|),\label{eq:Fourier_ext_force}
\end{equation}
\begin{equation}
\mathbf{X}_{0}(\boldsymbol{\phi})=\sum_{\boldsymbol{\kappa}\in\mathbb{Z}^{K}}\mathbf{x}_{0,\boldsymbol{\kappa}}e^{\mathrm{i}\left\langle \boldsymbol{\kappa},\boldsymbol{\phi}\right\rangle },\label{eq:Fourier_X0}
\end{equation}
\begin{equation}
\mathbf{S}_{0}(\boldsymbol{\phi})=\sum_{\boldsymbol{\kappa}\in\mathbb{Z}^{K}}\mathbf{s}_{0,\boldsymbol{\kappa}}e^{\mathrm{i}\left\langle \boldsymbol{\kappa},\boldsymbol{\phi}\right\rangle },\label{eq:Fourier_S0}
\end{equation}
where $\mathbf{F}_{0,\boldsymbol{\kappa}}\in\mathbb{C}^{N}$ are the
known Fourier coefficients for the forcing $\mathbf{F}^{ext}(\mathbf{z},\boldsymbol{\Omega}t)$
and $\mathbf{x}_{0,\boldsymbol{\kappa}}\in\mathbb{C}^{N}$; and $\mathbf{s}_{0,\boldsymbol{\kappa}}\in\mathbb{C}^{M}$
are the unknown Fourier coefficients for the leading-order, non-autonomous
components of $\mathbf{X},\mathbf{S}$. Upon substituting eqs.~\eqref{eq:Fourier_ext_force}-\eqref{eq:Fourier_S0}
into eq.~\eqref{eq:leading_order_forced} and comparing Fourier coefficients
at order $\boldsymbol{\kappa}$, we obtain linear equations in terms
of the variables $\mathbf{x}_{0,\boldsymbol{\kappa}},\mathbf{s}_{0,\boldsymbol{\kappa}}$
as

\begin{align}
\boldsymbol{\mathcal{L}}_{0,\boldsymbol{\kappa}}\mathbf{x}_{0,\boldsymbol{\kappa}} & =\mathbf{h}_{0,\boldsymbol{\kappa}}(\mathbf{s}_{0,\boldsymbol{\kappa}}),\quad\boldsymbol{\kappa}\in\mathbb{Z}^{K},\label{eq:vec_inv_non_aut}
\end{align}
where
\begin{align*}
\boldsymbol{\mathcal{L}}_{0,\boldsymbol{\kappa}} & :=\mathrm{i}\left\langle \boldsymbol{\kappa},\boldsymbol{\Omega}\right\rangle \mathbf{B}-\mathbf{A}\in\mathbb{C}^{N\times N},\\
\mathbf{h}_{0,\boldsymbol{\kappa}}(\mathbf{s}_{0,\boldsymbol{\kappa}}) & :=\mathbf{F}_{0,\boldsymbol{\kappa}}-\mathbf{B}\mathbf{W}_{1}\mathbf{s}_{0,\boldsymbol{\kappa}}\in\mathbb{C}^{N}.
\end{align*}
The coefficient matrix $\boldsymbol{\mathcal{L}}_{0,\boldsymbol{\kappa}}$
in~\eqref{eq:vec_inv_non_aut} becomes (nearly) singular when the forcing
is (nearly) resonant with any of eigenvalues of the system ($\mathbf{A},\mathbf{B})$, i.e., when
\begin{equation}
\mathrm{i}\left\langle \boldsymbol{\kappa},\boldsymbol{\Omega}\right\rangle \approx\lambda_{j},\quad j\in\{1,\dots,N\}.\label{eq:ext_res}
\end{equation}
Similarly to the autonomous setting, such nearly resonant forcing
leads to small divisors while we are solving system~\eqref{eq:vec_inv_non_aut}
(cf. Remark~\ref{rem:near_resonances}) and hence it is desirable
to include such terms in the reduced dynamics as per the normal form
style of parametrization. This results in (cf. eq.~\eqref{eq:normal_form})
\begin{equation}
\textrm{Normal-form style:}\quad\mathbf{s}_{0,\boldsymbol{\kappa}}=\mathbf{e}_{j}\left(\mathbf{u}_{j}^{\star}\mathbf{F}_{0,\boldsymbol{\kappa}}\right)\quad\forall\boldsymbol{\kappa}\in\mathbb{Z}^{K},\quad j\in\{1,\dots,M\}:\mathrm{i}\left\langle \boldsymbol{\kappa},\boldsymbol{\Omega}\right\rangle \approx\lambda_{j}.\label{eq:non_auto_normal}
\end{equation}
Alternatively, using a graph style parametrization, we obtain (cf.
eq.~\eqref{eq:reduced_dyn_graph})
\begin{equation}
\textrm{Graph style:}\quad\mathbf{s}_{0,\boldsymbol{\kappa}}=\mathbf{e}_{j}\left(\mathbf{u}_{j}^{\star}\mathbf{F}_{0,\boldsymbol{\kappa}}\right)\quad\forall\boldsymbol{\kappa}\in\mathbb{Z}^{K},\quad j\in\{1,\dots,M\}.\label{eq:non_auto_graph}
\end{equation}
Note, however, that these choices are only available for the modes
in the master subspace that are resonant with the external frequency
$\boldsymbol{\Omega}$, i.e., $j=1,\dots,M$ in the approximation~\eqref{eq:ext_res}.
If the near-resonance relation~\eqref{eq:ext_res} holds for any eigenvalues
outside the master subspace, i.e., $j=M+1,\dots,N$, then the domain
of convergence of our Taylor approximations is reduced. Depending
on the application, a workaround for this may be to include any nearly
resonant modes in the master subspace from the start. 

Finally, upon determining the reduced dynamics coefficients $\mathbf{s}_{0,\boldsymbol{\kappa}}$
specific to the chosen parametrization style (see eqs.~\eqref{eq:non_auto_normal},
\eqref{eq:non_auto_graph}), we compute a norm-minimizing solution
to~\eqref{eq:vec_inv_non_aut} given by
\begin{equation}
\mathbf{x}_{0,\boldsymbol{\kappa}}=\min_{\mathbf{y}\in\mathbb{C}^{N},\boldsymbol{\mathcal{L}}_{0,\boldsymbol{\kappa}}\mathbf{y}=\mathbf{h}_{0,\boldsymbol{\kappa}}(\mathbf{s}_{0,\boldsymbol{\kappa}})}\|\mathbf{y}\|^{2},\label{eq:lsqminnorm-1}
\end{equation}
as we did in the autonomous setting (cf. eq.~\eqref{eq:lsqminnorm}). 

\begin{rem}
	
	\label{rem:parallelization2} (Parallelization) For each $\boldsymbol{\kappa}\in\mathbb{Z}^{K}$, the reduced dynamics coefficients $\mathbf{s}_{0,\boldsymbol{\kappa}}$ and the manifold coefficients $\mathbf{x}_{0,\boldsymbol{\kappa}}$ can be determined independently of each other. Hence, parallel computation of these coefficients will result in high speed up due to minimal cross communication across the processes (see also Remark~\ref{rem:parallelization}). 
\end{rem}
\subsection*{Spectral submanifolds and forced response curves }

In structural dynamics, predicting the steady-state response of mechanical
systems in response to periodic forcing is often the end goal of the
analysis. This response is commonly expressed in terms of the FRC depicting
the response amplitude as a function of the external forcing frequency. FRCs are computationally expensive to obtain for large structural systems of engineering significance
(see, e.g., Ponsioen et al.~\cite{Ponsioen2020}, Jain et al.~\cite{Jain2019}).
The recent theory of spectral submanifolds~\cite{Haller2016} (SSM),
however, has enabled the fast extraction of such FRCs
via exact reduced-order models. The analytic results of Breunung \&
Haller~\cite{Breunung2018}, Ponsioen et al.~\cite{Ponsioen2019}
make it possible to obtain FRCs from the normal
form of the reduced dynamics on two-dimensional SSMs without any numerical
simulation. These approaches, however, develop SSM computations in
diagonal coordinates assuming semi-simplicity of the matrix $\mathbf{B}^{-1}\mathbf{A}$.
This has limited applicability for high-dimensional finite element-based
problems, as we have discussed in Section~\ref{subsec:Pitfalls-of-transformation}.
Here, we revisit their results in our context.

We consider the mechanical system~\eqref{eq:second-order} under periodic,
position-independent forcing as 
\begin{equation}
\mathbf{M}\ddot{\mathbf{x}}+\mathbf{C}\dot{\mathbf{x}}+\mathbf{K}\mathbf{x}+\mathbf{f}(\mathbf{x},\dot{\mathbf{x}})=\epsilon\mathbf{f}^{ext}(\Omega t),\label{eq:second-order-simple}
\end{equation}
where $\Omega\in\mathbb{R}_{+}$ is the external forcing frequency
and the periodic forcing $\mathbf{F}^{ext}(\Omega t)$ can be expressed
in a Fourier expansion as 
\begin{equation}
\mathbf{F}^{ext}(\Omega t)=\sum_{\kappa\in\mathbb{Z}}\mathbf{F}_{\kappa}^{ext}e^{\mathrm{i}\kappa\Omega t}.\label{eq:Fourier_periodic}
\end{equation}
We assume that the system~\eqref{eq:second-order-simple} represents
a lightly damped structure. This implies that $\Omega$ satisfies
the following near-resonance relationship\footnote{Note that exact resonance is not possible for damped eigenvalues, which exhibit non-zero (strictly negative) real parts.}  with a two-dimensional spectral
subspace associated with the eigenvalues $\lambda,\bar{\lambda}$:
\begin{equation}
\lambda-\mathrm{i}\eta\Omega\approx0,\quad\bar{\lambda}+\mathrm{i}\eta\Omega\approx0,\label{eq:near_ext_res}
\end{equation}
for some $\eta\in\mathbb{N}$. The left and right eigenvectors associated to the eigenvalues  $\{\lambda,\bar{\lambda}\}$ are $\{\mathbf{u},\bar{\mathbf{u}}\}$ and $\{\mathbf{v},\bar{\mathbf{v}}\}$. Furthermore, under light damping (i.e.,
$\frac{|\mathrm{Re}(\lambda)|}{|\mathrm{Im}(\lambda)|}\ll1$), the
near-resonance relationships
\begin{align}
\lambda & \approx\left(\ell+1\right)\lambda+\ell\bar{\lambda},\quad\bar{\lambda}\approx\ell\lambda+\left(\ell+1\right)\bar{\lambda},\label{eq:near_int_res}
\end{align}
will hold for any finite $\ell\in\mathbb{N}$ (see Szalai et al.~\cite{Szalai2017}).
As per eqs.~\eqref{eq:normal_form} and~\eqref{eq:non_auto_normal},
the near-resonances~\eqref{eq:near_ext_res}-\eqref{eq:near_int_res}
lead to the following normal form for the reduced dynamics (cf. Breunung
\& Haller~\cite{Breunung2018}):
\begin{equation}
\mathbf{R}_{\epsilon}(\mathbf{p},\Omega t)=\left[\begin{array}{c}
\lambda p\\
\bar{\lambda}\bar{p}
\end{array}\right]+\sum_{\ell\in\mathbb{N}}\left[\begin{array}{c}
\gamma_{\ell}p^{\ell+1}\bar{p}^{\ell}\\
\bar{\gamma}_{j}p^{\ell}\bar{p}^{\ell+1}
\end{array}\right]+\epsilon\left[\begin{array}{c}
\mathbf{u}^{\star}\mathbf{F}_{\eta}^{ext}e^{\mathrm{i}\eta\Omega t}\\
\bar{\mathbf{u}}^{\star}\bar{\mathbf{F}}_{\eta}^{ext}e^{-\mathrm{i}\eta\Omega t}
\end{array}\right]+\mathcal{O}(\epsilon|p|),\label{eq:normal_2D}
\end{equation}
{where the coefficients $\gamma_{\ell}$ are determined automatically from the normal-form style parametrization~\eqref{eq:normal_form} of the reduced dynamics on the two-dimensional SSM.}

Theorem 3.8 of Breunung \& Haller~\cite{Breunung2018} provides explicit
expressions for extracting FRCs from the reduced
dynamics on two-dimensional SSMs near a resonance with the forcing
frequency. Their expressions are derived under the assumption of proportional
damping; mono-harmonic, cosinusoidal and synchronous forcing on the
structure. The following statement generalizes their expressions to
system~\eqref{eq:second-order-simple} with periodic forcing~\eqref{eq:Fourier_periodic}
and provides us a tool to extract forced-response curves near
resonance from two-dimensional SSMs in physical coordinates. 
\begin{lem}
	\label{thm:frc}Under the near-resonance relationships~\eqref{eq:near_ext_res}
	and~\eqref{eq:near_int_res}\emph{:}
	
	(i) Reduced-order model on SSMs: The reduced dynamics~\eqref{eq:normal_2D} in polar coordinates
	$(\rho,\theta)$ is given by 
	\begin{align}
	\left[\begin{array}{c}
	\dot{\rho}\\
	\rho\dot{\psi}
	\end{array}\right] & =\mathbf{r}(\rho,\psi,\Omega):=\left[\begin{array}{c}
	a(\rho)\\
	b(\rho,\Omega)
	\end{array}\right]+\left[\begin{array}{cc}
	\cos\psi & \sin\psi\\
	-\sin\psi & \cos\psi
	\end{array}\right]\left[\begin{array}{c}
	\mathrm{Re}\left(f\right)\\
	\mathrm{Im}\left(f\right)
	\end{array}\right],\label{eq:reduced_dyn_polar}\\
	\dot{\phi} & =\Omega,
	\end{align}
	where 
	\begin{align*}
	a(\rho) & =\mathrm{Re}\left(\rho\lambda+\sum_{\ell\in\mathbb{N}}\gamma_{\ell}\rho^{2\ell+1}\right),\\
	b(\rho,\Omega) & =\mathrm{Im}\left(\rho\lambda+\sum_{\ell\in\mathbb{N}}\gamma_{\ell}\rho^{2\ell+1}\right)-\eta\rho\Omega,\\
	f & =\epsilon\mathbf{u}^{\star}\mathbf{F}_{\eta}^{ext},\\
	\psi & =\theta-\eta\phi.
	\end{align*}
	(ii) FRC: The fixed points of the system~\eqref{eq:reduced_dyn_polar}
	correspond to periodic orbits with frequency $\eta\Omega$ and are
	given by the zero level set of the scalar function 
	\begin{align}
	\mathcal{F}(\rho,\Omega) & :=\left[a(\rho)\right]^{2}+\left[b(\rho,\Omega)\right]^{2}-|f|^{2}.\label{eq:zero_level_F}
	\end{align}
	(iii) Phase shift: The constant phase shift $\psi$ between the external
	forcing $\mathbf{f}^{ext}(\Omega t)$ and a $\rho$-amplitude periodic
	response, obtained as a zero of eq.~\eqref{eq:zero_level_F}, is given
	by 
	\begin{equation}
	\psi=\arctan\left(\frac{b(\rho,\Omega)\mathrm{Re}(f)-a(\rho)\mathrm{Im}(f)}{-a(\rho)\mathrm{Re}(f)-b(\rho,\Omega)\mathrm{Im}(f)}\right).
	\end{equation}
	(iv) Stability: The stability of the periodic response is determined
	by the eigenvalues of the Jacobian 
	\begin{equation}
	J(\rho)=\left[\begin{array}{cc}
	\partial_{\rho}a & -b(\rho,\Omega)\\
	\frac{\partial_{\rho}b(\rho,\Omega)}{\rho} & \frac{a(\rho)}{\rho}
	\end{array}\right].\label{eq:Jacobian}
	\end{equation}
\end{lem}
\begin{proof}
	See Appendix~\ref{sec:Proof-of-Theorem}
\end{proof}
Note that the zero level set of $\mathcal{F}$, which provides the
FRC, can also be written as the zero-level set of the functions 
\[
\mathcal{G}^{\pm}(\rho,\Omega):=b(\rho,\Omega)\pm\sqrt{|f|^{2}-\left[a(\rho)\right]^{2}}.
\]
Despite the equivalence in the zero-level sets of the functions $\mathcal{F}$
and $\mathcal{G}^{\pm}$, one over the other might be preferred
to avoid numerical difficulties. The zero level set of $\mathcal{F}$
is a one-dimensional submanifold in the $(\rho,\Omega)$ space for
a given forcing of small enough amplitude $|f|$. The parameter values
for which the FRC contains more than one connected component is referred
in literature as the emergence of \emph{detached resonance curves}
or \emph{isolas}. The non-spurious zeros of the polynomial
$a(\rho)$ result in the non-trivial steady-state for the full system
(see Ponsioen et al.~\cite{Ponsioen2019}). The analytical formulas
given in Lemma~\ref{thm:frc} enable us to compute the FRCs along
with isolas, if those exist.

In the case of (near-) outer resonances of $\lambda$ with any of
the remaining eigenvalues of the system, such a two-dimensional SSM
does not exist (see Haller \& Ponsioen~\cite{Haller2016}) and one should include the resonant eigenvalues in the master modal subspace $E$, resulting in higher-dimensional
SSMs with inner resonances. The reduced dynamics on such high-dimensional SSMs can again be used to compute FRCs via numerical continuation, as discussed by Li et al.~\cite{Li2021}. 

The automated computation procedure
developed here is also applicable for treating high-dimensional problems
with inner resonances up to arbitrarily high order of accuracy.	A numerical implementation of the computational methodology developed
in this work, is available in the form of the open-source MATLAB package, SSMTool 2.0~\cite{SSMTool2.0}, which is integrated with a generic finite-element
solver (Jain et al.~\cite{FECode}) and \textsc{coco}~~\cite{Dankowicz}. This allows us to treat high-dimensional mechanics problems, as we demonstrate over several numerical examples
in the next section.

\section{Numerical examples \label{sec:Numerical-examples}}

In the following examples, we perform local SSM computations on mechanical systems 
systems following the methodology discussed in Sections~\ref{sec:computation_autonomous}
and~\ref{sec:computations_non-autonomous}, which involves the solution
of invariance equations~\eqref{eq:inv_eqn} and~\eqref{eq:inv_non-autonomous}. We use the reduced-dynamics on two-dimensional SSMs attached
to periodic orbits for obtaining FRCs of various nonlinear mechanical
systems via Lemma~\ref{thm:frc}.

The equations of motion governing the following examples are given
in the general form: 
\begin{equation}
\mathbf{M}\ddot{\mathbf{x}}+\mathbf{C}\dot{\mathbf{x}}+\mathbf{K}\mathbf{x}+\mathbf{f}(\mathbf{x})=\epsilon\mathbf{f}^{ext}(\Omega t),\qquad\mathbf{x}(t)\in\mathbb{R}^{n}.\label{eq:oscillator}
\end{equation}	 
{ An SSM characterizes the deformation in the corresponding modal subspace that arises due to the addition of nonlinearities in the linearized counterpart of system~\eqref{eq:oscillator}. Specifically, the nonlinear terms in the Taylor expansions $\mathbf{W}$, $\mathbf{R}$ (see eqs.~\eqref{eq:expW} and \eqref{eq:expR}) end up being nontrivial precisely due to the presence of the nonlinearity $\mathbf{f}$ in system~\eqref{eq:oscillator}. For each of the following examples, we illustrate this deformation of the modal subspace  by taking a snapshot (Poincar\'e section) of the non-autonomous SSM along with its reduced dynamics at an arbitrary time instant, $t=t_0$. We then plot the SSM as a graph over the modal coordinates $[\rho \cos \theta,~\rho \sin \theta]$, where $\theta = (\psi + \eta \Omega t_0)$ (see Lemma~\ref{thm:frc}). 
	
	To this end, we simply simulate the autonomous, two-dimensional ROM~\eqref{eq:reduced_dyn_polar} which results in the reduced dynamics trajectories $\rho(t)$ and $ \theta (t)$ on the SSM in polar coordinates. We then map these trajectories onto the SSM using the parametrization $\mathbf{W}_{\epsilon} (\mathbf{p}(t), \Omega t_0)) $, where 
	\begin{equation}
	\label{eq:redsimplot}
	\mathbf{p}(t) = \rho(t) \left[\begin{array}{c}
	e^{\mathrm{i}\theta(t)}\\
	e^{-\mathrm{i}\theta(t)}
	\end{array}\right]= \rho(t) \left[\begin{array}{c}
	e^{\mathrm{i}(\psi(t) + \eta \Omega t_0)}\\
	e^{-\mathrm{i}(\psi(t) + \eta \Omega t_0)}
	\end{array}\right].
	\end{equation}
}

We also compare these results with global computational  techniques involving numerical continuation of the periodic response via collocation, spectral and shooting-based approximations. { While the local manifold computations we have discussed would benefit greatly from parallel computing (see Remarks~\ref{rem:parallelization} and \ref{rem:parallelization2}), in this work, we refrain from any parallel computations for a fair comparison of computation time with other techniques, where the tools we have employed do not use parallelization. We perform all computations via openly available  MATLAB packages on version 2019b of MATLAB.}

\subsection{Finite-element-type oscillator chain}

As a first example, we consider the nonlinear oscillator chain
example used by Jain et al.~\cite{Jain2019}, whose computational implementation
can be made to resemble a finite-element assembly, with each of the
nonlinear springs treated as an element. 

\begin{figure}[H]
	\centering \includegraphics[width=0.7\linewidth]{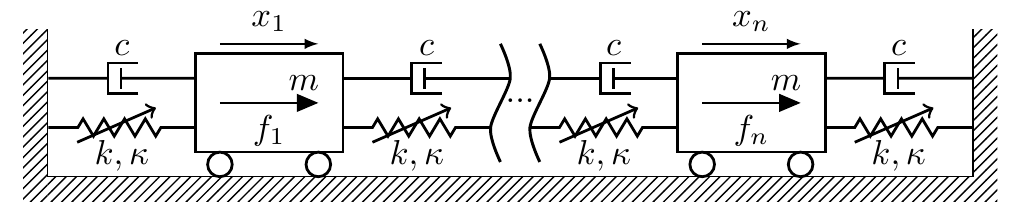}
	\caption{\label{fig:chain} \small The schematic of an $n-$degree-of-freedom,
		nonlinear oscillator chain where each spring has linear stiffness
		$k$ [N/m] and cubic stiffness $\kappa$, [N/m$^{3}$]; each damper has
		linear damping coefficient $c$ [N.s/m]; and each mass ($m$ [kg]) is
		forced periodically at frequency $\Omega$ [rad/s] (see eq.~\eqref{eq:oscillator})}
\end{figure}

The equations of motion
for the $n$-mass oscillator chain, shown in Figure~\ref{fig:chain},
are given by system~\eqref{eq:oscillator} with 
\begin{equation}
\mathbf{M}=m\mathbf{I}_{n},\quad\mathbf{K}=k\mathbf{L}_{n},\quad\mathbf{C}=c\mathbf{L}_{n},\quad\mathbf{f}(\mathbf{x})=\kappa\mathbf{f}_{3}\mathbf{x}^{\otimes3},
\end{equation}
where $\mathbf{L}_{n}$ is a Toeplitz matrix given as 
\begin{equation}
\mathbf{L}_{n}=\left[\begin{array}{ccccc}
2 & -1\\
-1 & 2 & -1\\
& \ddots & \ddots & \ddots\\
&  & -1 & 2 & -1\\
&  &  & -1 & 2
\end{array}\right]\in\mathbb{R}^{n\times n},
\end{equation}
and $\mathbf{f}_{3}\in\mathbb{R}^{n\times n^{3}}$ is a sparse cubic-coefficients
array such that 
\begin{equation}
\mathbf{f}_{3}\mathbf{x}^{\otimes3}=\left[\begin{array}{c}
x_{1}^{3}-(x_{2}-x_{1})^{3}\\
(x_{2}-x_{1})^{3}-(x_{3}-x_{2})^{3}\\
\vdots\\
(x_{n}-x_{n-1})^{3}-x_{n}^{3}
\end{array}\right].
\end{equation}
We choose the parameter values 
\begin{equation}
m=1,\quad k=1,\quad c=0.1,\quad\kappa=0.3,\quad\mathbf{f}^{ext}(\Omega t)=\mathbf{f}_{0}\cos(\Omega t),\quad\epsilon=0.1,
\end{equation}
where forcing frequency $\Omega$ in the range of 0.23-1 rad/s and
the forcing shape 
\begin{equation}
\mathbf{f}_{0}=[-0.386,-0.587,-0.521,-0.243,0.095,0.335,0.402,0.323, 0.188,0.075]^{\top}
\end{equation}
are chosen to excite the first three modes of the system. For the
chosen parameter values, the pairs of eigenvalues associated to the
first three modes are 
\begin{align}
\lambda_{1,2} & =-0.0041\pm0.2846\mathrm{i},\label{eq:lambda_12}\\
\lambda_{3.4} & =-0.0159\pm0.5632\mathrm{i},\label{eq:lambda_34}\\
\lambda_{5,6} & =-0.0345\pm0.8301\mathrm{i}.\label{eq:lambda_56}
\end{align}
For $\Omega\in[0.23,1]$, these three pairs of eigenvalues~\eqref{eq:lambda_12}-\eqref{eq:lambda_56}
are nearly resonant with $\Omega$ as per approximations~\eqref{eq:near_ext_res}
with $\eta=1$. We subdivide the frequency range into three intervals
around each of these near-resonant eigenvalue pairs. We then perform
SSM computations up to quintic order to approximate the near-resonant
FRC via Lemma~\ref{thm:frc} for each pair of near-resonant eigenvalues. 

{Figure~\ref{fig:chain_FRC}a illustrates the Poincar\'e section of the non-autonomous SSM computed around the second mode with eigenvalues~\eqref{eq:lambda_34} and near-resonant forcing frequency $\Omega = 0.6158$ rad/s (period $T = 2\pi/\Omega$). Each curve of the reduced dynamics shown in Figure 7a represents iterates of the period $T$-Poincar\'e map. In particular, any hyperbolic fixed points correspond to $T$-periodic orbits of the full system with the same hyperbolicity according to Lemma~\ref{thm:frc}. Hence, we directly obtain unstable and stable periodic orbits on the FRC by investigating the stable (blue) and unstable (red) fixed points of the reduced dynamics on the SSM for different values of $\Omega$, as shown in Figure~\ref{fig:chain_FRC}b.} 

Figure~\ref{fig:chain_FRC}b further shows that the FRC obtained from these SSM computations agrees with the spectral (harmonic balance) and collocation-based approximations. We perform these harmonic balance approximations using
an openly available MATLAB package, NLvib~\cite{Krack2019}, which
implements an alternating frequency-time (AFT) approach. We choose
$5$ harmonics for approximations in the frequency domain and $2^{7}$
time steps for the approximations in the time domain. For performing
collocation-based continuation, we use the $\texttt{po-}$toolbox
of $\textsc{coco}$~\cite{Dankowicz} with default settings and adaptive
refinement of collocation mesh and one-dimensional atlas algorithm.

{ The total computation time consumed in
	model generation, coefficient assembly and computation of all eigenvalues of this system was less than 1 second on a Windows-PC with Intel Core i7-4790 CPU @ 3.60GHz and 32 GB RAM.} We compare the computation times for obtaining the FRC using different
methods in Table~\ref{tab:comp_chain}.  
\begin{figure}[H]
	\centering \subfloat[]{\centering{}\includegraphics[width=0.49\linewidth]{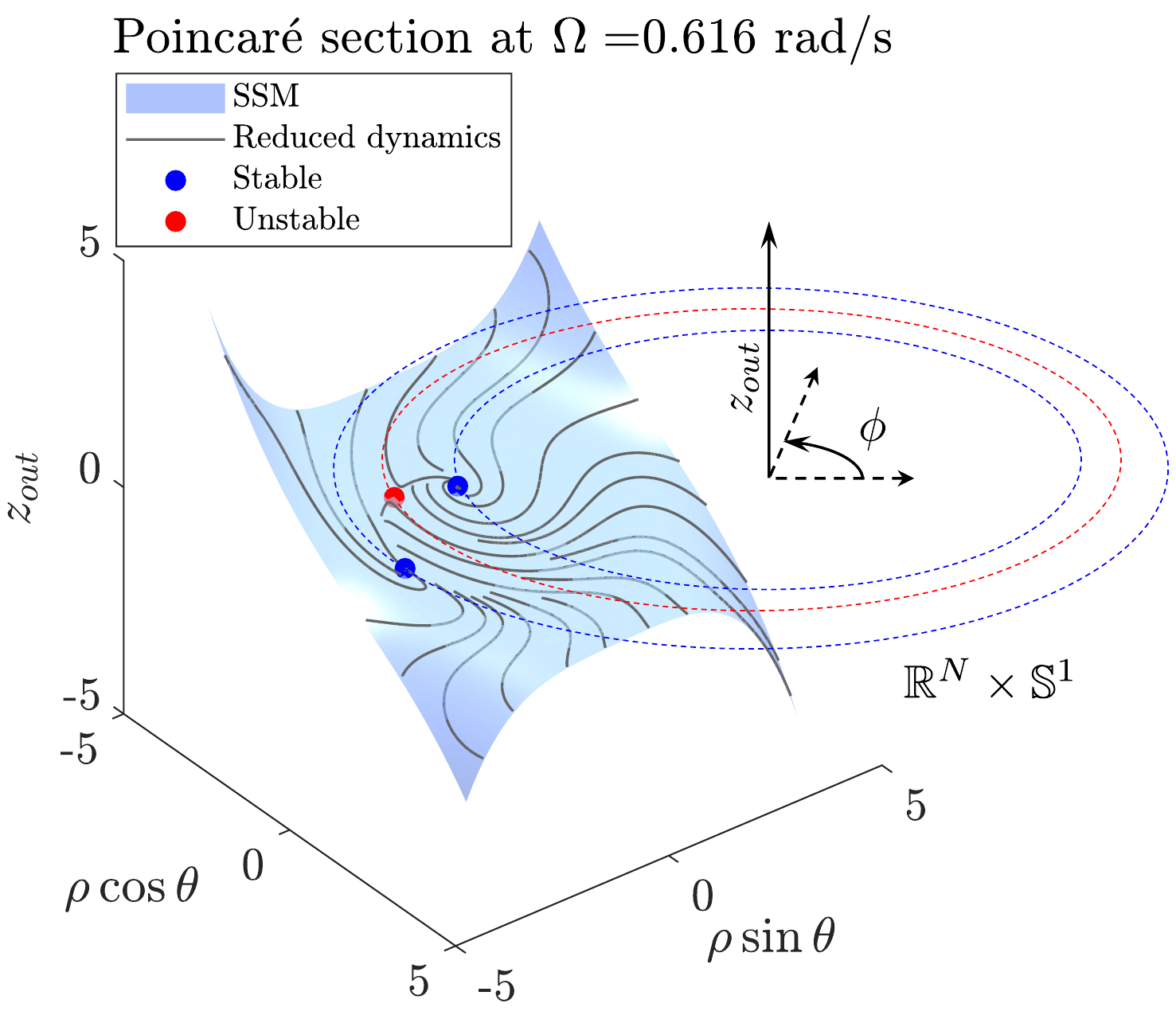}}\hfill{}\subfloat[]{\centering{}\includegraphics[width=0.49\linewidth]{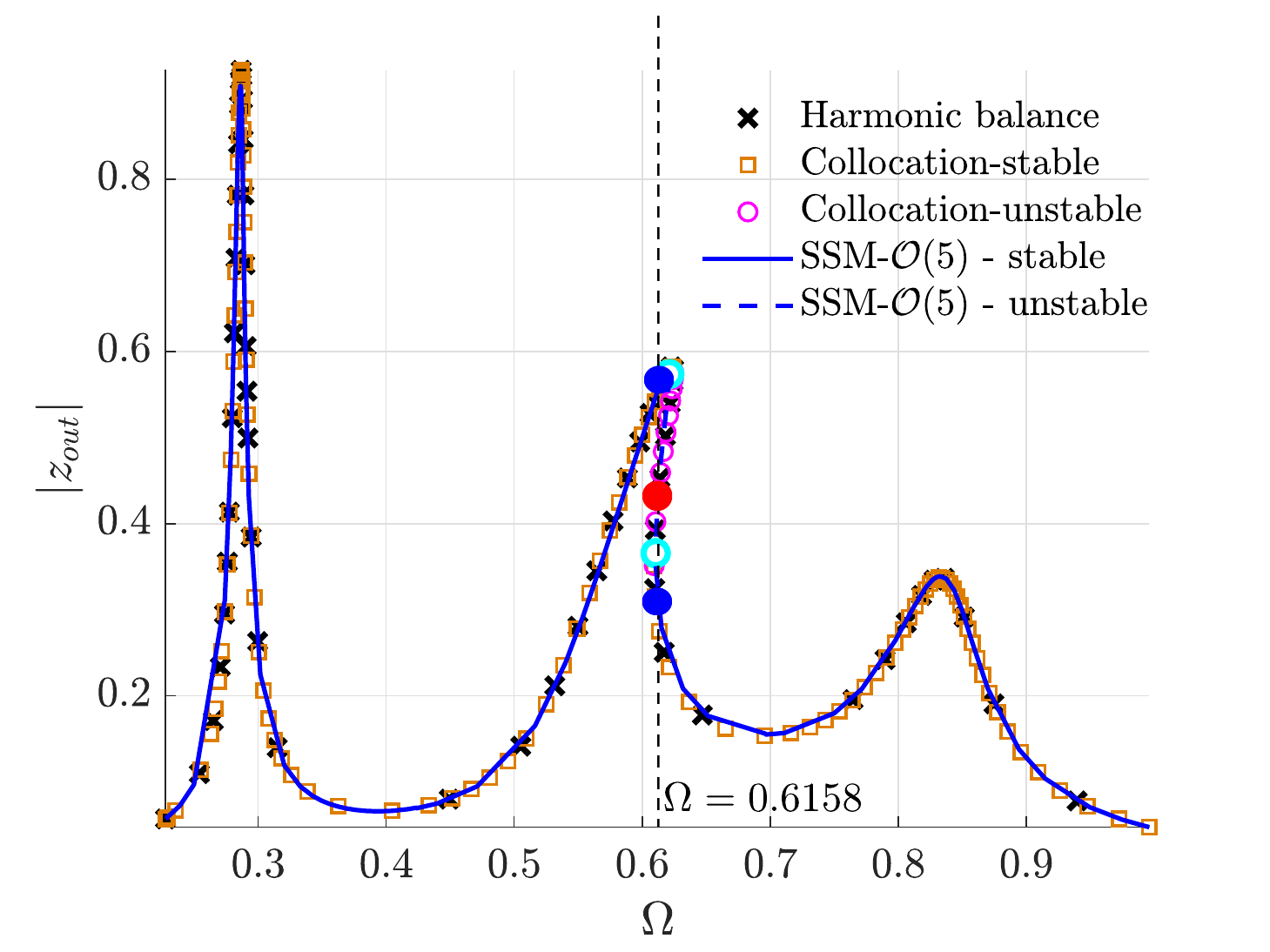}}
	\caption{\label{fig:chain_FRC}\small { (a) Poincar\'e section of the  non-autonomous SSM computed around the second mode (eigenvalues~\eqref{eq:lambda_34}) for $\Omega = 0.6158$~rad/s, where the reduced dynamics in polar coordinates $\rho$, $\theta$ is obtained by simulating the ROM~\eqref{eq:reduced_dyn_polar} (see eq.~\eqref{eq:redsimplot}). The fixed points in blue and red directly provide us the stable and unstable periodic orbits on the FRC. (b)} FRC obtained via local computations of SSM at $\mathcal{O}(5)$
		agree with those obtained using global continuation methods involving the
		harmonic balance method (NLvib~\cite{Krack2019}) and collocation
		(\textsc{coco}~\cite{Dankowicz}); the computation is performed
		for $n=10$ (see Table~\ref{tab:comp_chain} for computation times);
		the plot shows the displacement amplitude for the $5^{\mathrm{th}}$ (middle)
		degree of freedom.}
\end{figure}

In this example, the SSM-based analytic approximation to FRC using
Lemma~\ref{thm:frc} involves the computation of the $\mathcal{O}(5)$-autonomous
SSM three times (once around each resonant pair). The leading-order
non-autonomous SSM computation needs to be repeated for each $\Omega$
in the frequency span~$[0.23,1]$. We { emphasize} that while each of these
SSM computations is parallelizable (see Remark~\ref{rem:parallelization}) in contrast to
continuation-based global methods, we have reported computation times
via a sequential implementation in Table~\ref{tab:comp_chain}. As expected, we observe from Table~\ref{tab:comp_chain} that local approximations
to SSMs are a much faster means to compute FRCs in comparison to global
techniques that involve collocation or spectral (harmonic balance)
approximations. 

\begin{table}[H]
	\begin{centering}
		\begin{tabular}{|>{\centering}p{3cm}|>{\centering}p{3cm}|>{\centering}p{3cm}|>{\centering}p{3cm}|}
			\hline 
			\multirow{1}{*}{$n$} & \multicolumn{3}{c|}{Computation time [minutes:seconds]}\tabularnewline
			\hline 
			& SSM $\mathcal{O}(5)$  & Harmonic balance & Collocation \tabularnewline
			(number of degrees of freedom)  & (SSMTool 2.0~\cite{SSMTool2.0}) & (NLvib~\cite{Krack2019}) & 
			($\textsc{coco}$~\cite{Dankowicz}, atlas-1d) \tabularnewline
			\hline 
			\hline 
			10  & 00:07  & 00:14  & 02:47\tabularnewline
			\hline 
		\end{tabular}
		\par\end{centering}
	\caption{\label{tab:comp_chain} \small Computation times for obtaining the FRC depicted
		in Figure~\ref{fig:chain_FRC}. These computations were performed on
		MATLAB version 2019b, installed on a Windows-PC with Intel Core i7-4790
		CPU @ 3.60GHz and 32 GB RAM. }
\end{table}

\FloatBarrier
\subsection{Von Kármán Beam}

We now consider a finite element model of a geometrically nonlinear,
cantilevered von Kármán beam (Jain et al.~\cite{Jain2018}), illustrated
in Figure~\ref{fig:beam}a. The geometric and material properties
of the beam are given in Table~\ref{tab:pars_beam}. The equations
of motion are again given in the general form~\eqref{eq:oscillator}.
This model is programmed in the finite element solver~\cite{FECode},
which directly provides us the matrices $\mathbf{M},\mathbf{C},\mathbf{K}$
and the coefficients of the nonlinearity $\mathbf{f}$ in physical
coordinates. We discretize this model using 10 elements resulting
in $n=30$ degrees of freedom. 
\begin{figure}[H]
	\centering\includegraphics[width=0.48\linewidth]{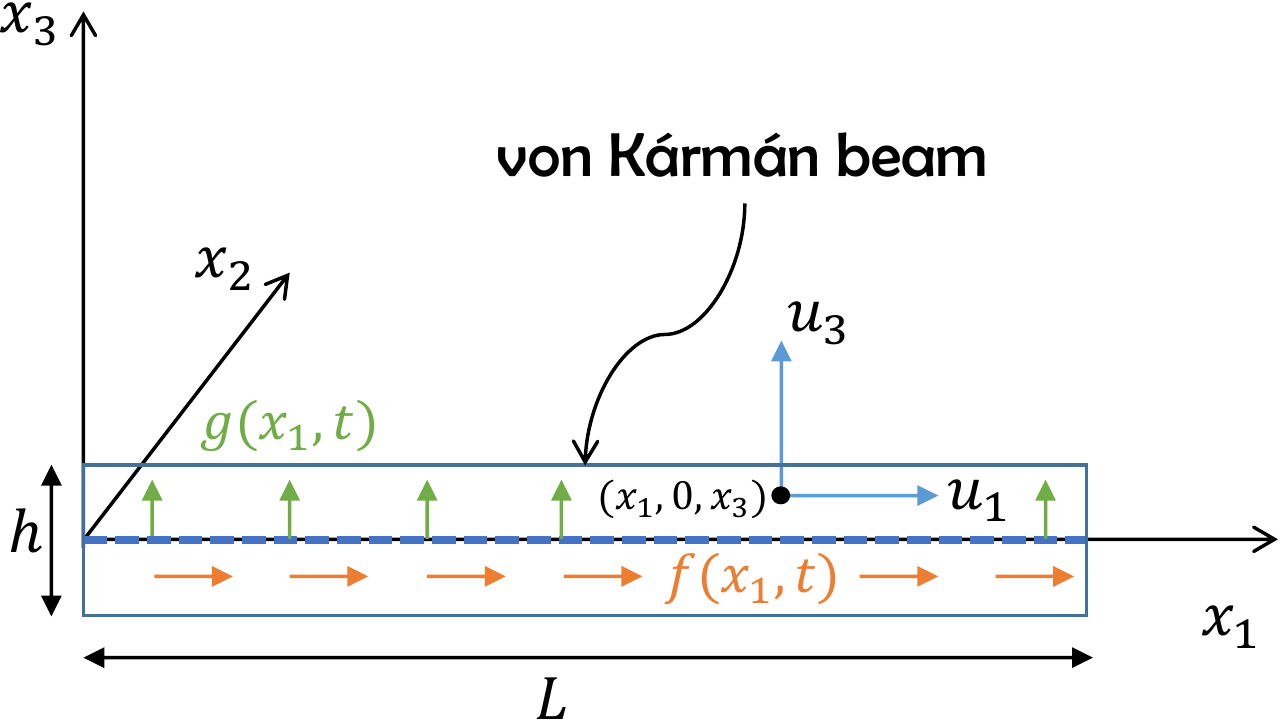}
	
	\caption{\label{fig:beam} The schematic of a two-dimensional von Kármán
		beam model (Jain et al.~\cite{Jain2018}) with height $h$ and length
		$L$, initially aligned with the $x_{1}$ axis, see Table~\ref{tab:pars_beam}
		for geometric and material properties.}
\end{figure}

\begin{table}[H]
	\begin{centering}
		\begin{tabular}{|c|l|c|}
			\hline 
			Symbol  & Meaning  & Value {[}unit{]}\tabularnewline
			\hline 
			\hline 
			$L$  & Length of beam  & 1 {[}m{]}\tabularnewline
			\hline 
			$h$  & Height of beam  & 1 {[}mm{]}\tabularnewline
			\hline 
			$b$  & Width of beam  & 0.1 {[}m{]}\tabularnewline
			\hline 
			$E$  & Young's Modulus  & 70 {[}GPa{]}\tabularnewline
			\hline 
			$\kappa$  & Viscous damping rate of material  & $10^{7}$ {[}Pa s{]}\tabularnewline
			\hline 
			$\rho$  & Density  & 2700 {[}kg/m$^{3}${]}\tabularnewline
			\hline 
		\end{tabular}
		\par\end{centering}
	\caption{\label{tab:pars_beam} \small Physical parameters of the von Kármán beam model
		(see Figure~\ref{fig:beam}a) }
\end{table}

The eigenvalue pair associated to the first mode of vibration of the
beam is given by 
\begin{equation}
\lambda_{1,2}=-0.0019+5.1681\mathrm{i},\label{eq:eig_beam}
\end{equation}
and external forcing is chosen as
\begin{equation}
\mathbf{f}^{ext}(\Omega t)=\mathbf{f}_{0}\cos(\Omega t),\quad\epsilon=10^{-3},\label{eq:beam_forcing}
\end{equation}
where $\mathbf{f}_{0}$ represents a spatially uniform forcing vector
with transverse forcing magnitude of $0.5$N/m across the length of
the beam. We choose the forcing frequency $\Omega$ in the range 4.1-6.2
rad/s for which the eigenvalue pair $\lambda_{1,2}$~\eqref{eq:eig_beam}
is nearly resonant with $\Omega$ (see~\eqref{eq:near_ext_res}).
We then perform $\mathcal{O}(5)$ SSM computations to approximate
the near-resonant FRC around the first natural frequency via Lemma
\ref{thm:frc}. 

{Figure~\ref{fig:vonkarmanbeam-ssm} illustrates the Poincar\'e section of the non-autonomous SSM computed around the first mode with eigenvalues~\eqref{eq:eig_beam} and near-resonant forcing frequency $\Omega = 5.4$ rad/s (period $T = 2\pi/\Omega$). We observe in~Figure~\ref{fig:vonkarmanbeam-ssm}a that the graph of the manifold is flat along the transverse degree of freedom, which gives the impression that there is no significant deformation of the modal subspace under the addition of nonlinearities in this system. At the same time, however, ~Figure~\ref{fig:vonkarmanbeam-ssm}b depicts a significant curvature of the SSM along the axial degree of freedom, which is related to the bending-stretching coupling introduced by the geometric nonlinearities in any beam model. Hence, we note that the invariance computation automatically accounts for the important physical effects arising due to nonlinearities in the form of the parametrizations $\mathbf{W}$ and $\mathbf{R}$ of the manifold and its reduced dynamics. These effects, otherwise, are typically captured by a heursitic projection of the governing equation onto carefully selected modes (see Jain et al.~\cite{Jain2018}, Buza et al.~\cite{Buza2020} for a discussion). }

\begin{figure}[H]
	\centering \subfloat[]{\centering{}\includegraphics[width=0.43\linewidth]{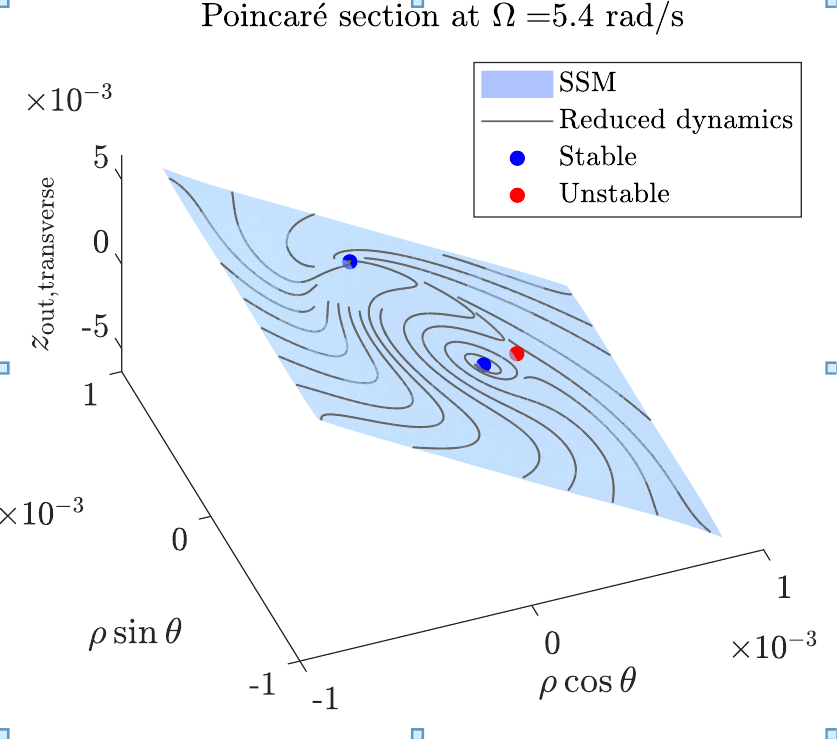}}\hfill{}\subfloat[]{\centering{}\includegraphics[width=0.53\linewidth]{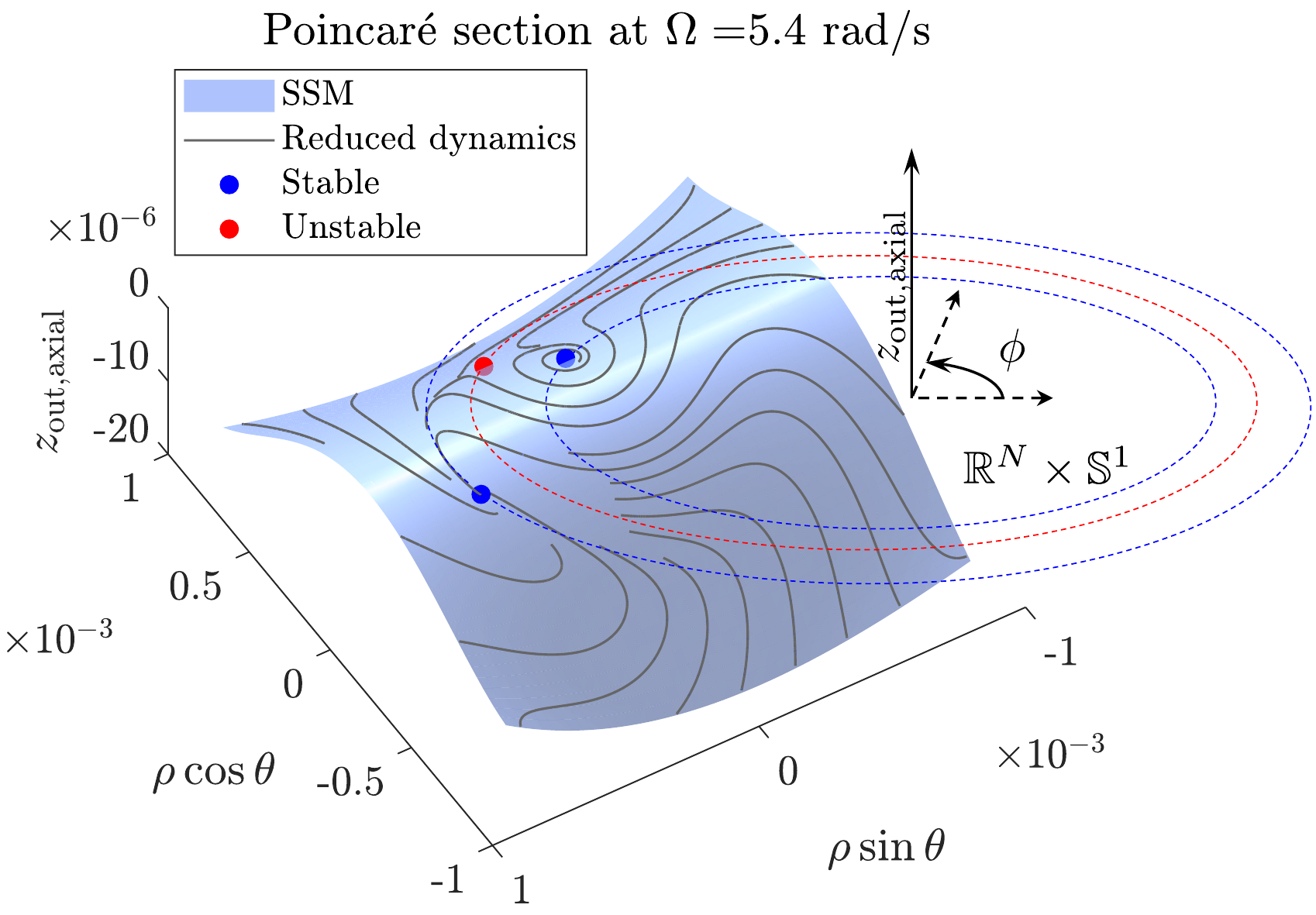}}
	
	\caption{\label{fig:vonkarmanbeam-ssm}\small  Poincar\'e sections of the  non-autonomous SSM computed around the first mode (eigenvalues~\eqref{eq:eig_beam}) of the beam  for near-resonant forcing frequency $\Omega = 5.4$ rad/s. The projection of the SSM onto the axial degree of freedom (b) located at the tip of the beam shows significant curvature in contrast to that onto the transverse degree of freedom (a), which appears relatively flat. The reduced dynamics in polar coordinates $\rho$, $\theta$ is obtained by simulating the ROM~\eqref{eq:reduced_dyn_polar} (see eq.~\eqref{eq:redsimplot}); the fixed points in blue and red directly provide us the stable and unstable periodic orbits on the FRC for different values of $\Omega$ (see Figure~\ref{fig:beam-FRC}).}
\end{figure}

{ Finally, in Figure~\ref{fig:beam-FRC}, we obtain unstable and stable periodic orbits on the FRC by investigating the stable (blue) and unstable (red) fixed points of the reduced dynamics on the SSM for different values of $\Omega$.	Figure~\ref{fig:beam-FRC} also shows that the FRC obtained via local SSM computation closely approximates the FRCs obtained using various global continuation techniques: collocation approximations via $\textsc{coco}$~\cite{Dankowicz}; and harmonic balance approximations via NLvib~\cite{Krack2019}. These continuation were performed with the same settings as in the previous example.}

\begin{figure}[H]
	\centering\includegraphics[width=0.48\linewidth]{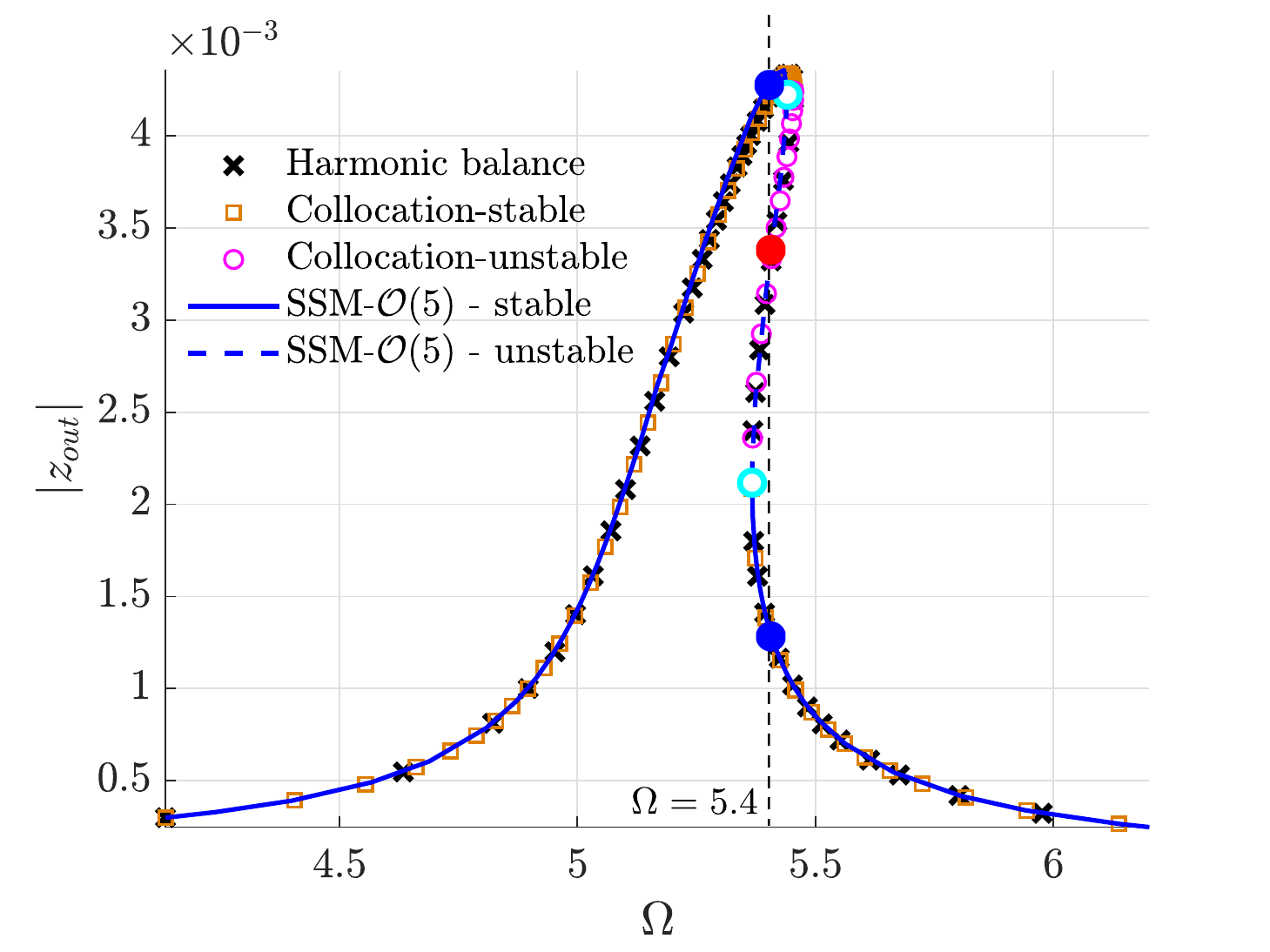}
	
	\caption{\small  \label{fig:beam-FRC} FRCs of the von K\'arm\'an beam model (see Figure~\ref{fig:beam}) with $n=30$ degrees of freedom under harmonic,
		spatially uniform transverse loading (see eq.~\eqref{eq:beam_forcing}). The FRC obtained via local
		computations of SSM at $\mathcal{O}(5)$ agrees with those obtained
		using global continuation methods involving the harmonic balance method
		(NLvib~\cite{Krack2019}) and collocation ($\textsc{coco}$~\cite{Dankowicz});
		the plot shows the displacement amplitude for in the $x_{3}$ direction
		at the tip of the beam (see Table~\ref{tab:comp_times_beam} for computation
		times).}
\end{figure}

{ Once again, the total computation time spent
	on model generation, coefficient assembly and computing the first 10 eigenvalues of this system was less than 1 second on a Windows-PC with Intel Core i7-4790 CPU @ 3.60GHz and 32 GB RAM.} Table~\ref{tab:comp_times_beam}
records the computation times to obtain FRCs via each of these methods. For the collocation-based response computation via \textsc{coco}~\cite{Dankowicz}, we also employ the atlas-$k$d algorithm~(see Dankowicz et al.~\cite{Dankowicz2020}) in addition to the default atlas-1d algorithm used in the previous example. Atlas-$k$d allows the user to choose the subspace of the continuation variables along which the continuation step size $ h $ is measured. Here, we choose this subset to be $ ({z}_{out}(0), \Omega, T) $, where $ T = \frac{2\pi}{\Omega} $ is the time period of periodic response and $ {z}_{out} $ is the response at output degree of freedom shown in Figure~\ref{fig:beam-FRC}. We allow for the continuation step size to adaptively vary between the values $ h_{min} = 10^{-5} $ to $ h_{max} = 50 $ and a maximum residual norm for the predictor step to be 10. We found these settings to be optimal for this atlas-$k$d run since relaxing these tolerances further has no effect on the continuation speed. Once again, the computation times in Table~\ref{tab:comp_times_beam} indicate orders-of-magnitude higher speed in reliably approximating
FRC via local SSMs computations in comparison to global techniques
that involve collocation or spectral approximations.

\begin{table}[H]
	\begin{centering}
		\begin{tabular}{|>{\centering}p{3cm}|>{\centering}p{3cm}|>{\centering}p{3cm}|>{\centering}p{2cm}>{\centering}p{2cm}|}
			\hline 
			$n$ & \multicolumn{4}{c|}{Computation time [hours:minutes:seconds]}\tabularnewline
			\hline 
			& SSM $\mathcal{O}(5)$ & Harmonic balance & \multicolumn{2}{c|}{Collocation } \tabularnewline
			(number of degrees  & (SSMTool 2.0~\cite{SSMTool2.0}) & (NLvib~\cite{Krack2019}) & 
			\multicolumn{2}{c|}{($\textsc{coco}$~\cite{Dankowicz})} \tabularnewline
			of freedom)& & & atlas-1d & atlas-$k$d \tabularnewline
			\hline 
			\hline 
			30  & 00:00:03  & 00:31:15  & 05:36:15 & 05:09:18\tabularnewline
			\hline 
		\end{tabular}
		\par\end{centering}
	\caption{\label{tab:comp_times_beam} \small Computation time for obtaining the FRCs
		depicted in Figure~\ref{fig:beam}b. These computations were performed
		on MATLAB version 2019b installed on a Windows-PC with Intel Core
		i7-4790 CPU @ 3.60GHz and 32 GB RAM. }
\end{table}

\FloatBarrier
\subsection{Shallow-arch structure }

Next, we consider a finite element model of a geometrically nonlinear
shallow arch structure, illustrated in Figure~\ref{fig:plate}a (Jain
\& Tiso~\cite{Jain2018-1}). 

\begin{figure}[H]
	\subfloat[]{\centering{}\includegraphics[width=0.38\linewidth]{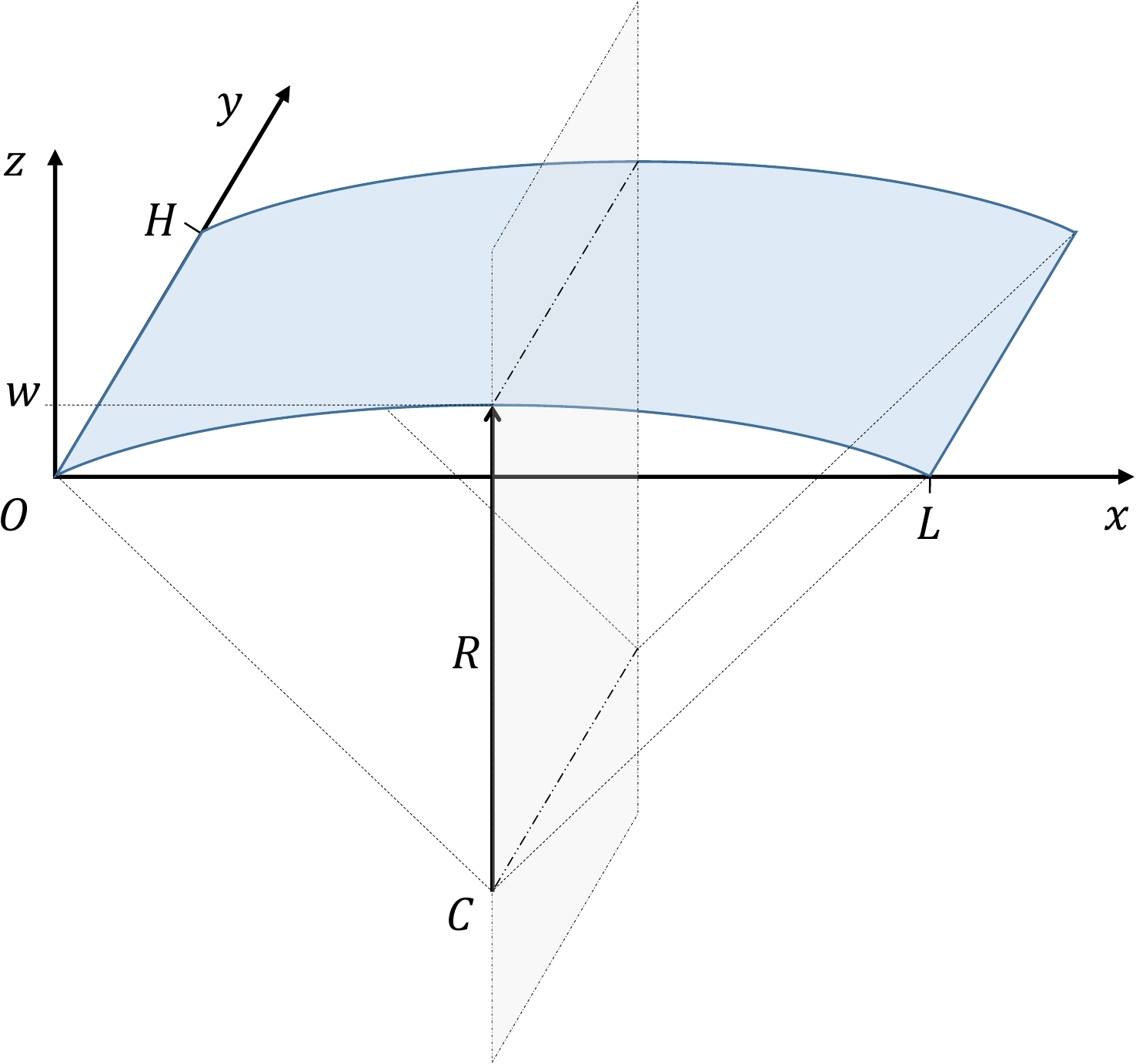}}\hfill{}\subfloat[]{\centering{}\includegraphics[width=0.48\linewidth]{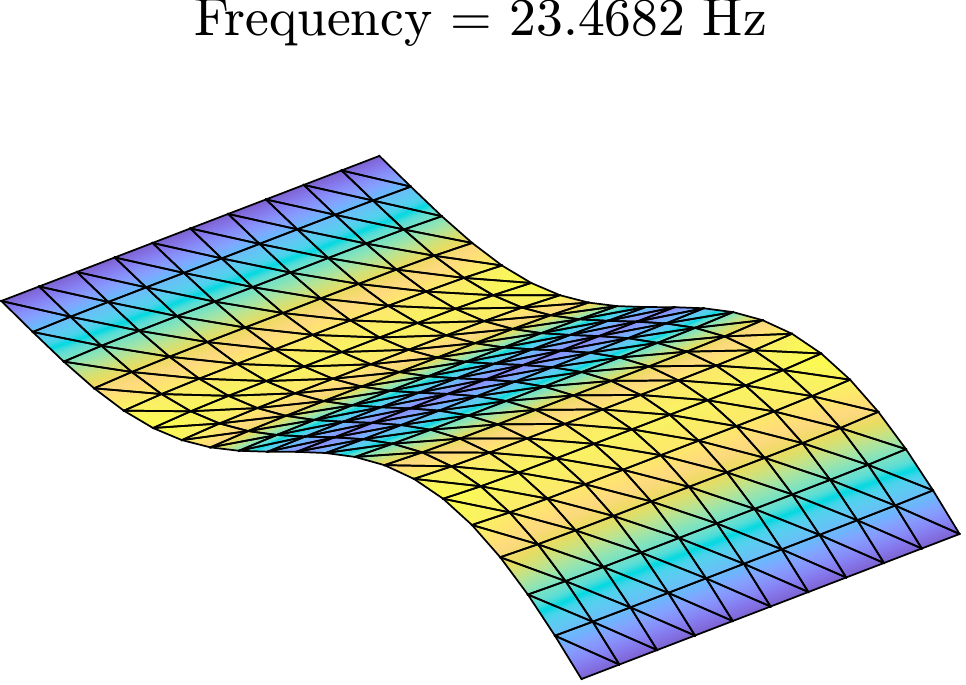}~}
	
	\caption{\label{fig:plate} \small (a) schematic of a shallow-arch structure (Jain
		\& Tiso~\cite{Jain2018-1}), see Table~\ref{tab:pars_plate} for geometric
		and material properties. This plate is simply supported at the two
		opposite edges aligned along the $y$-axis. (b) The finite element
		mesh (containing 400 elements, 1,320 degrees of freedom) deformed along first bending
		mode having undamped natural frequency of approximately 23.47 Hz.}
\end{figure}
The geometrical and material properties
of this curved plate are given in Table~\ref{tab:pars_plate}. The
plate is simply supported at the two opposite edges aligned along
the $y$-axis in Figure~\ref{fig:plate}a. The model is discretized
using flat, triangular shell elements and contains 400 elements, resulting
in $n=1320$ degrees of freedom. The open-source finite element code~\cite{FECode}
directly provides us the matrices $\mathbf{M},\mathbf{C},\mathbf{K}$
and the coefficients of the nonlinearity $\mathbf{f}$ in the equations
of motion~\eqref{eq:oscillator}.

\begin{table}[H]
	\begin{centering}
		\begin{tabular}{|c|l|c|}
			\hline 
			Symbol  & Meaning  & Value {[}unit{]}\tabularnewline
			\hline 
			\hline 
			$L$  & Length of plate  & 2 {[}m{]}\tabularnewline
			\hline 
			$t$  & thickness of plate  & 10 {[}mm{]}\tabularnewline
			\hline 
			$H$  & Width of beam  & 1 {[}m{]}\tabularnewline
			\hline 
			$E$  & Young's Modulus  & 70 {[}GPa{]}\tabularnewline
			\hline 
			$\nu$  & Poisson's ratio  & 0.33 {[}-{]}\tabularnewline
			\hline 
			$\kappa$  & Viscous damping rate of material  & $10^{5}$ {[}Pa s{]}\tabularnewline
			\hline 
			$\rho$  & Density  & 2700 {[}kg/m$^{3}${]}\tabularnewline
			\hline 
		\end{tabular}
		\par\end{centering}
	\caption{\label{tab:pars_plate} \small Geometrical and material parameters of the
		shallow-arch structure in Figure~\ref{fig:plate}a}
\end{table}

The first mode of vibration of this structure is shown in Figure~\ref{fig:plate}b
and the corresponding eigenvalue pair is given by 
\begin{equation}
\lambda_{1,2}=-0.29\pm147.45\mathrm{i}.\label{eq:eig_plate}
\end{equation}

The external forcing is again given by
\begin{equation}
\mathbf{f}^{ext}(\Omega t)=\mathbf{f}_{0}\cos(\Omega t),\quad\epsilon=0.1,\label{eq:plate_forcing}
\end{equation}
where $\mathbf{f}_{0}$ represents a vector of concentrated load in
$z$-direction with magnitude of 100 N at the mesh node located at
$x=\frac{L}{2},y=\frac{H}{2}$ in Figure~\eqref{fig:plate}a. We choose
the forcing frequency $\Omega$ in the range 133-162 rad/s for which
the eigenvalue pair $\lambda_{1,2}$ is nearly resonant with $\Omega$
(see~\eqref{eq:near_ext_res}).

\begin{figure}[H]
	\subfloat[]{\centering{}\includegraphics[width=0.41\linewidth]{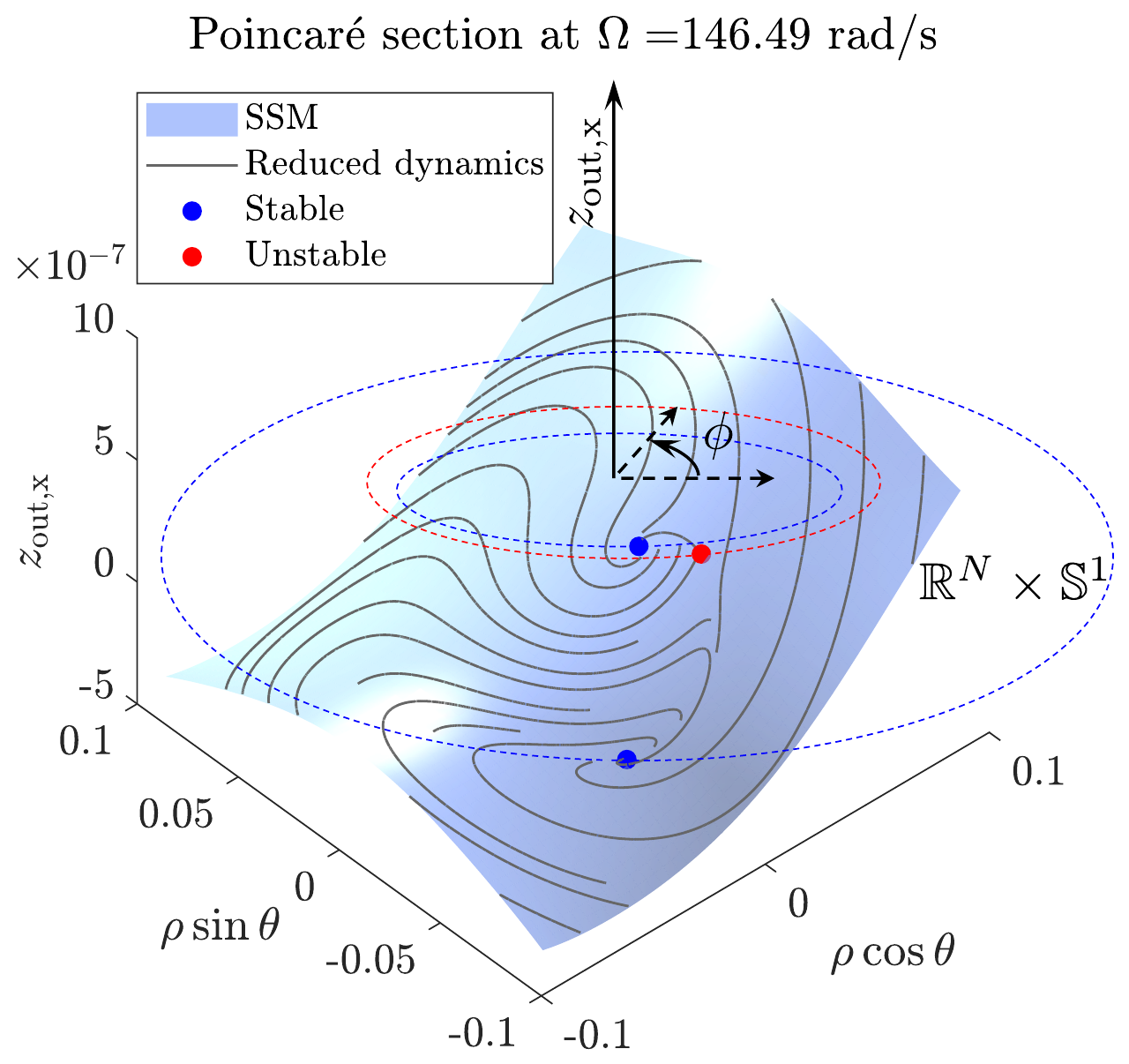}}\hfill{}\subfloat[]{\centering{}\includegraphics[width=0.57\linewidth]{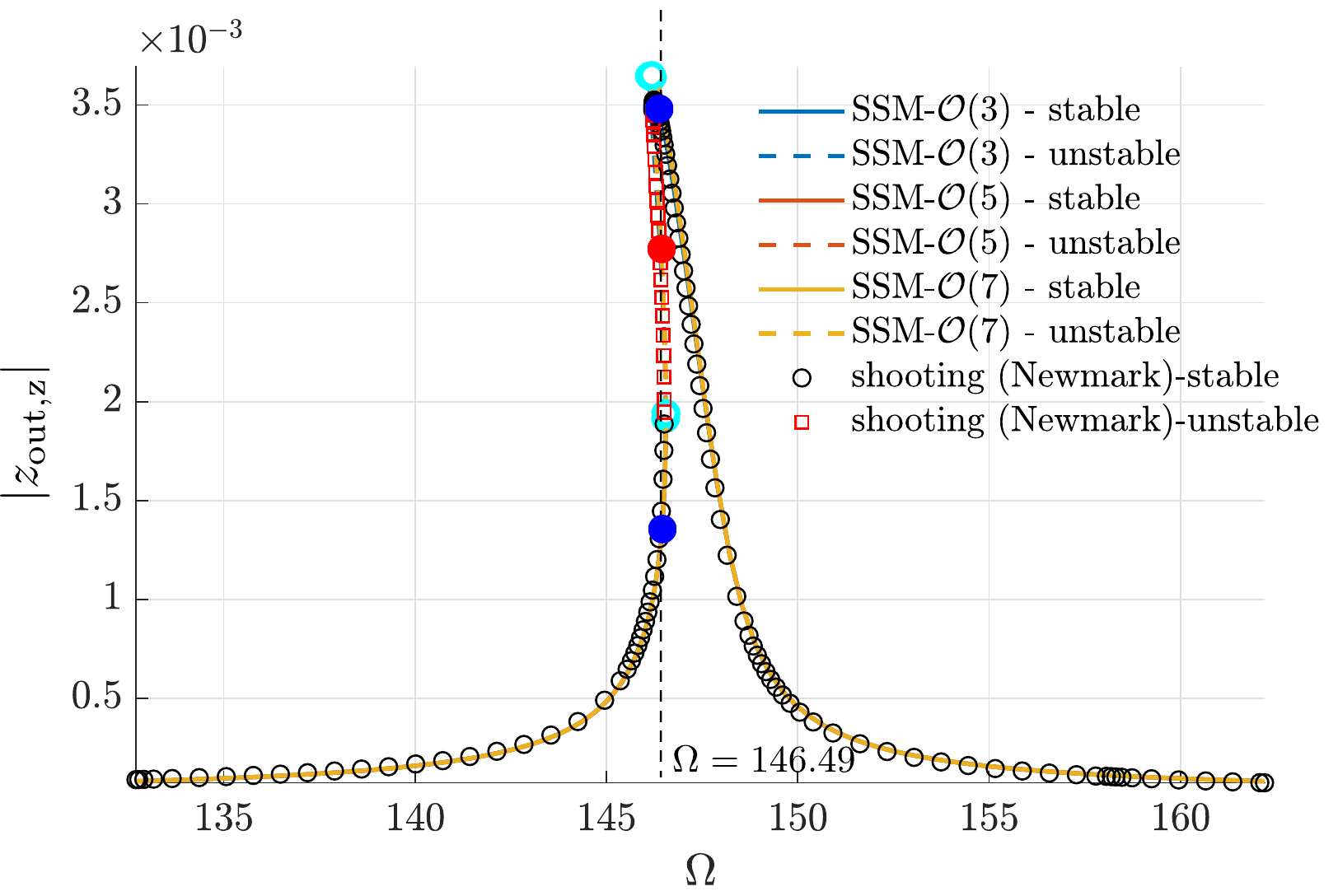}~}				
	\caption{ \label{fig:plate_FRC} \small (a) Poincar\'e section of the  non-autonomous SSM  of the shallow-arch structure (see Figure~\ref{fig:plate}) computed around the first mode (eigenvalues~\eqref{eq:eig_plate})  for near-resonant forcing frequency $\Omega = 146.49$ rad/s. The reduced dynamics in polar coordinates $\rho$, $\theta$ is obtained by simulating the ROM~\eqref{eq:reduced_dyn_polar} (see eq.~\eqref{eq:redsimplot}); the fixed points in blue and red directly provide us the stable and unstable periodic orbits on the FRC~(b) for different values of $\Omega$. FRCs obtained using local SSM computations at $\mathcal{O}(3),\mathcal{O}(5)$
		and $\mathcal{O}(7)$ agree with that obtained via global continuation
		based on the shooting method, which implements the Newmark time integration (see Table~\ref{tab:comp_times_plate} for computation times);
		plots (a) and (b) show the displacements in the $x$ and $z$-directions at the
		mesh node located at $x=\frac{L}{2},y=\frac{H}{2}$ in Figure~\ref{fig:plate}.}
\end{figure}

We then compute the near-resonant
FRC around the first natural frequency via $\mathcal{O}(3),\mathcal{O}(5),$
and $\mathcal{O}(7)$ SSM computations using Lemma~\ref{thm:frc}. { Once again, Figure~\ref{fig:plate_FRC}a shows the Poincar\'e section of the non-autonomous SSM for the near-resonant forcing frequency $\Omega = 146.49$ rad/s, where we directly obtain the unstable (red) and stable (blue) periodic orbits on the FRC as hyperbolic fixed points of the reduced dynamics~(89) on the SSM.} 
The three FRCs at $\mathcal{O}(3),\mathcal{O}(5),$ and $\mathcal{O}(7)$
seem to converge to softening response shown in Figure~\ref{fig:plate_FRC}b.
Note that we expect a softening behavior in the FRC of shallow-arches
(see, e.g., Buza et al.~\cite{Buza2020,Buza2021}). 

Due to excessive memory requirements, this FRC could not
be computed using collocation approximations via $\textsc{coco}$~\cite{Dankowicz}
or using harmonic balance approximations via NLvib~\cite{Krack2019}.
Instead, we compare this FRC to another global continuation technique
based on the shooting method, which is still feasible (see Introduction).

For shooting, we use the classic Newmark time integration scheme (Newmark~\cite{Newmark}, see Géradin \& Rixen~\cite{Geradin} for a review)
as the common Runge-Kutta schemes (e.g., $\texttt{ode45}$ of MATLAB)
struggle to converge in structural dynamics problems. We use an open-source toolbox~\cite{coco-shoot}, based on the atlas-1d algorithm of $\textsc{coco}$~\cite{Dankowicz} for continuation of the periodic solution trajectory obtained via shooting (see Dancowicz et al.~\cite{Dankowicz2020}). We use a constant
time step throughout time integration which is chosen by dividing
the time span $T=\frac{2\pi}{\Omega}$ into 100 equal intervals. We
found this choice of time step to be nearly optimal for this problem
as larger time steps lead to non-quadratic convergence during Newton-Raphson
iterations and smaller time steps result in slower computations. The
stability of the response is computed by integrating the equations
of variation around the converged periodic orbit. 

\noindent 
\begin{table}[H]
	\begin{centering}
		\begin{tabular}{|>{\centering}p{3cm}|>{\centering}p{3cm}|ccc|}
			\hline 
			\multirow{1}{*}{$n$ } & \multicolumn{4}{c|}{Computation time {[}hours:minutes:seconds{]}}\tabularnewline
			\hline 
			(number of degrees & Shooting method  & SSM $\mathcal{O}(3)$  & SSM $\mathcal{O}(5)$ & SSM 
			$\mathcal{O}(7)$\tabularnewline
			of freedom) & \multirow{1}{*}{(Newmark)} & \multicolumn{3}{c|}{(SSMTool 2.0~\cite{SSMTool2.0})}\tabularnewline
			
			\hline 
			\hline 
			1,320  & 52:50:14 & 00:00:07  & 00:00:12 & 00:00:28\tabularnewline
			\hline 
		\end{tabular}
		\par\end{centering}
	\caption{\label{tab:comp_times_plate} \small Computation time for obtaining the FRCs
		depicted in Figure~\ref{fig:plate_FRC}. All computations were performed
		on MATLAB version 2019b installed on a Windows-PC with Intel Core
		i7-4790 CPU @ 3.60GHz and 32 GB RAM. }
\end{table}
{ The total time consumed
	in model generation and coefficient assembly was 33 seconds on a Windows-PC with Intel Core	i7-4790 CPU @ 3.60GHz and 32 GB RAM. This includes the time spent in computing the first 10 eigenvalues of this system, which took less than 1 second.} Figure~\ref{fig:plate_FRC}b shows that this shooting based global continuation agrees with the
SSM-based approximation to the FRC. Obtaining this FRC via the shooting
methods, however, takes more than 2 days, in contrast to SSM-based
approximation using the proposed computational methodology, which
still takes less than a minute even at $\mathcal{O}(7)$, as shown
in Table~\ref{tab:comp_times_plate}.

\subsection{Aircraft Wing}
\label{sec:wing}
As a final example, we consider the finite element model of a geometrically
nonlinear aircraft wing originally presented by Jain et al~\cite{Jain2017}
(see Figure~\ref{fig:wing}). The wing is cantilevered at one of its
ends and the structure is meshed using flat triangular shell elements
featuring 6 degrees of freedom per node. With 49,968 elements and
133,920 degrees of freedom, this model provides a physically
relevant as well as computationally realistic problem that is beyond
feasibility for global continuation techniques based on collocation,
spectral and shooting methods, as shown by previous examples. The open-source
finite element code~\cite{FECode} directly provides us the matrices
$\mathbf{M},\mathbf{K}$ and the coefficients of the nonlinearity
$\mathbf{f}$ in the equations of motion~\eqref{eq:oscillator}. 
\begin{table}[H]
	\begin{centering}
		\begin{tabular}{|c|l|c|}
			\hline 
			Symbol & Meaning & Value {[}unit{]}\tabularnewline
			\hline 
			\hline 
			$L$ & Length of wing ($z$ direction) & 5 {[}m{]}\tabularnewline
			\hline 
			$H$ & Height of wing ($y$ direction) & 0.1 {[}m{]}\tabularnewline
			\hline 
			$W$ & Width of wing ($x$ direction) & 0.9 {[}m{]}\tabularnewline
			\hline 
			$t$ & thickness of elements & 1.5 {[}mm{]}\tabularnewline
			\hline 
			$E$ & Young's Modulus & 70 {[}GPa{]}\tabularnewline
			\hline 
			$\nu$ & Poisson's ratio & 0.33 {[}-{]}\tabularnewline
			\hline 
			$\rho$ & Density & 2700 {[}kg/m$^{3}${]}\tabularnewline
			\hline 
		\end{tabular}
		\par\end{centering}
	\caption{\label{tab:pars_wing}\small Geometrical and material parameters of the shallow-arch
		structure in Figure~\ref{fig:plate}a}
\end{table}

\begin{figure}[H]
	\subfloat[]{\centering{}\includegraphics[width=0.48\linewidth]{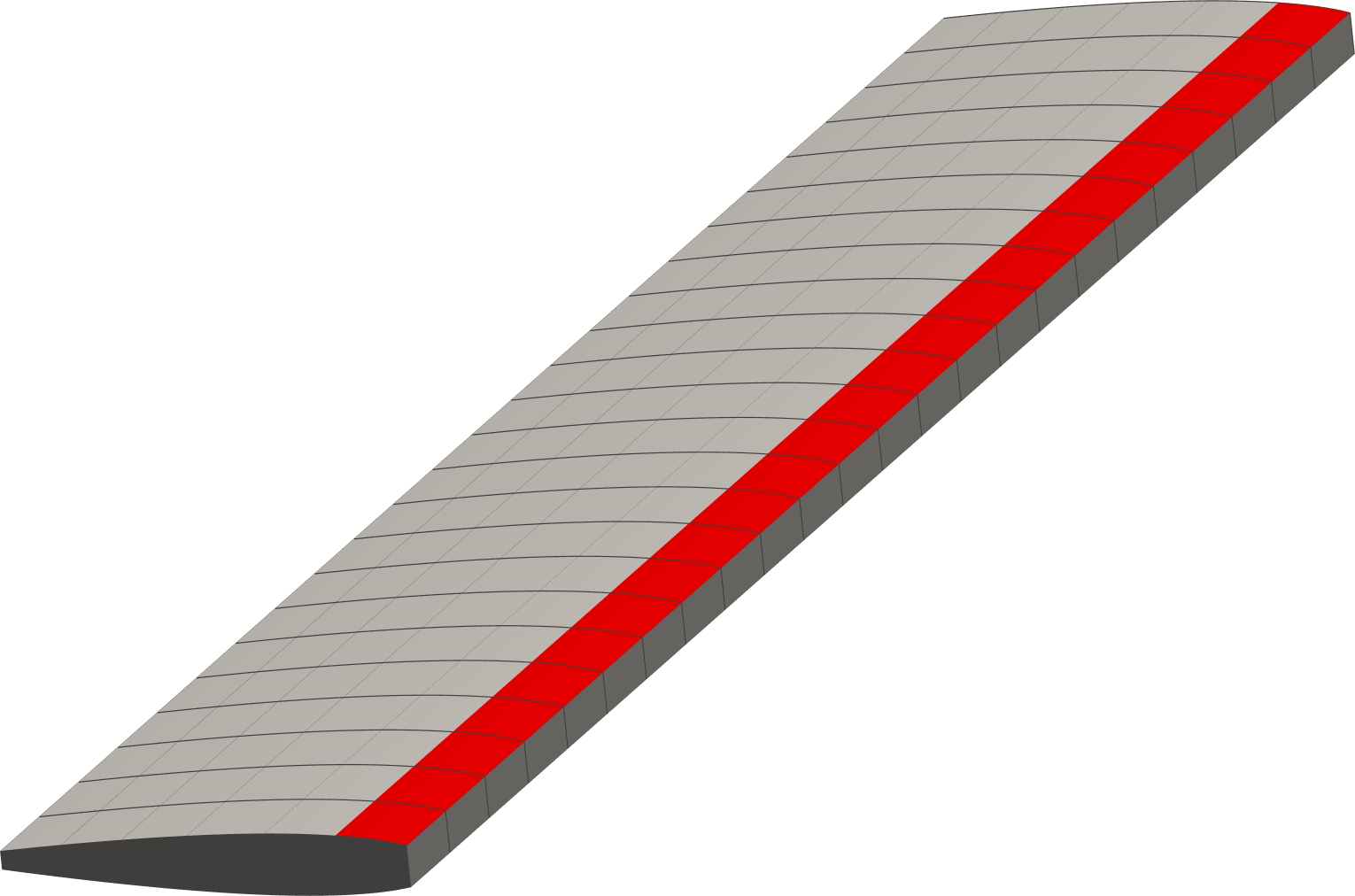}}\hfill{}\subfloat[]{\centering{}\includegraphics[width=0.48\linewidth]{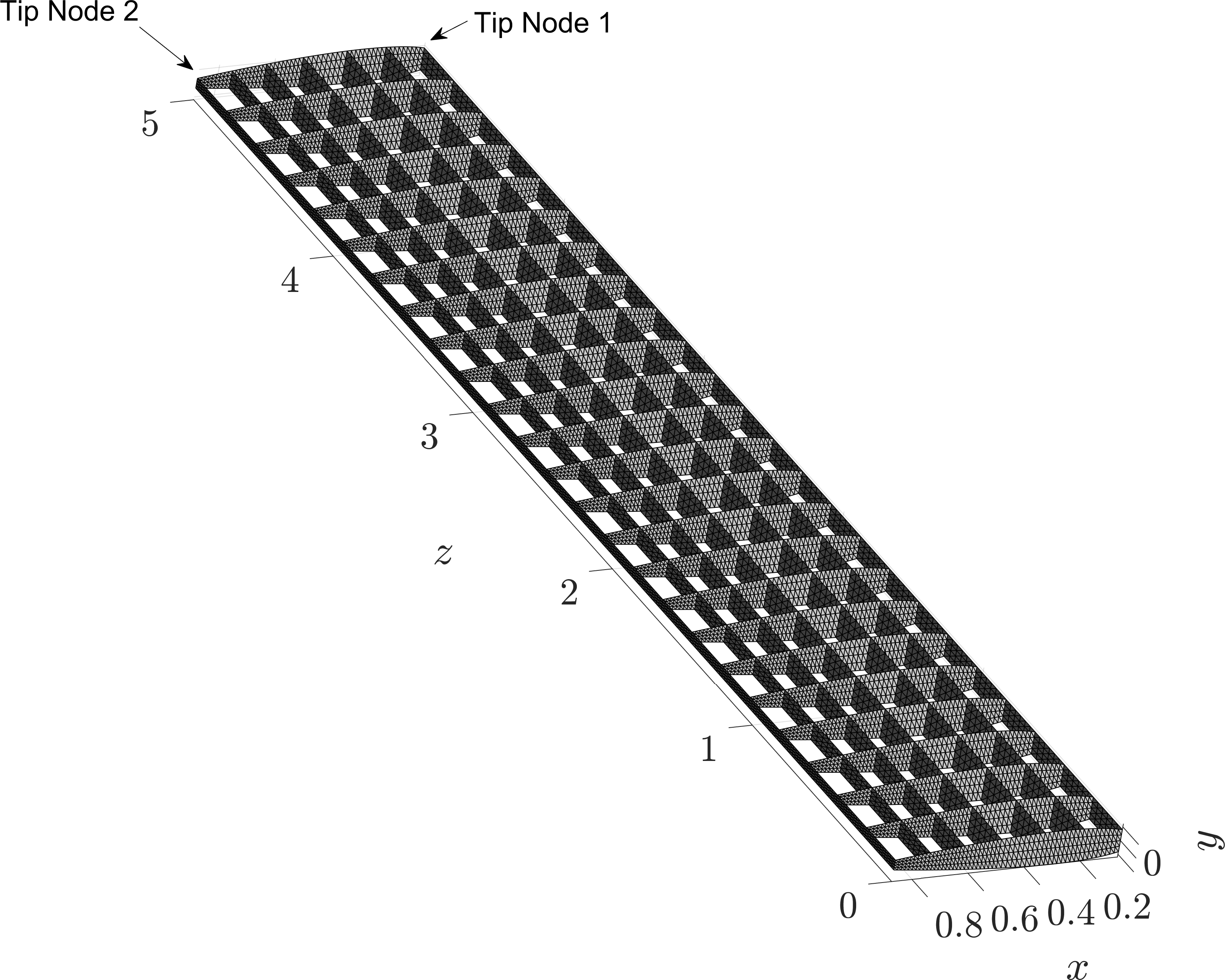}~}
	
	\caption{\label{fig:wing} \small (a) A wing structure with NACA 0012 airfoil stiffened
		with ribs (Jain et al.~\cite{Jain2017}), see Table~\ref{tab:pars_plate}
		for geometric and material properties. (b) The finite-element mesh
		is illustrated after removing the skin panels. The wing is cantilevered
		at the $z=0$ plane. The mesh contains 49,968 elements which results
		in $n=133,920$ degrees of freedom.}
\end{figure}

For assembling coefficients on a problem of this size, we used the
Euler supercomputering cluster at ETH Zurich. The total time consumed
in model generation and coefficient assembly was 1 hour 21 minutes
and 38 seconds \emph{without} any parallelization. { This time includes the time taken for computing the first 10 eigenvalues of this system, which was approximately 5 seconds.} The main bottleneck
was the memory consumption during the assembly of the coefficients
of the nonlinearity $\mathbf{f}$, where the peak memory consumption
was around $183$ GB. However, once assembled, these coefficients
consume only about 1.8~GB of RAM. This extraordinary memory consumption
during assembly occurs due to a sub-optimal assembly procedure of
sparse tensors~\cite{Bader2021}. To avoid these bottlenecks, parallel
computing and distributed memory architectures need to be employed,
which are currently not available in the packages we have used.  

In this example, we choose Rayleigh damping (see, e.g., Géradin \&
Rixen~\cite{Geradin}), which is commonly employed in structural dynamics
applications to construct the damping matrix $\mathbf{C}=\alpha\mathbf{M}+\beta\mathbf{K}$
as a linear combination of mass and stiffness matrices. The constants
$\alpha,\beta$ are chosen to ensure a damping ratio of 0.4\% along
the first two vibration modes. The eigenvalue pair associated to the
first mode of vibration is given by
\begin{equation}
\lambda_{1,2}=-0.0587\pm29.3428\mathrm{i}.\label{eq:eig_wing}
\end{equation}
Once again, we choose harmonic external forcing given by
\begin{equation}
\mathbf{f}^{ext}(\Omega t)=\mathbf{f}_{0}\cos(\Omega t),\quad\epsilon=0.01,\label{eq:plate_forcing-1}
\end{equation}
where $\mathbf{f}_{0}$ represents a vector of concentrated loads
at the tip nodes 1 and 2 (see Figure~\eqref{fig:wing}b) in the transverse
$y$-direction each with a magnitude of 100 N. We choose the forcing
frequency $\Omega$ in the range 26.4-32.3 rad/s for which the eigenvalue
pair $\lambda_{1,2}$ is nearly resonant with $\Omega$ (see~\eqref{eq:near_ext_res}).
We then compute the near-resonant FRCs around the first natural frequency
via $\mathcal{O}(3),\mathcal{O}(5),$ and $\mathcal{O}(7)$ SSM computations
using Lemma~\ref{thm:frc}. 

{ Similarly to the previous examples, Figure~\ref{fig:wing_FRC}a shows the Poincar\'e section of the non-autonomous SSM for the near-resonant forcing frequency $\Omega = 29.8$ rad/s. The hyperbolic fixed points of the reduced dynamics~\eqref{eq:reduced_dyn_polar} on the SSM directly provide the stable (blue) and unstable (red) periodic orbits on the FRC for different values of forcing frequency $\Omega$.} On a macro-level, this wing example resembles
a cantilevered beam and we expect a hardening type response. Indeed,
the three FRCs at $\mathcal{O}(3),\mathcal{O}(5),$ and $\mathcal{O}(7)$
converge towards a hardening-type response, as shown in Figure~\ref{fig:wing_FRC}b.

\begin{figure}[H]
	\subfloat[]{\centering{}\includegraphics[width=0.54\linewidth]{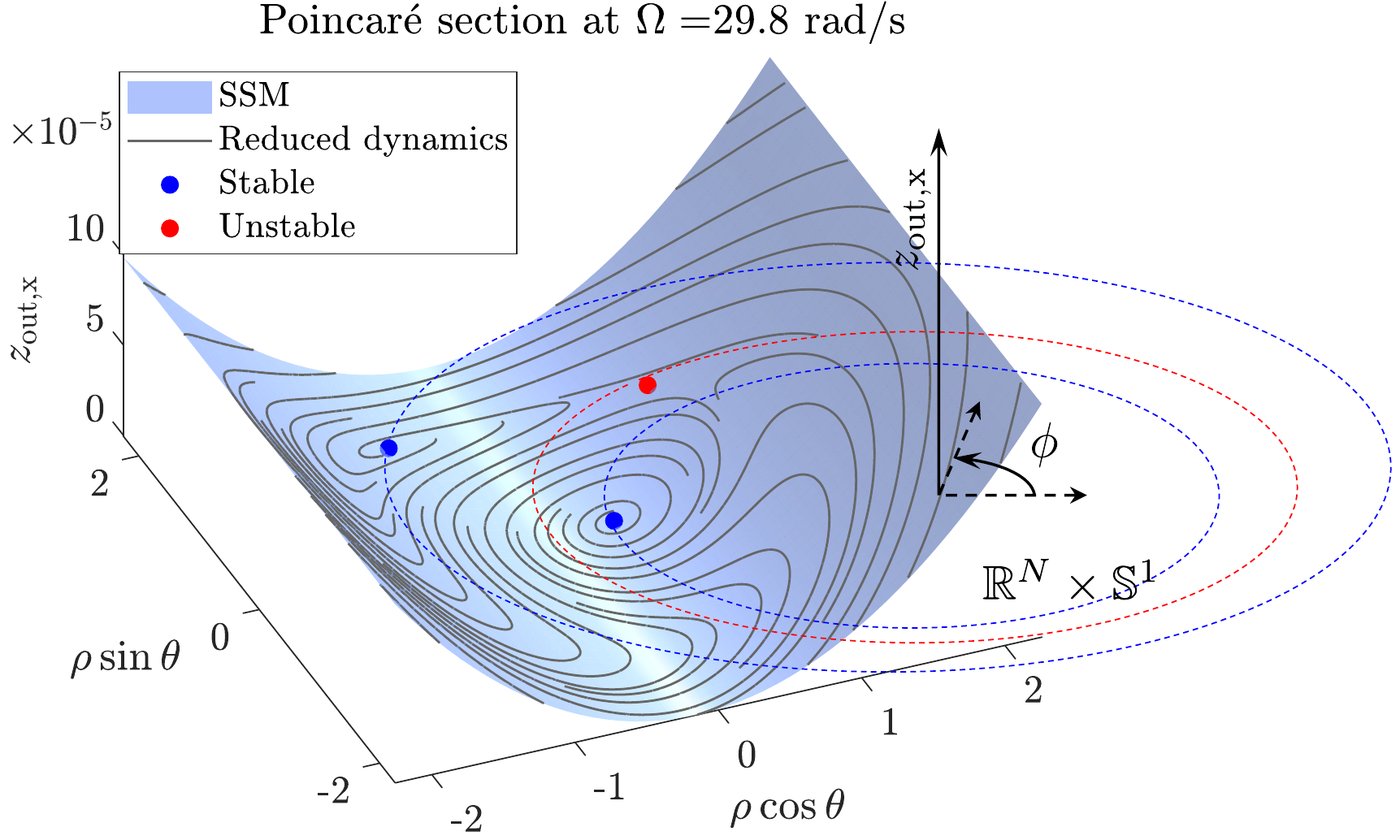}}\hfill{}\subfloat[]{\centering{}\includegraphics[width=0.44\linewidth]{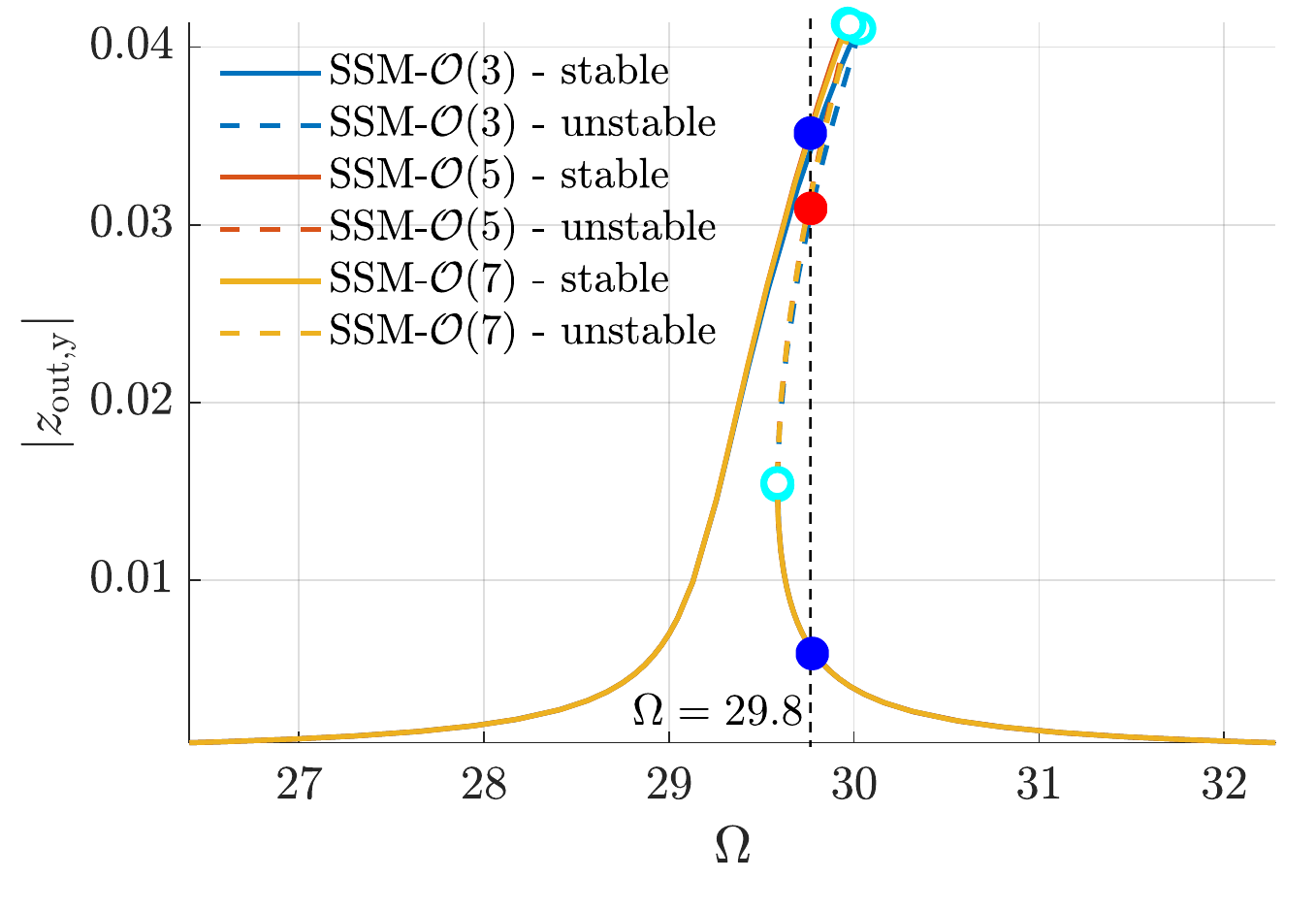}~}	
	%\centering{}\includegraphics[width=0.9\linewidth]{Figures/wingFRC}		
	\caption{ \small \label{fig:wing_FRC} (a) Poincar\'e section of the  non-autonomous SSM  of the aircraft wing structure with $ n= $133,920 degrees of freedom (see Figure~\ref{fig:wing}) computed around the first mode (eigenvalues~\eqref{eq:eig_wing})  for near-resonant forcing frequency $\Omega = 29.8$ rad/s. The reduced dynamics in polar coordinates $\rho$, $\theta$ is obtained by simulating the ROM~\eqref{eq:reduced_dyn_polar} (see eq.~\eqref{eq:redsimplot}); the fixed points in blue and red directly provide us the stable and unstable periodic orbits on the FRC (b) for different values of $\Omega$. FRCs obtained using local SSM computations at $\mathcal{O}(3),\mathcal{O}(5)$
		and $\mathcal{O}(7)$ converge towards a hardening response; plots (a) and (b) show the displacements in the $x$ and $y$-directions at the tip-node 1 shown in see Figure~\ref{fig:wing}b
		(see Table~\ref{tab:comp_wing} for the computational resources
		consumed).}
\end{figure}
Table~\ref{tab:comp_wing} depicts the computational resources consumed
in obtaining these three FRCs. The peaks in memory consumption reported
in Table~\ref{tab:comp_wing} occur during the composition of nonlinearity
(see eq.~\eqref{eq:FoW_i}). Note that these peaks are short-lived,
however, as the average memory consumption during all these computations
are much lower. We remark that in the context of finite-element applications,
these memory peaks can be significantly reduced by implementing the
nonlinearity composition at the element-level in contrast to the currently
performed implementation at the full system level. Once again, use
of parallel computing and distributed memory architectures would be
greatly beneficial in this context.

\begin{table}[H]
	\begin{centering}
		\begin{tabular}{|c|>{\centering}p{4cm}|>{\centering}p{2cm}>{\centering}p{2cm}|}
			\hline 
			SSM order & Computation time {[}hours:minutes:seconds{]}  & Peak memory consumption  & Average memory consumption\tabularnewline
			\hline 
			\hline 
			SSM-$\mathcal{O}(3)$ & 00:11:17 & $\approx$24 GB & $\approx$9 GB\tabularnewline
			\hline 
			SSM-$\mathcal{O}(5)$ & 00:35:47 & $\approx$33 GB & $\approx$10 GB\tabularnewline
			\hline 
			SSM-$\mathcal{O}(7)$ & 01:47:51 & $\approx$88 GB & $\approx$19 GB\tabularnewline
			\hline 
		\end{tabular}
		\par\end{centering}
	\caption{\label{tab:comp_wing}Computation time and memory requirements for
		obtaining the three FRCs depicted in Figure~\ref{fig:wing_FRC}. All
		computations were performed on MATLAB version 2019b installed on the
		ETH Z\"{u}rich Euler supercomputing cluster on a single node with $100,000$ MB ($\approx 100$ GB) of RAM.}
\end{table}

\section{Conclusions}

In this work, we have reformulated the parametrization method for local
approximations of invariant manifolds and their reduced dynamics in the context of high-dimensional nonlinear mechanics problems. In this class of problems, the classically used system diagonalization at the linear level is no longer feasible. Instead,
we have developed expressions that enable the computation of invariant
manifolds and their reduced dynamics in physical coordinates using only
the master modes associated with the invariant manifold. Hence, these computations facilitate mathematically rigorous nonlinear model reduction in very high-dimensional problems. A numerical implementation of the proposed computational methodology is available
in the open-source MATLAB package, SSMTool~2.0~\cite{SSMTool2.0}, which enables
the computation of invariant manifolds in finite element-based discretized problems via an
integrated finite element solver~\cite{FECode} and bifurcation analysis of the reduced dynamics on these invariant manifolds via its \textsc{coco}~\cite{Dankowicz} integration.

We have connected this computational methodology to several applications
of engineering significance, including the computation of parameter-dependent center manifolds; Lyapunov subcenter manifolds
(LSM) and their associated conservative backbone curves; and Spectral Submanifolds (SSM) and their associated forced response
curves (FRCs) in dissipative mechanical systems. We have also demonstrated
fast and reliable computations of FRCs via a normal form style
parametrization of SSMs in very large mechanical
structures, which has been a computationally intractable task for other available approaches.

While our examples focused on the applications of two-dimensional
SSMs, this automated computation procedure and its numerical implementation
\cite{SSMTool2.0} can treat higher-dimensional invariant manifolds
as well. { Specifically, the reduced dynamics on higher-dimensional
	SSMs can be used for the direct computation of FRCs
	in internally-resonant mechanical systems featuring energy transfer among multiple modes, as will be demonstrated
	in forthcoming publications (Li et al.~\cite{Li2021}, Li \& Haller~\cite{Li2021b})}.
Furthermore, in the non-autonomous setting, we have restricted our expressions
to the leading-order contributions from the forcing. Similar expressions,
however, can also be obtained for higher-order terms at the non-autonomous
level, which is relevant for the nonlinear analysis of parametrically
excited systems. These expressions and the related numerical implementation
are currently under development.

Finally, as we have noted, these computations will further benefit from parallelization since the invariance equations can be
solved independently for each monomial/Fourier multi-index { (see Remarks~\ref{rem:parallelization} and \ref{rem:parallelization2}, and Section~\ref{sec:wing})}. This development is currently underway and will be reported elsewhere.

\section*{Acknowledgements}
We are thankful to Mingwu Li for help in using \textsc{coco}, for his careful proof-reading of this manuscript and for providing valuable comments. We also thank Harry Dankowicz for helpful suggestions and pointing us to the atlas-$k$d algorithm in \textsc{coco}. 

\section*{Conflicts of interest}
The authors declare that they have no conflict of interest.

\section*{Funding}
No specific funding was received for this work.

\section*{Data availability}
The numerical implementation of algorithms and results discussed in this work are available in the form of an open-source MATLAB software at the following DOI: \url{http://doi.org/10.5281/zenodo.4614202}

\appendix
\section{Basic propositions}

\label{sec:Proof-of-Proposition1}
\begin{prop}
	\label{thm:1}Let $\mathbf{Q}\in\mathbb{C}^{M\times M}$ be any semisimple
	matrix with eigenvalues $\mu_{1},\dots,\mu_{M}$ (including repetitions)
	and corresponding left and right eigenvectors $\mathbf{p}_{1},\dots,\mathbf{p}_{M}\in\mathbb{C}^{M}$
	and $\mathbf{r}_{1},\dots,\mathbf{r}_{M}\in\mathbb{C}^{M}$. Then,
	for any $i\in\mathbb{N}$, the matrix 
	\begin{equation}
	\boldsymbol{\mathcal{Q}}_{i}:=\text{\ensuremath{\sum_{j=1}^{i}}\ensuremath{\ensuremath{\left(\mathbf{I}_{M}\right)^{\otimes j-1}\otimes\mathbf{Q}\otimes\left(\mathbf{I}_{M}\right)^{\otimes i-j}}}}\in\mathbb{C}^{M^{i}\times M^{i}}\label{eq:Q_i}
	\end{equation}
	is semisimple and its eigenvalues are
	\begin{equation}
	\mu_{\mathbf{\boldsymbol{\ell}}}=\mu_{\ell_{1}}+\dots+\mu_{\ell_{i}},\quad\boldsymbol{\ell}\in\Delta_{i,M}\label{eq:mu_l}
	\end{equation}
	with the left and right eigenvectors corresponding to any eigenvalue
	$\mu_{\mathbf{\boldsymbol{\ell}}}$ explicitly given as $\mathbf{p}_{\boldsymbol{\ell}}:=\mathbf{p}_{\ell_{1}}\otimes\dots\otimes\mathbf{p}_{\ell_{i}}$
	and $\mathbf{q}_{\boldsymbol{\ell}}:=\mathbf{q}_{\ell_{1}}\otimes\dots\otimes\mathbf{q}_{\ell_{i}}.$ 
\end{prop}
\begin{proof}
	The proof involves a straight-forward verification of the statement.
	The eigenvalues and the left and right eigenvectors of $\mathbf{Q}$
	satisfy
	\begin{align}
	\mathbf{Qp}_{j} & =\mu_{j}\mathbf{p}_{j},\quad j=1,\dots,M,\label{eq:righteig-1}\\
	\mathbf{q}_{j}^{\star}\mathbf{Q} & =\mu_{j}\mathbf{q}_{j}^{\star},\quad j=1,\dots,M,\label{eq:lefteig-1}
	\end{align}
	We first verify that $\mu_{\mathbf{\boldsymbol{\ell}}}$ (see eq.
	\eqref{eq:mu_l}) is an eigenvalue of $\boldsymbol{\mathcal{Q}}_{i}$
	(see eq.~\eqref{eq:Q_i}) with eigenvector $\mathbf{p}_{\boldsymbol{\ell}}:=\left(\mathbf{p}_{\ell_{1}}\otimes\dots\otimes\mathbf{p}_{\ell_{i}}\right)$,
	i.e.,
	\begin{align*}
	\boldsymbol{\mathcal{Q}}_{i}\mathbf{p}_{\boldsymbol{\ell}} & =\left[\ensuremath{\sum_{j=1}^{i}}\ensuremath{\ensuremath{\underbrace{\mathbf{I}_{m}\otimes\dots\otimes\overset{\overset{j-\mathrm{th\,position}}{\uparrow}}{\mathbf{Q}}\otimes\dots\otimes\mathbf{I}_{m}}_{i-\mathrm{terms}}}}\right]\left[\mathbf{p}_{\ell_{1}}\otimes\dots\otimes\mathbf{p}_{\ell_{i}}\right]\\
	& =\ensuremath{\sum_{j=1}^{i}}\mathbf{p}_{\ell_{1}}\otimes\dots\otimes\mathbf{Q}\mathbf{p}_{\ell_{j}}\otimes\dots\otimes\mathbf{p}_{\ell_{i}}\\
	& =\ensuremath{\sum_{j=1}^{i}}\mathbf{p}_{\ell_{1}}\otimes\dots\otimes\mu_{\ell_{j}}\mathbf{p}_{\ell_{j}}\otimes\dots\otimes\mathbf{p}_{\ell_{i}}\\
	& =\ensuremath{\sum_{j=1}^{i}\mu_{\ell_{j}}}\left(\mathbf{p}_{\ell_{1}}\otimes\dots\otimes\mathbf{p}_{\ell_{i}}\right)\\
	& =\mu_{\mathbf{\boldsymbol{\ell}}}\mathbf{p}_{\boldsymbol{\ell}},
	\end{align*}
	where we have used basic properties of the Kronecker product (see,
	e.g., Van Loan~\cite{vanLoan2000}) along with with eq.~\eqref{eq:lefteig-1}.
	Similarly, using eq.~\eqref{eq:righteig-1}, we can verify that the
	left eigenvector corresponding to the eigenvalue $\mu_{\mathbf{\boldsymbol{\ell}}}$
	is given by $\mathbf{q}_{\boldsymbol{\ell}}:=\mathbf{q}_{\ell_{1}}\otimes\dots\otimes\mathbf{q}_{\ell_{i}}$,
	i.e.,
	\[
	\mathbf{q}_{\boldsymbol{\ell}}^{\star}\boldsymbol{\mathcal{Q}}_{i}=\mu_{\mathbf{\boldsymbol{\ell}}}\mathbf{q}_{\boldsymbol{\ell}}^{\star}.
	\]
	We note that since $\mathbf{p}_{1},\dots,\mathbf{p}_{M}$ are linearly
	independent since $\mathbf{Q}$ is semisimple. Then, using the fact
	that for any two linearly independent vectors $\mathbf{a}$ and $\mathbf{b}$,
	the vectors $\mathbf{a}\otimes\mathbf{b}$ and $\mathbf{b}\otimes\mathbf{a}$
	are also linearly independent, we conclude that all $M^{i}$ eigenvectors
	$\mathbf{p}_{\boldsymbol{\ell}}$ with $\boldsymbol{\ell}\in\Delta_{i}$
	of $\boldsymbol{\mathcal{Q}}_{i}$ are linearly independent. Hence,
	$\boldsymbol{\mathcal{Q}}_{i}$ is semi-simple.
\end{proof}
\begin{prop}
	\label{thm:2}Let $\mathbf{\lambda}\in\mathbb{\mathbb{C}}$ be a generalized
	eigenvalue of a matrix pair $\mathbf{A},\mathbf{B}\in\mathbb{C}^{N\times N}$
	with the corresponding left and right eigenvectors $\mathbf{u},\mathbf{v}\in\mathbb{C}^{N}$,
	and $\mu\in\mathbb{C}$ be a generalized eigenvalue of a matrix pair
	$\mathbf{C},\mathbf{D}\in\mathbb{C}^{M\times M}$ with the corresponding
	left and right eigenvectors $\mathbf{e},\mathbf{f}\in\mathbb{C}^{M}$
	for any $M,N\in\mathbb{N}$. Then $\lambda\mu$ is a generalized eigenvalue
	for the matrix pair $\left(\mathbf{C}\otimes\mathbf{A}\right),\left(\mathbf{D}\otimes\mathbf{B}\right)\in\mathbb{C}^{NM\times NM}$
	with the corresponding left and right eigenvectors $(\mathbf{e}\otimes\mathbf{u}),(\mathbf{f}\otimes\mathbf{v})\in\mathbb{C}^{MN}$.
	Furthermore, if $\lambda\mu=1$, then the matrix 
	\begin{equation}
	\mathbf{E}:=\mathbf{D}\otimes\mathbf{B}-\mathbf{C}\otimes\mathbf{A}\in\mathbb{C}^{NM\times NM}\label{eq:Edef}
	\end{equation}
	is singular with $(\mathbf{e}\otimes\mathbf{u})\in\ker\left(\mathbf{E}^{\star}\right)$
	and $(\mathbf{f}\otimes\mathbf{v})\in\ker\left(\mathbf{E}\right)$.
\end{prop}
\begin{proof}
	Since $\lambda$ is a generalized eigenvalue of the matrix pair $\mathbf{A},\mathbf{B}$
	with $\mathbf{w},\mathbf{v}$ being the corresponding left and right
	eigenvectors, we have 
	\begin{align}
	\mathbf{Av} & =\lambda\mathbf{Bv},\label{eq:righteigAB}\\
	\mathbf{u^{\star}}\mathbf{A} & =\lambda\mathbf{u^{\star}}\mathbf{B}.\label{eq:lefteigAB}
	\end{align}
	Similarly, for the matrix pair $\mathbf{C},\mathbf{D}$ with eigenvalue
	$\mu$ and $\mathbf{e},\mathbf{f}$ being the corresponding left and
	right eigenvectors, we have 
	\begin{align}
	\mathbf{Cf} & =\mu\mathbf{D}\mathbf{f},\label{eq:righteigCD}\\
	\mathbf{e^{\star}}\mathbf{C} & =\mu\mathbf{e^{\star}}\mathbf{D}.\label{eq:lefteigCD}
	\end{align}
	We verify that $\lambda\mu$ is a generalized eigenvalue of the matrix
	pair $\left(\mathbf{C}\otimes\mathbf{A}\right),\left(\mathbf{D}\otimes\mathbf{B}\right)$
	with eigenvector $(\mathbf{f}\otimes\mathbf{v})$, i.e.,
	\begin{align}
	\left(\mathbf{C}\otimes\mathbf{A}\right)(\mathbf{f}\otimes\mathbf{v}) & =\left(\mathbf{C}\mathbf{f}\right)\otimes\left(\mathbf{A}\mathbf{v}\right)\nonumber \\
	& =\left(\mu\mathbf{D}\mathbf{f}\right)\otimes\left(\lambda\mathbf{Bv}\right)\nonumber \\
	& =\lambda\mu\left(\mathbf{D}\mathbf{f}\right)\otimes\left(\mathbf{Bv}\right)\nonumber \\
	& =\lambda\mu\left(\mathbf{D}\otimes\mathbf{B}\right)(\mathbf{f}\otimes\mathbf{v}),\label{eq:righteigABCD}
	\end{align}
	where we have used eqs.~\eqref{eq:righteigAB},~\eqref{eq:righteigCD}
	along with basic properties of the Kronecker product of matrices (see,
	e.g., Van Loan~\cite{vanLoan2000}). Similarly, using eqs.~\eqref{eq:lefteigAB},
	\eqref{eq:lefteigCD}, we can show that 
	\begin{equation}
	(\mathbf{e}\otimes\mathbf{u})^{\star}\left(\mathbf{C}\otimes\mathbf{A}\right)=\lambda\mu(\mathbf{e}\otimes\mathbf{u})^{\star}\left(\mathbf{D}\otimes\mathbf{B}\right),\label{eq:lefteigABCD}
	\end{equation}
	which proves that $(\mathbf{e}\otimes\mathbf{u})$ is the generalized
	left-eigenvector for the generalized eigenvalue $\lambda$ for the
	generalize first part of the proposition. 
	
	Finally, substituting $\lambda\mu=1$ in eqs.~\eqref{eq:righteigABCD},
	\eqref{eq:righteigABCD}, we obtain 
	\begin{align*}
	\left[\left(\mathbf{C}\otimes\mathbf{A}\right)-\left(\mathbf{D}\otimes\mathbf{B}\right)\right](\mathbf{f}\otimes\mathbf{v}) & =\mathbf{0},\\
	(\mathbf{e}\otimes\mathbf{u})^{\star}\left[\left(\mathbf{C}\otimes\mathbf{A}\right)-\left(\mathbf{D}\otimes\mathbf{B}\right)\right] & =\mathbf{0},
	\end{align*}
	which proves that the matrix $\mathbf{E}$ (see definition~\eqref{eq:Edef})
	is singular and the vectors $(\mathbf{f}\otimes\mathbf{v})$ and $(\mathbf{e}\otimes\mathbf{u})$
	belong to its left and right kernels.
\end{proof}
\section{Proof of Lemma~\ref{lemma:cons_bb}}

\label{sec:Proof-of-Lemma}

In polar coordinates $\mathbf{p}=\mathbf{P}(\rho,\theta):=\rho\left[\begin{array}{c}
e^{\mathrm{i}\theta}\\
e^{-\mathrm{i}\theta}
\end{array}\right]$, the reduced dynamics~\eqref{eq:red_2D_cons} is given as 
\begin{equation}
\left[\begin{array}{c}
\dot{\rho}\\
\dot{\theta}
\end{array}\right]=\left[\begin{array}{c}
a(\rho)\\
\omega(\rho)
\end{array}\right]=\left[\begin{array}{c}
0\\
\omega_{m}
\end{array}\right]+\sum_{\ell\in\mathbb{N}}\left[\begin{array}{c}
\mathrm{Re}(\gamma_{\ell})\rho^{2\ell+1}\\
\mathrm{Im}(\gamma_{\ell})\rho^{2\ell}
\end{array}\right].\label{eq:polar_LSM-1}
\end{equation}
Now, since the equation for $\dot{\rho}$ is a scalar ODE decoupled
from the phase $\theta$, the only possible steady states are fixed
points $\dot{\rho}=0$, i.e., $\rho=\rho(0)=$constant. Therefore,
$\dot{\theta}=\omega(\rho)$ represents a constant angular frequency
depending on the constant steady-state amplitude $\rho$. Since the
LSM is an analytic manifold, the function $a,\omega$ describing reduced
dynamics on the LSM are also analytic. Note that since the LSM is
filled with periodic orbits in an open neighborhood of the origin,
$a(\rho)$ must be identically zero because any non-trivial polynomial
expression for $a$ would result in isolated periodic orbits on the
LSM. Thus, we deduce that our computational procedure must result
in 
\[
\mathrm{Re}(\gamma_{\ell})=0\quad\forall\ell\in\mathbb{N}.
\]
Finally, as the LSM is foliated with periodic orbits in an open neighborhood
of the origin, the equation $\dot{\theta}=\omega(\rho)$ expresses
the frequency of oscillation as a function any given constant amplitude~$\rho$ for the periodic orbits on the LSM. Hence the conservative
backbone around the $m^{\mathrm{th}}$ mode is given by the relation
\eqref{eq:cons_backbone}.

\section{Proof of Lemma~\ref{thm:frc}}

\label{sec:Proof-of-Theorem}

Using polar coordinates $\mathbf{p}=\mathbf{P}(\rho,\theta):=\rho\left[\begin{array}{c}
e^{\mathrm{i}\theta}\\
e^{-\mathrm{i}\theta}
\end{array}\right]$, we rewrite the first equation in the system~\eqref{eq:normal_2D}
as
\begin{equation}
\dot{p}_{1}=\dot{\rho}e^{\mathrm{i}\theta}+\mathrm{i}\dot{\theta}\rho e^{\mathrm{i}\theta}=\lambda\rho e^{\mathrm{i}\theta}+\sum_{\ell\in\mathbb{N}}\gamma_{\ell}\rho^{2\ell+1}e^{\mathrm{i}\theta}+fe^{\mathrm{i}\eta\phi},\label{eq:polar1-1}
\end{equation}
where 
\[
f:=\epsilon\mathbf{u}_{j}^{\star}\mathbf{F}_{\eta}^{ext}.
\]
Dividing eq.~\eqref{eq:polar1-1} by $e^{\mathrm{i}\theta}$ and introducing
the phase shift $\psi=\theta-\eta\phi$, we obtain
\begin{equation}
\dot{\rho}+\mathrm{i}\left(\dot{\psi}+\eta\Omega\right)\rho=\lambda\rho+\sum_{\ell\in\mathbb{N}}\gamma_{\ell}\rho^{2\ell+1}+\epsilon fe^{-\mathrm{i}\psi}.\label{eq:polar2-1}
\end{equation}
Comparing the real and imaginary parts in eq.~\eqref{eq:polar2-1},
we obtain the polar reduced dynamics given in eq.~\eqref{eq:reduced_dyn_polar},
which concludes the proof of (i).

The fixed points of system~\eqref{eq:reduced_dyn_polar} are obtained
by equating its right-hand-side to zero as 
\begin{equation}
\mathbf{r}(\rho,\psi,\Omega)=\mathbf{0}.\label{eq:red_dyn_RHS-2-1}
\end{equation}
Any such fixed point represent a periodic orbit for the reduced system
\eqref{eq:reduced_dyn_polar} with constant polar radius $\rho$ and
constant phase difference $\psi$ with respect to the cyclic variable
$\eta\phi$, which has the angular frequency $\eta\Omega$. Hence,
we obtain a 1-dimensional submanifold of zeros upon solving~\eqref{eq:red_dyn_RHS-2-1},
whose projection in the $(\rho,\Omega)$ provides us the FRC. Eliminating $\psi$ from eq.~\eqref{eq:red_dyn_RHS-2-1}, we
obtain the fixed points as the set of points $\left(\rho,\Omega\right)$
that satisfy the equation 
\begin{equation}
\left[a(\rho)\right]^{2}+\left[b(\rho,\Omega)\right]^{2}=|f|^{2},\label{eq:frc}
\end{equation}
which proves statement (ii). 

Solving the two equations~\eqref{eq:red_dyn_RHS-2-1} for $\cos\psi$
and $\sin\psi$, we obtain 
\begin{align*}
\cos\psi & =-\frac{\left(a(\rho)\mathrm{Re}(f)+b(\rho,\Omega)\mathrm{Im}(f)\right)}{|f|^{2}},\\
\sin\psi & =\frac{\left(b(\rho,\Omega)\mathrm{Re}(f)-a(\rho)\mathrm{Im}(f)\right)}{|f|^{2}},
\end{align*}
which provides the phase shift as 
\[
\psi=\arctan\left(\frac{b(\rho,\Omega)\mathrm{Re}(f)-a(\rho)\mathrm{Im}(f)}{-a(\rho)\mathrm{Re}(f)-b(\rho,\Omega)\mathrm{Im}(f)}\right),
\]
and hence proves statement (iii).

Finally, we rewrite the reduced dynamics~\eqref{eq:reduced_dyn_polar}
in its standard form as
\begin{equation}
\left[\begin{array}{c}
\dot{\rho}\\
\dot{\psi}
\end{array}\right]=\left[\begin{array}{c}
a(\rho)\\
\frac{b(\rho,\Omega)}{\rho}
\end{array}\right]+\left[\begin{array}{cc}
1 & 0\\
0 & \frac{1}{\rho}
\end{array}\right]\left[\begin{array}{cc}
\cos\psi & \sin\psi\\
-\sin\psi & \cos\psi
\end{array}\right]\left[\begin{array}{c}
\mathrm{Re}\left(f\right)\\
\mathrm{Im}\left(f\right)
\end{array}\right].\label{eq:reduced_standard}
\end{equation}
The Jacobian of the right-hand side of eq.~\eqref{eq:reduced_standard}
evaluated at the fixed point $(\rho,\psi)$ is given by~\eqref{eq:Jacobian}
and its eigenvalues can be used to conclude the stability of any hyperbolic
fixed point via linearized stability analysis, which proves statement
(iv).

\end{document}